\documentclass[12pt]{article}
\usepackage[pass]{geometry}

\usepackage{graphicx} 
\usepackage{subfigure} 
\usepackage{lscape}	
\usepackage{pst-plot, pstricks} 
\usepackage{fancybox,amssymb,color} 
\usepackage[english]{babel}
\usepackage{setspace} 
\usepackage{epstopdf}
\usepackage[capposition=top]{floatrow} 
\usepackage{float}
\usepackage{enumitem}
\usepackage{longtable}
\usepackage{tabularx,calc}
\usepackage{multirow}
\usepackage{booktabs}
\usepackage{setspace}
\usepackage{lscape}
\usepackage{mathtools}
\usepackage{rotating} 
\usepackage{adjustbox}
\usepackage{blkarray}
	
\usepackage[authoryear]{natbib}
\usepackage{caption} 

\usepackage{amsmath}
\usepackage{amsthm} 
\usepackage{amssymb}
\usepackage{amsbsy}
\usepackage{bm}

\DeclareMathOperator*{\argmin}{arg\,min}

\usepackage{changes}
\definechangesauthor[name={Sascha}, color=green]{Sascha}

\newtheorem{proposition}{Proposition} 
\newtheorem{lemma}{Lemma} 
\newtheorem{corollary}{Corollary} 
\newtheorem{assumption}{Assumption} 
 
\usepackage[dvipsnames]{xcolor}

\let\b\mathbf

\let\mc\mathcal
\let\bs\boldsymbol

\newcommand{\vect}[1]{\mathrm{vec}\left( {#1} \right)}

\newcommand{\R}{\mathbb{R}}	
\newcommand{\N}{\mathbb{N}}	

\newcommand{\E}{\mathbb{E}}	
\newcommand{\var}{\text{Var}}	
\newcommand{\B}{\mathcal{B}}	

\DeclareMathOperator*{\diag}{diag}		
\DeclareMathOperator*{\rank}{rank}

\newtheorem{assumptionalt}{Assumption}[assumption]
\newenvironment{assumptionp}[1]{
  \renewcommand\theassumptionalt{(#1)}
  \assumptionalt
}{\endassumptionalt}

\usepackage[most]{tcolorbox}
\newtcbtheorem{Notes}{\bfseries Notes}{enhanced,drop shadow={black!50!white},
  coltitle=black,
  top=0.3in,
  attach boxed title to top left=
  {xshift=1.5em,yshift=-\tcboxedtitleheight/2},
  boxed title style={size=small,colback=Apricot}
}{notes}
\newtcolorbox[auto counter]{notes}[1][]{title={\bfseries Notes~\thetcbcounter},enhanced,drop shadow={black!50!white},
  coltitle=black,
  top=0.3in,
  attach boxed title to top left=
  {xshift=1.5em,yshift=-\tcboxedtitleheight/2},
  boxed title style={size=small,colback=Apricot},#1}

\newcommand*\dd{\mathop{}\!\mathrm{d}} 

\title{\vspace{-8ex}
Testing strong exogeneity in proxy SVARs with internal synthetic proxies
}
\author{Sascha  A. Keweloh\thanks{TU Dortmund, Department of Economics, D-44221 Dortmund, Germany. Sascha Keweloh gratefully acknowledges financial support by TRR 391 “Spatio-temporal Statistics for the Transition of Energy and Transport” (520388526)  by the Deutsche Forschungsgemeinschaft (DFG, German Research Foundation).} 
 \and
 Shu Wang\thanks{University of G\"ottingen, Chair of Econometrics. European University Viadrina Frankfurt, Department of Statistics, Gro{\ss}e Scharrnstra{\ss}e 59, D-15230 Frankfurt (Oder), Germany. Shu Wang gratefully acknowledges financial support by the DFG  (HE 2188/17-2). We thank Martin Bruns, Taro Fujita, Raffaella Giacomini, Helmut Herwartz, Christian Matthes, Elena Pesavento, Christoph Rothe, Carsten Trenkler, Endong Wang, and organizers and participants at various research seminars (University of Mannheim, University of G\"ottingen) and conferences (EEA-ESEM 2026, AMES-China 2026, Bundesbank Hannover 2026, 4th edition of the UEA ECO Time Series Workshop) for helpful discussions and comments.}}
\date{\footnotesize{\today}}

\begin{document}
	\clearpage\maketitle
	\thispagestyle{empty}
    
    \vspace{-1cm}
	\begin{abstract}
Proxy SVARs often rely on a single credible proxy, leaving proxy exogeneity untestable. 
Narrative arguments used to justify proxy validity imply a stronger form of exogeneity. We formalize the strong exogeneity assumption and propose a corresponding test based on internal synthetic proxies. 
The approach uses the external proxy to locate the target-shock direction
in the reduced-form innovations and synthesize the candidate shock into an additional internal synthetic proxies, which allows to test the joint exogeneity of the original and internal synthetic proxy.
We derive the large-sample theory for the ISP test and show that the identification strength is governed by a direct measure of the target shock's non-Gaussianity.
We study weak identification under asymptotic Gaussianity and develop an identification-robust version of the test. Monte Carlo experiments demonstrate the favorable power of the test when the proxy has a low signal-to-noise ratio. Applications to monetary policy and oil supply news shocks find no evidence against the exogeneity of the benchmark proxies.
\end{abstract}
	
	\vspace{0.2cm} \noindent   
    {\it Keywords:} Structural vector autoregression, proxy SVAR, exogeneity test \\
    \noindent
	{\it JEL Codes:} C12, C32, C51

\newpage 
\section{Introduction}

Proxy variables (also known as external instruments) are widely used to identify structural shocks in SVARs, with prominent applications to monetary policy (\citealt{gertler2015monetary}, \citealt{caldara2019monetary}, \citealt{jarocinski2020deconstructing}, and \citealt{miranda2021transmission}), fiscal policy (\citealt{mertens2014reconciliation} and \citealt{caldara2017analytics}), uncertainty (\citealt{angelini2019exogenous}), and oil markets (\citealt{kanzig2021macroeconomic}). Identification requires a proxy that is relevant for the target shock and exogenous (i.e., orthogonal) to the remaining structural shocks. 
With a single proxy, the exogeneity assumption leaves the impact of the target shock just-identified, leaving no over-identifying restrictions to test exogeneity.
The credibility of the identifying assumption must instead rest largely on the economic narrative underlying the construction of the proxy. \citet{BK2024} observe that this narrative often has stronger implications than the conventional notion of exogeneity that involves only uncorrelatedness. Instead, the economic reasoning typically entails that non-target shocks are not predictable in their locations given the proxy variable, which is stronger than mere uncorrelatedness and can create testable restrictions.

A proxy $z_t$ for a target shock $\varepsilon_{1t}$ is commonly 
viewed as a noisy measurement of $\varepsilon_{1t}$, $z_t = f(\varepsilon_{1t}, v_t)$ (with $v_t$ being the noise). 
\citet{BK2024} argue for the notion of strong proxy exogeneity, namely that a valid proxy should not contain information on the conditional mean of non-target shocks $\bs\varepsilon_{2t}$, i.e., $\E[\bs\varepsilon_{2t}|z_t] = \b 0$. 
Under the notion of strong proxy exogeneity, measurable transformations (e.g., power functions) of the proxy remain exogenous and can be used as \textit{synthetic proxies} (SPs). Combining the original proxy with such transformations produces over-identifying restrictions.
Importantly, strong proxy exogeneity $\E[\bs\varepsilon_{2t}|z_t] = \b 0$ is stronger than the exogeneity assumption $\E[\bs\varepsilon_{2t}z_t] = \b 0$, and the latter suffices for just-identifying the shock. Nevertheless, strong proxy exogeneity offers an attractive way for the researcher to assess her underlying narrative arguments when no additional external information is available.

Despite being intuitive and easy to implement, the SP-based test suffers from an important drawback that limits its practical usefulness.
Higher-order transformations magnify the noise together with the signal,
causing the over-identifying moment restrictions to exhibit a substantially higher sample variation.
As a result,  tests with SPs can have very poor power to detect violations of the null.
The power problem becomes particularly severe when the external proxy has a low signal-to-noise ratio, or in general, is only weakly correlated with the target shock.
In this case, the additional moments remain valid but are measured too imprecisely to provide an informative test, even in relatively large samples.

The main contribution of this paper is to develop the approach of \emph{internal synthetic proxies} (ISPs).
We use the proxy $z_t$ to locate a candidate target shock in the space spanned by the reduced-form innovations.
Under exogeneity, $z_t$ identifies the target shock, and we can apply similar transformations to the identified shock to obtain an ISP.
Such an internal proxy remains exogenous and is free from the measurement error $v_t$ idiosyncratic to the proxy.
Combining the original $z_t$ with the ISP yields an over-identified system and allows us to test the joint exogeneity.

A second contribution is to develop a weak identification theory for the proposed test. 
The analysis reveals that the test draws its identifying information from non-Gaussian features of the target shock.
Specifically, the local Jacobian of the ISP moments is governed by a sequence of direct measures of non-Gaussianity with respect to the selected moments. These measures correspond to the deviations from the Stein identity for standard Gaussian random variables \citep[see,][Lemma 1]{Stein1981}, which we call Stein deviations. The identification content of a quadratic ISP is governed by the skewness of $\varepsilon_{1t}$, a cubic ISP by its excess kurtosis, and higher-order ISPs by departures from the corresponding Gaussian moment. 
Hence, weak proxy (i.e., local-to-zero relevance) does not imply weak ISP identification, as the target shock can continue to supply identifying information internally through a nonzero selected Stein deviation. 
This has also been observed in the recent contribution of \cite{Junare2026}, who shows that with sufficient non-Gaussianity standard asymptotics still apply under weak proxy. 
Local weak identification arises when \textit{both} the relevance of the original proxy vanishes \textit{and} the distribution of the target shock is asymptotically indistinguishable from the Gaussianity with respect to the moment selected. 

Hence, the ISP approach combines identification by an external proxy with statistical information from non-Gaussianity. The proxy selects and economically labels a potential target-shock direction, while selected higher-order moments are used to assess that direction. 
This characterization also places the ISP approach at the intersection of the proxy SVAR literature and the growing literature on statistical identification through non-Gaussianity and higher-order moments \citep[see, e.g.,][]{gourieroux2017statistical, lanne2017identification, keweloh2021generalized}.\footnote{Notable contributions to this literature also include \cite{matteson2017independent, guay2021identification,mesters2022non,amengual2022moment,
karamysheva2022we,lanne2022statistical,lanne2022identifying,herwartz2023srhl,
keweloh2023monetary,drautzburg2023refining,fiorentini2024identification,hafner2025kml,keweloh2025higher,keweloh2025uncertain, Junare2026}.}
The latter literature has shown that exploiting higher-order cross-moments can be useful for shock identification and has so far yielded fruitful results in explaining the transmission of monetary policy shocks, different oil price shocks, etc.
At the same time, ISPs do not require mutual independence of the shocks, so that dependencies beyond their conditional mean (e.g., common time-varying volatility) are generally permitted.
Additionally, the approach imposes no assumption on the non-target shocks other than that they must remain mutually orthogonal.
Our weak-identification theory is related to \citet{LeeMesters2024} and \citet{HoeschLeeMesters2024}, who study weak identification under local-to-Gaussian sequences in models identified through independence and non-Gaussianity (i.e., local-to-Gaussianity asymptotics). 
We provide a directly interpretable characterization of local-to-Gaussianity as a measure of identification strength, which pins down the intuitive moment-based distributional feature that supplies identification and links it precisely to the Jacobian of the population moments. 
Going beyond local weak identification, we suggest a robust test, which is consistent against fixed alternatives and has controlled size under local weak identification, albeit conservative. 

The paper contributes to a growing literature on tests of proxy exogeneity in SVARs based on heteroskedasticity and volatility breaks \citep{schlaak2023monetary,angelini2025invalid}, sign restrictions \citep{braun2023identification}, and non-Gaussianity  \citep{keweloh2023monetary,keweloh2025estimating}. Simulations show that the ISP test provides substantial power gains when the proxy is noisy, while the relative performance depends on the non-Gaussian features of the target shock and on the precision with which the higher-order moments can be estimated. Additional experiments illustrate the behavior of the identification-robust tests and show that omitted lags or nonlinear dynamics can produce over-rejections for a proxy that is exogenous with respect to the shocks of the correctly specified model. We apply both tests to high-frequency proxies for monetary policy shocks \citep{gertler2015monetary} and oil supply news shocks \citep{kanzig2021macroeconomic}. In both applications, we find no evidence against the proxies. These findings speak to whether the benchmark short-rate surprise can be used directly or whether announcement components should first be separated \citep{jarocinski2020deconstructing,degasperi2025}. 

The remainder of the paper is organized as follows. Section~\ref{sec:intuition} provides an intuition of testing with SPs and ISPs. Section~\ref{sec:theory} discusses theoretical properties of the test under identification, while Section~\ref{sec:SP_weak} addresses weak identification. Section~\ref{sec:MC} reports the MC evidence. Section~\ref{sec:app} contains the empirical applications, and  Section~\ref{sec:fazit} concludes. The Appendix collects proofs and additional theoretical, MC and empirical results. Throughout, we use the symbol ``$\Rightarrow$'' for weak convergence, the shorthand $\E[X\mid Y]$ and $\E[X\mid Y,Z]$ for $\E[X\mid\bs\sigma(Y)]$ and $\E[X\mid\bs\sigma(Y,Z)]$, respectively, where conditional expectations are understood to be defined almost surely (a.s.).

\section{The intuition of testing strong exogeneity}\label{sec:intuition}
For endogenous variables $y_t \in \R^n$, we are interested in a finite-order SVAR($p$)  
\begin{align} 
        \label{eq:SVAR}
y_t=\nu+ \sum_{i=1}^{p} A_{i} y_{t-i} +u_t \quad \text{and } \quad  u_t = B_0 \varepsilon_t, 
\end{align}
where $u_t$ contains the reduced-form shocks and $\varepsilon_t = (\varepsilon_{1t}, \ldots, \varepsilon_{nt})'$ the structural shocks, which have zero mean and a diagonal covariance matrix $\E[\varepsilon_t \varepsilon_t'] = \Sigma_\varepsilon = \diag(\sigma_{1}^2, \ldots, \sigma_{n}^2)$. The structural shocks are mapped to the reduced-form model through the non-singular matrix $B_0 \in \R^{n \times n}$, whose diagonal entries are normalized to one and $\E[u_t u_t'] = \Sigma_u = B_0 \Sigma_\varepsilon B_0'$. Let $\varepsilon_{1t}$ denote the target shock, $\bs{\varepsilon_{2t}} \coloneqq (\varepsilon_{2t}, \ldots, \varepsilon_{nt})'$ the non-target shocks, and let $\beta_0 \in \R^{n-1}$ with $b_0 = b(\beta_0) \coloneq (1, \beta_0')'$ denote the first column of $B_0$, which contains the structural impact multipliers of the target shock $\varepsilon_{1t}$. 

Given the estimated reduced form, a proxy SVAR estimates  $b(\beta)$  using a  proxy  $z_t \in \R$. A valid proxy is relevant and exogenous, see \cite{mertens2013dynamic}.

\begin{assumptionp}{A1}\label{ass:exogeneity}
The proxy $z_t$ satisfies
\[
\E[\varepsilon_{1t} z_t] \neq 0 \quad \text{(relevance)}
\qquad
\text{and}
\qquad
\E[\bs{\varepsilon_{2t}}  z_t] = \bs{0} \quad \text{(exogeneity)}. 
\]
\end{assumptionp}
 Assumption~\ref{ass:exogeneity} is sufficient to identify the simultaneous impact of the target shock  by  
\[
\beta_0 = \frac{\E[z_t\bs u_{2t}]}{\E[z_t u_{1t}]},
\qquad
\bs u_{2t}:=\mc J u_t,\quad \mc J=(\iota_2,\ldots,\iota_n)' = 
  (\bs 0_{(n-1) \times 1} \quad I_{n-1}),
\]
where $\iota_i$ is the $i$th canonical basis in $\R^n$. 
However, with a single proxy, the exogeneity assumption is not testable from the data.

\cite{BK2024} argue that exogeneity is usually defended through economic narratives of how the proxy was constructed. 
For example, changes in asset prices within a narrow window around policy  announcements are used as proxies for monetary policy or oil market shocks \citep[see, e.g.,][]{gertler2015monetary,nakamura2018high,jarocinski2020deconstructing,miranda2021transmission,kanzig2021macroeconomic}. The narrative argument is that within the announcement window, the proxy is affected by the policy news, while other non-target shocks are unlikely to arrive within the event window. 
Therefore, the proxy shall be interpreted as a noisy measure of the target shock, without being driven by non-target shocks.

This argument establishes that the proxy is an unknown function of the target shock and noise, $z_t=f(\varepsilon_{1t},v_t)$. However, this argument does not by itself imply even conventional proxy exogeneity  \(\E[z_t\varepsilon_{2t}]=\bs 0\).
In fact, \cite{kanzig2021macroeconomic} explicitly acknowledged that to be able to interpret the surprises as the intended proxy, it cannot contain any information about other factors such as oil demand or economic activity.
In general, the bridge between the narrative argument establishing $z_t=f(\varepsilon_{1t},v_t)$ and proxy exogeneity  $\E[z_t\varepsilon_{2t}]=\bs 0$  is the assumption that the target shock and noise contain no information on the expected value of the non-target shocks.\footnote{
In general, the assumption of uncorrelated shocks $\E[\varepsilon_{1t} \varepsilon_{2t}]=0$ is not sufficient to conclude exogeneity $\E[z_t\varepsilon_{2t}]=\bs 0$ for an arbitrary functional form $z_t=f(\varepsilon_{1t},v_t)$. For example, the shocks $\varepsilon_{1t} \sim N(0,1)$ and $\varepsilon_{2t}=\varepsilon_{1t}^2-1$ are uncorrelated and the proxy $z_t=f(\varepsilon_{1t})=\varepsilon_{1t}+\varepsilon_{1t}^2$ is only a function of the target shock, but $\E[z_t \varepsilon_{2t}]=2 \neq 0$.
Following Lemma~\ref{lem:conditional_expectation}, it can be shown that $\E[z_t\bs\varepsilon_{2t}]=\bs 0$ holds for arbitrary Borel-measurable $f$ if and only if $\E[\bs\varepsilon_{2t}\mid \varepsilon_{1t},v_t]=\bs 0$. 
The inability to conclude proxy exogeneity from $\E[\varepsilon_{1t} \varepsilon_{2t}]=0$ and $z_t=f(\varepsilon_{1t},v_t)$ established by the narrative argument comes from the fact that the function $f(.)$ is unknown. 
If the narrative argument additionally justified a specific functional form, proxy exogeneity can follow from a weaker assumption than strong exogeneity, e.g., it follows from uncorrelated shocks for a linear function $f(.)$. 
} 
The strong exogeneity test proposed in the next section is built on this assumption.

\begin{assumptionp}{A2}\label{ass:strong_exogeneity}
The proxy  $z_t=f(\varepsilon_{1t},v_t)$ with $f: \R \times \R \mapsto \R$ measurable and $v_t$ uncorrelated with $\varepsilon_{1t}$ satisfies
\[
\E[\varepsilon_{1t} z_t] \neq 0 \quad \text{(relevance)}
\qquad
\text{and}
\qquad
\E[\bs\varepsilon_{2t}\mid \varepsilon_{1t},v_t]=\bs{0} \quad \text{(strong exogeneity)}. 
\]
\end{assumptionp}

While the narrative argument establishes that the proxy is an unknown function of the target shock and noise, $z_t=f(\varepsilon_{1t},v_t)$, the strong  exogeneity assumption \ref{ass:strong_exogeneity} is sufficient to conclude exogeneity in assumption \ref{ass:exogeneity}. Specifically, with  the law of iterated expectations \ref{ass:strong_exogeneity} implies
\(
\E[\varepsilon_{2t}\mid z_t]=0,
\)
and therefore \(\E[\varepsilon_{2t}h(z_t)]=0\) any measurable function $h(.)$ and consequentially \(\E[\varepsilon_{2t}z_t]=0\). \cite{BK2024} define the implication that the proxy contains  no information on the expected value of the non-target shocks as strong proxy exogeneity and develop a strong proxy exogeneity test.

\begin{assumptionp}{A3}\label{ass:strong_proxy_exogeneity}
The proxy $z_t$ satisfies
\[
\E[\varepsilon_{1t} z_t] \neq 0 \quad \text{(relevance)}
\qquad
\text{and}
\qquad
\E[\bs{\varepsilon_{2t}} \mid z_t] = \bs{0} \quad \text{(strong proxy exogeneity)}. 
\]
\end{assumptionp}
Consequently, conventional exogeneity in \ref{ass:exogeneity}  is sufficient for identification of $b(\beta)$, while
strong exogeneity in  \ref{ass:strong_exogeneity} is sufficient for concluding exogeneity  \ref{ass:exogeneity} of a proxy from the narrative argument. Strong proxy exogeneity in \ref{ass:strong_proxy_exogeneity} is an implication of strong exogeneity. 
The strong exogeneity test developed in the next section builds on  \ref{ass:strong_exogeneity}. Rejection of the test indicates that the link of strong exogeneity between the narrative argument and the conclusion of an exogenous proxy is not supported by the data. 

In what follows, the strong exogeneity is the object of interest in our analysis. 
While assuming that the non-target shocks are unpredictable in mean by the target shock, $\E[\bs\varepsilon_{2t}\mid \varepsilon_{1t}]=\bs 0$, seems stronger than assuming uncorrelated shocks, it has been argued that this assumption is implicit when interpreting coefficients in $B$ (and structural MA coefficients) as impulse responses \citep[e.g.,][]{keweloh2025uncertain}.
To illustrate, consider a bivariate SVAR(0) with $y_{1t}=b_{11} \varepsilon_{1t} + b_{12} \varepsilon_{2t}$. The contemporaneous response of $y_{1t}$ to a unit change in the target shock is the difference between a shocked path with $\varepsilon_{1t}=1$ and its potential path, where the variable is already at its ``steady state'' \citep[e.g.,][]{PesaranShin1998,Shephard2019}:
\begin{align}
\label{eq:causal_irf}
\E[y_{1t}|\varepsilon_{1t}=1,y_{s,s<t}] - \E[y_{1t}|\varepsilon_{1t}=0,y_{s,s<t}] 
= b_{11} +b_{12}
    \left(\E[\varepsilon_{2t}|\varepsilon_{1t}=1 ]-\E[\varepsilon_{2t}|\varepsilon_{1t}=0]\right),
\end{align}
where we use the MDS property of the shocks. Without $\E[\varepsilon_{2t} | \varepsilon_{1t}]= 0$, we cannot guarantee that the second term vanishes, which is required to establish the interpretation of $b_{11}$ as the contemporaneous impulse response to the target shock.\footnote{This holds unless the system is Gaussian which implies that $\E[\varepsilon_{2t} | \varepsilon_{1t}]= 0$ is equivalent to $E[\varepsilon_{2t}\varepsilon_{1t}]=0$. In general, the conditional distribution of non-target shocks given a target shock differs from their marginal distributions, and often the conditional expectation has to be computed via simulation \citep[see, e.g.,][]{PesaranShin1998}. This is also discussed in the literature on average treatment effects, where the researcher is interested in the expected average effect of changing from an experimental design for which the treatment is fixed at unity and to another design for which its treatment is fixed at zero \citep[see, e.g.,][]{savje2021average}. Since the averages are computed over different distributions, it is well known that this effect may not correspond to structural parameters of the model.
}
It is worth recalling that the standard assumption in the measurement error literature implies that the conditional expectation of $y_t$ given $(\varepsilon_{1t}, z_t)$ depends only on $\varepsilon_{1t}$, i.e., $\E[y_t | \varepsilon_{1t}, z_t] = \E[y_t |\varepsilon_{1t}]$ (see, e.g., \citealt{carroll2006measurement} Chp.2.6, see also \cite{KolesarPMgoodbadugly}, Sec.3.2).
In other words, once the shock $\varepsilon_{1t}$ was observed, $z_t$ would not hold any further information for predicting $y_t$. In the current example, this is equivalent to 
\[
\E[\bs\varepsilon_{2t}\mid \varepsilon_{1t}] = \E[\bs\varepsilon_{2t}\mid \varepsilon_{1t},v_t],
\]
which coincides with strong exogeneity formalized in~\ref{ass:strong_exogeneity} under mean independence.

Nevertheless, in applications where non-target shocks are related to higher-moments, the strong exogeneity assumption can be inappropriate. 
For example, consider a target shock $\varepsilon_{1t} \sim N(0,1)$ and a non-target shock $\varepsilon_{2t} = (\varepsilon_{1t}^2 - 1) + \tilde{\varepsilon}_{2t}$, where $\tilde{\varepsilon}_{2t}  \sim N(0,1)$ and independent of $\varepsilon_{1t}$, such that the volatility of the target shock changes the location of the conditional distribution of the non-target shock.
Such higher-order dependence may arise for uncertainty or risk premium shocks, although a majority of empirical studies adopt an identification scheme that exploits only restrictions up to the second-order moment \citep[e.g.,][]{Bloom2009, BachmannElstnerSims2013, Scotti2016}.
In such applications, the researcher should not invoke strong exogeneity without additional justification. 
However, this scenario  makes it problematic to interpret coefficients in $B_0$ as contemporaneous impulse responses. In the example above, $b_{11}$ loses the interpretation as the contemporaneous response of $y_{1t}$ to a unit change in the target shock.

\cite{BK2024} propose a test for the strong proxy exogeneity condition \ref{ass:strong_proxy_exogeneity}. They define a \emph{synthetic proxy} (SP) as the variable $h(z_t)$, with the original proxy $z_t$ and a measurable function $h(.)$. Strong proxy exogeneity implies that the SP is exogenous, i.e. $\E[h(z_t) \bs\varepsilon_{2t}] = \bs{0}$ by the law of iterated expectations. Therefore, the SP can also identify the impact of the target shock. The strong proxy exogeneity test simply assesses joint exogeneity of the original and synthetic proxy.

To see where the identifying information of the SP comes from, write 
$
B_0 =
\begin{pmatrix}
    1 & b_2'\\
    \beta_0 & B_{22}
\end{pmatrix}
$. For a candidate value, write 
$r_t(\beta) \coloneq \bs{u_{2t}}- \beta u_{1t}$ and at the true value, $r_t(\beta_0)  = (B_{22} - \beta_0 b_2') \bs\varepsilon_{2t}$ is proportional to non-target shocks. For example, consider the quadratic SP $z_t^2$, then
\begin{align}   
\E[z_t^2 r_t(\beta) ] = 
(\beta_0 - \beta) \E[z_t^2 \varepsilon_{1t}]
+ (B_{22} - \beta b_2') \E[z_t^2 \bs\varepsilon_{2t}]
=
(\beta_0 - \beta) \E[z_t^2 \varepsilon_{1t}] =: (\beta_0-\beta)\,\rho_2,
\label{eq:moment_expansion}
\end{align}
where $\rho_2:=\E[z_t^2 \varepsilon_{1t}]$ measures relevance of the SP. Hence, $\E[z_t^2 r_t(\beta_0) ]=\bs  0$ can serve as an over-identifying restriction in addition to $\E[z_t r_t(\beta_0) ] = \bs 0$ and allows for a $J$-test. However, as indicated in~\eqref{eq:moment_expansion}, the power of the test depends on the relevance of the SP. If $z_t^2$ contains little information about the target shock so that $\rho_2$ is small, the test has no power. Hence, the practical usefulness of the test depends on (nonlinear) structure in the proxy-shock relationship. In addition, raising a noisy proxy to higher powers can severely inflate the dispersion of the associated sample moments and thereby reduce the power of over-identification tests. To illustrate this, consider the linear proxy $z_t=\phi\,\varepsilon_{1t}+ v_t$, where the noise $v_t \sim N(0, \sigma_v^2)$ is independent of the structural shocks. 
It can be shown that the variance of the moment statistic is proportional to $\E[z_t^4]=\phi^4\E[\varepsilon_{1t}^4]+6\phi^2\sigma_1^2\sigma_v^2+3\sigma_v^4$, which suffers from noise inflation via $\sigma_v^2$ and $3\sigma_v^4$.

To mitigate the power loss from this noise inflation, we propose to first ``denoise'' the external proxy using the information contained in the reduced-form innovations and only then synthesize it. 
Since the noise is orthogonal to $\varepsilon_t$ and $u_t=B_0\varepsilon_t$ spans the shock space, one can remove the measurement noise by projecting $z_t$ onto the linear span of $u_t$. Consider the normalized population linear projection
\[
\tilde z_t^* = \frac{\E[z_tu_t'] \E[u_t u_t']^{-1}}{\sqrt{\E[z_t u_t'] \E[u_t u_t']^{-1}\E[u_t z_t]}} u_t = \frac{\rho_1 b_0' \Sigma_u^{-1} u_t }{\sqrt{\rho_1^2 b_0' \Sigma_u^{-1} b_0  }}.
\]
It can be shown that under~\ref{ass:strong_exogeneity}, $ \tilde z_t^* = \operatorname{sgn}(\rho_1) \varepsilon_{1t}/ \sigma_1$.\footnote{Under~\ref{ass:strong_exogeneity}, $\E[z_t \bs\varepsilon_{2t}] = \b 0$ and
$\E[\bs\varepsilon_{2t}\mid \varepsilon_{1t}]=\bs 0$, which implies $\E[\varepsilon_{1t}\bs\varepsilon_{2t} ]=\bs 0$, and hence $b_0'\Sigma_u^{-1}u_t
    = \iota_1'\Sigma_\varepsilon^{-1}\varepsilon_t
    = \varepsilon_{1t}/\sigma_1^2$
and
$b_0'\Sigma_u^{-1}b_0
    = \iota_1'\Sigma_\varepsilon^{-1}\iota_1
    = 1/\sigma_1^2.$
}
Because the projection lives in the space spanned by the standardized residuals, we can write it as a function of the parameter $\beta$. This gives rise to the \emph{internal proxy}:
\begin{equation}\label{eq:z_t*(beta)}
    z_t^*(\beta)
:=
\frac{b(\beta)'\Sigma_u^{-1}u_t}{\sqrt{b(\beta)'\Sigma_u^{-1}b(\beta)}}.
\end{equation}
For each $\beta$, the internal proxy $z_t^*(\beta)$ is the unit-variance innovation associated with the candidate structural impact parameter vector $b(\beta)$. 
Evaluated at $\beta_0$, it coincides with the standardized structural shock, $z_t^*(\beta_0)=\varepsilon_{1t}/\sigma_1$. 
Thus, $z_t^*$ is a noise-free proxy for the standardized target shock, and internally synthesized $ z_t^*$, e.g., $h(z_t^{*})$, are referred to as
\emph{internal synthetic proxies} (ISPs). The strong exogeneity test assesses the joint exogeneity of the original and internal synthetic proxy.

The ISP uses the non-Gaussianity of the target shock as identifying information. To see this heuristically, consider a quadratic ISP $z_t^{*}(\beta_0)^2$, then
\begin{align}
    \E\left[(z_t^{*})^2\,r_t(\beta)\right] =
\frac{1}{\sigma_1^2}
\left(
(\beta_0 - \beta) \E[\varepsilon_{1t}^{3}] +
(B_{22}-\beta b_2')\E[\varepsilon_{1t}^2\bs\varepsilon_{2t}]\right)
=
\frac{1}{\sigma_1^2} (\beta_0 - \beta) \E[\varepsilon_{1t}^{3}].
\label{eq:moment_expansion_isp}
\end{align}
Hence, the ISP obtains over-identifying information when the target shock is skewed, but becomes irrelevant when the target shock is Gaussian and the strong exogeneity test loses its power.

The feature that for a valid proxy $z_t$, our proposed ISP is unaffected  by the proxy noise, is the main advantage of our  strong exogeneity test over the strong proxy exogeneity test based on SPs in \cite{BK2024}.
SPs directly synthesize the observed proxy and, consequently, also incorporate its noise, leading to a loss of power when proxies are noisy. In contrast, ISPs first construct an internal, denoised proxy and then synthesize it, thereby avoiding this loss of power.
This distinction is particularly important when proxies are both endogenous and noisy. In such settings, the SP-based test is unlikely to reject exogeneity due to its low power, whereas the ISP-based test retains the ability to detect violations of exogeneity.

The ability to create an internal, denoised proxy relies on the assumption that the SVAR in \eqref{eq:SVAR} is correctly specified.
Under misspecification (e.g., due to  incorrect lag length,  omitted variables, non-fundamentalness, non-linear shock transmission, etc.), the non-target shocks implied by the fitted VAR can remain predictable from the ISPs, even when $z_t$ is indeed strongly exogenous with respect to the shocks in the true model. This is illustrated in the MC results in Section~\ref{sec:MC}.  
Therefore, before employing the suggested exogeneity tests and attributing a rejection to proxy contamination, the researcher should conduct model diagnostics to check whether a given model represents the data generating process of the variables adequately.

\section{Strong exogeneity test}\label{sec:theory}

This section derives the large-sample properties of the strong exogeneity test based on ISP.\footnote{Appendix~\ref{app:SP_asymptotics} provides the large-sample theory for the SP-based strong proxy exogeneity test. \cite{BK2024} develop the SP-based test for static simultaneous equations models (SEMs) or VAR(0), which is no longer valid once the VAR(\(p\)) contains lags ($p>0$). In particular, the test statistic does not have the usual $\chi^2$ limit under the null, unless adaptivity occurs. See also Appendix~\ref{app:adaptivity} on adaptive testing.}

Denote $\pi_0 = \text{vec}([\nu,  A_1,...,\  A_p])$  the vector that collects the reduced-form parameters, define vectors $\b y = \vect{y_1, \ldots, y_T}$ and $Y_t = (1, y_t', \ldots, y_{t-p+1}')'$ of dimension $Tn\times 1$ and $(np+1) \times 1$ respectively, and matrix $Y = (Y_0, \ldots, Y_{T-1})$ of dimension $(np+1) \times T$ , then reduced-form model in \eqref{eq:SVAR} can be compactly written as $y_t = (Y_{t-1}' \otimes I_n) \pi_0 + u_t$. 

Typically, a structural VAR is estimated in two steps. First, the reduced-form parameter $\pi$ is estimated via least squares (LS) by solving the normal equation
\[
\frac1T \sum^T_{t=1}m_t(\pi)=\bs 0, \qquad \text{with } \;   m_t(\pi) = (Y_{t-1} \otimes I_n)u_t(\pi), 
\]
which leads to the LS estimator with the familiar form $\hat \pi_T = [(YY')^{-1}Y\otimes I_n]\b y.$ 
Let $u_t(\hat\pi_T)$ denote the reduced-form residuals implied by $\hat\pi_T$.
In the second step, the simultaneous interaction $u_t = B_0 \varepsilon_t$ is estimated based on the reduced-form residuals $u_t(\hat\pi_T)$.

The projection $z_t^*$ in \eqref{eq:z_t*(beta)} involves unobservable reduced-form residuals $u_t$ and its covariance $\Sigma_u$, both of which depend on parameter $\pi$. To obtain an operational procedure that automatically updates with $(\beta,\pi)$, we define $\omega \coloneqq \vect{\Sigma_u}$ with $\Sigma_u(\omega) \equiv \operatorname{vec}^{-1}(\omega)$, $\Sigma_u(\omega_0) = \Sigma_u$ being positive definite, and write
\begin{equation}
    z_t^*(\beta,\pi,\omega)
:=
\frac{b(\beta)'\Sigma_u(\omega)^{-1}u_t(\pi)}{\dot\sigma_1(\beta,\omega)},
\qquad
\dot\sigma_1^2(\beta,\omega):=b(\beta)'\Sigma_u(\omega)^{-1}b(\beta),
\end{equation}
where $\dot\sigma_1(\beta,\omega):=\sqrt{\dot\sigma_1^2(\beta,\omega)}$ with $\dot\sigma_1(\beta_0,\omega_0) = \sigma_1^{-1}$. 
Consider a sequence of power functions $\{h_i\}_{i = 1}^L$ with $h_i(z) = z^i$, define 
\[
g_{i,t}^*(\beta,\pi,\omega):=h_i \big(z_t^*(\beta,\pi,\omega)\big)\,r_t(\beta,\pi)\in\R^{n-1} .
\]
To exploit information in the external proxy while avoiding noise inflation for its synthetics,
we combine the original linear proxy moment with ISP moments and stack
\[
g_t^{(Z^*)}(\beta,\pi,\omega)
:=
\big(g_{1,t}^*(\beta,\pi)',g_{2,t}^*(\beta,\pi,\omega)',\ldots,g_{L,t}^*(\beta,\pi,\omega)'\big)',
\qquad
g_{1,t}^*(\beta,\pi):=z_t\,r_t(\beta,\pi) 
\]
and $\bar g_T^{(Z^*)}(\beta,\pi,\omega):=T^{-1}\sum_{t=1}^Tg_t^{(Z^*)}(\beta,\pi,\omega)$.

Let $s_t(\pi, \omega) \coloneq \vect{u_t(\pi) \, u_t(\pi)' - \Sigma_u(\omega)}$.
Given the LS estimator $\hat\pi_T$ of the reduced form, let $\hat\omega_T$ solve
$\bar s_T(\hat \pi_T, \omega) \coloneqq \frac1T \sum^T_{t=1} s_t(\hat \pi_T, \omega) = \bs 0$.
The GMM estimator based on ISPs is
\[
\hat\beta_T^*(W)
:=
\argmin_{\beta}\;
\bar g_T^{(Z^*)}(\beta,\hat\pi_T,\hat\omega_T)'\,\hat W_T\,
\bar g_T^{(Z^*)}(\beta,\hat\pi_T,\hat\omega_T).
\]

The \textit{strong exogeneity test} jointly evaluates the original linear proxy moment $g_{1,t}^*(\beta,\pi)$ and the internal synthetic proxy moments $g_{i,t}^*(\beta,\pi,\omega)'$ for $i=2,...,L$ based on the J-statistic
\begin{equation} 
\label{eq:strong_exogeneity_test}
J_T^*(W)
=
T\,\bar g_T^{(Z^*)}(\hat\beta_T^*(W),\hat\pi_T,\hat\omega_T)'\,\hat W_T\,
\bar g_T^{(Z^*)}(\hat\beta_T^*(W),\hat\pi_T,\hat\omega_T).
\end{equation}

The next proposition derives large-sample theory for $\hat\beta_T^*(W)$ and $J_T^*(W)$,
including the optimal weighting matrix that leads to the usual $\chi^2$ limit.

\begin{proposition}\label{thm:beta_ISP}
Under Assumption~\ref{ass:strong_exogeneity}, $\E[g_t^{(Z^*)}(\beta_0,\pi_0,\omega_0)] = \bs 0.$ 
Denote
\[
G_\beta^*:=\E\!\left[\frac{\partial g_t^{(Z^*)}(\beta_0,\pi_0,\omega_0)}{\partial \beta'}\right],\quad
G_\pi^*:=\E\!\left[\frac{\partial g_t^{(Z^*)}(\beta_0,\pi_0,\omega_0)}{\partial \pi'}\right],\quad
G_\omega^*:=\E\!\left[\frac{\partial g_t^{(Z^*)}(\beta_0,\pi_0,\omega_0)}{\partial \omega'}\right],
\]
and $
M \coloneqq  - (\Gamma \otimes I_n).$ Define moment functions
$
g_t^{(Z^*),\perp}
:=
g_t^{(Z^*)}(\beta_0,\pi_0,\omega_0)
-
G_\pi^*M^{-1}m_t(\pi_0)
+
G_\omega^*\,s_t(\pi_0,\omega_0).
$
Then, under Assumptions~\ref{ass:strong_exogeneity},~\ref{Tass:base} and~\ref{Tass:ISP}
$$
\sqrt{T} \left(\hat \beta_T^* (W) - \beta_0 \right) \ \Rightarrow \ N(\bs 0, V_{\beta}^*),
\qquad
V_{\beta}^* = (G_\beta^{*'} W G^{*}_\beta)^{-1} G_\beta^{*'} W  S_{\perp}^* W  G_\beta^* (G_\beta^{*'} W G^{*}_\beta)^{-1},
$$
with
$S_\perp^*:=\E \left[g_t^{(Z^*),\perp}g_t^{(Z^*),\perp\,\prime} \right]$.
Hence, the asymptotically efficient weight is $ S_\perp^{*-1}$ such that $\sqrt{T} \left(\hat \beta^*_T (S_\perp^{*-1}) - \beta_0 \right) \ \Rightarrow \ N \left(\bs 0, (G_\beta^{*'} S_{\perp}^{*-1} G_\beta^{*})^{-1}\right)$. 

Moreover, if $L>1$, under the null $H_0:\ \E[g_t^{(Z^*)}(\beta_0,\pi_0,\omega_0)]=\bs 0$,
\[
 J_T^*\left(S_\perp^{*-1}\right) \Rightarrow \chi^2_{(L-1)(n-1)}.
\]
\end{proposition}
\begin{proof}
    The proof is provided in Appendix~\ref{app:proof}.
\end{proof}
Efficiency and valid inference for the over-identifying restrictions require a weighting matrix that accounts for first-step estimation of the reduced-form parameter $\pi$ and the second-step estimation of the reduced-form covariance $\omega=\vect{\Sigma_u}$.\footnote{
Appendix~\ref{app:adaptivity} provides some adaptivity properties when the researcher employs a single quadratic ISP ($L=2$). We show that a simple set of primitive conditions yields \emph{joint adaptivity} with respect to $(\pi,\omega)$, so that their estimation uncertainty becomes asymptotically negligible for testing $\beta$.
} 
A feasible estimator is obtained by the plug-in covariance of the orthogonalized moment sequence:
\[\hat S_{\perp,T}^* \;:=\; \frac1T\sum_{t=1}^T
\hat g_{t,T}^{(Z^*),\perp}\,\hat g_{t,T}^{(Z^*),\perp\,\prime},
\]
where $\hat g_{t,T}^{(Z^*),\perp}
:=
g_t^{(Z^*)}(\tilde\beta_T^*,\hat\pi_T,\hat\omega_T)
-\hat G_{\pi,T}^*
\hat M_T^{-1}\, m_t(\hat\pi_T)
+\hat G_{\omega,T}^*\, s_t(\hat\pi_T,\hat\omega_T)$ with $\tilde\beta_T^*$ being any preliminary consistent estimator of $\beta$ (e.g., the one-step GMM estimator using identity or naive weighting). Appendix~\ref{app:weighting_est} provides further implementation details, including closed-form expressions for $\hat M_T^{-1}$, $\hat G_{\pi,T}^*$ and $\hat G_{\omega,T}^*$.

The strong exogeneity test in \eqref{eq:strong_exogeneity_test} examines the (joint) validity of the $L(n-1)$ over-identifying moment restrictions $\E\!\left[g_t^{(Z^*)}(\beta_0,\pi_0,\omega_0)\right]=\b 0$. A rejection therefore does not by itself reveal which individual restriction is responsible for the rejection. In particular, it cannot determine whether the moment condition for the original proxy or those for the ISP are violated. 
However, under the strong exogeneity assumption \ref{ass:strong_exogeneity}, the original and ISP moment conditions hold, $\E\!\left[g_t^{(Z^*)}(\beta_0,\pi_0,\omega_0)\right]=\b 0$. 
Therefore, a rejection of the strong exogeneity $J$-test yields evidence against the strong exogeneity assumption \ref{ass:strong_exogeneity}. 
This directly implies that the impulse responses to the target shock obtained from the original proxy cannot be interpreted causally in the sense of \eqref{eq:causal_irf}, since rejection of strong exogeneity implies that expected value of the non-target value is non-zero in response to the target shock.
Additionally, the strong exogeneity assumption is a sufficient condition which allows to conclude that the proxy is exogenous, $\E[z_t \varepsilon_{2t}]=0$, from the narrative argument establishing that the proxy is an unknown function of the target shock and noise, $z_t=f(\varepsilon_{1t},v_t)$, see Section \ref{sec:intuition}. Consequently, rejecting the strong exogeneity $J$-test yields evidence that the narrative argument may not imply conventional proxy exogeneity.

Finally, the suggested strong exogeneity test requires the researcher to choose the number of synthetics $L$.
Increasing $L$ enlarges the set of implications being tested, but also introduces higher moments, which are known to be difficult to estimate reliably and sensitive to outliers.
Hence, adding higher-order synthetics can introduce bias and cause oversized tests in finite samples.\footnote{Since optimal weighting also involves estimating the variation of these higher-order moments, i.e., $\E[z_t^{2L}]$ and $\E[\varepsilon_{1t}^{2L}]$, the resulting GMM estimator and test statistic may become imprecise for large $L$,  especially in samples of the size commonly available in macroeconomic applications.}  In addition, higher moments also become marginally less informative for a proxy defined on a compact set (e.g., $z_t^{10}$ and $z_t^{12}$ can be highly correlated for small $z_t$). Therefore, the number of synthetic moments $L$ should remain small.
\\
In practice, the decision about $L$ can also be made based on the non-Gaussianity of the target shock, which governs the information content of the ISP. The next section establishes that the identifying strength of a quadratic ISP is governed by $\delta_2=\E[X_t^3]$, whereas a cubic ISP additionally exploits
$\delta_3=\E[X_t^4]-3$, with $X_t=\varepsilon_{1t}/\sigma_1$. 
Hence, sample skewness and excess kurtosis of the estimated target shock 
$\hat z_t^* := z_t^*(\hat \beta, \hat\pi, \hat\omega)$, with $\hat \beta$ being an estimate based on the just-identifying proxy alone, are useful for gauging the additional identification content of $h_2(z_t^*)$ and $h_3(z_t^*)$, respectively.
Higher values of $L$ should be employed only in larger samples and with a clear economic or
reliable distributional justification. It is also worth noting that selecting $L$ based on a first-stage pretest can induce size distortion in standard inference \citep[see][for issues related to pretesting]{Andrews_weak_2019}.

\section{Weak proxy and identification-robust test} 
\label{sec:SP_weak}
The last section assumed $\rank(G_\beta^*)=n-1$ such that $\beta$ is identified. A sufficient condition for this is that the proxy stays relevant, i.e., $\rho_1 \neq 0$. Proxies, especially those obtained from a high-frequency identification (HFI), are often found to be weak. 
This section studies the   behavior of the strong exogeneity test under weak proxies.

Let $\rho_{1,T} = \E_T[z_{t,T}\varepsilon_{1t} ]$ and let $G_{\pi,T}^*$, $G_{\omega,T}^*$ and $G_{\beta,T}^*$ be defined analogously.
To simplify notation, we suppress the dependence of objects (e.g., $\E_T$, $z_{t,T}$, $G_{\bullet, T}$ etc.) on $T$ whenever no confusion arises while keeping the triangular array setting in mind.\footnote{
Without complicating the notation, the technical Assumption~\ref{Tass:base} and the corresponding identification assumptions~\ref{ass:strong_proxy_exogeneity} and~\ref{ass:strong_exogeneity} will be understood to hold along the triangular-array sequence. This requires that the conditions in \ref{Tass:base} hold uniformly in $T$, e.g., $\sup_T\E_T|z_t|^{2L+\delta}<\infty$, $\Gamma:=\lim_{T\to\infty}\E_T[Y_{t-1}Y_{t-1}']$ exist and is nonsingular, $\E_T[\varepsilon_{2t,T}\mid \varepsilon_{1t,T},v_{t,T}]=0$ holds for all $T$.
}
It is illustrative to follow the weak identification literature \citep[e.g.,][]{StaigerStock1997, StockWright2000} by assuming that the covariance between the original proxy and the target shock is local to zero.
\begin{assumptionp}{W1}\label{ass:weak1}
There exists $\gamma_1 \in \R$, such that $\sqrt{T} \rho_{1,T} \to \gamma_1$ as $T\to \infty$.
\end{assumptionp}
While Assumption~\ref{ass:weak1} provides a standard framework for analyzing proxy SVARs \citep[see, e.g.,][]{MontielOleaStockWatson2021}, it remains elusive for the moment what this means for the synthetic proxies. 
In Appendix~\ref{app:weakSP} we show that weak identification arises for the SP-based strong proxy exogeneity test, if all SP are jointly weak and we provide a weak-identification-robust test under this condition. The higher-order synthetics used in the ISP approach, however, are not functions of $z_t$, but those of the parameters. It turns out that in case the external proxy is weak,  identification of the ISP criterion can be supplied by non-Gaussian features of the target shock, as demonstrated below.\footnote{To make the analysis more transparent and insightful, the discussion in the main text focuses on the scenario where the non-target shock are conditionally homoskedastic, while leaving the general case for Lemma~\ref{lem:Stein} in Appendix~\ref{app:proof}.} 
\begin{proposition}\label{thm:G_beta_weak}
Let $\Sigma_{22}:=\E[\bs\varepsilon_{2t}\bs\varepsilon_{2t}']$.
Suppose $\E[\bs\varepsilon_{2t}\bs\varepsilon_{2t}'\mid \varepsilon_{1t}]=\Sigma_{22}$, let	$X_{t,T} := \frac{\varepsilon_{1t}}{\sigma_1}$,
write $h_i'(x) = \frac{\mathsf d}{\mathsf{d} x} h_i(x) = i x^{i-1}$ for $i =2, 3, \ldots, L$ and define $\bs\delta_T^{(L)}
:=(\delta_{2,T},\ldots,\delta_{L,T})'\in \R^{L-1},$ where
\begin{equation}\label{eq:delta_stein}
	\delta_{i,T} \equiv \delta_{i,T}(X_{t,T}) := \E_T[X_{t,T}h_i(X_{t,T})] - \E_T[h_i'(X_{t,T})],
\end{equation}
provided that the expectations exist.
Then, under Assumption~\ref{ass:strong_exogeneity}, 
\begin{equation}
G_{\beta, T}^* :=
\E_T\!\left[
\frac{\partial g_t^{(Z^*)}(\beta_0,\pi_0,\omega_0)}
{\partial\beta'}
\right] = -\begin{pmatrix}
    \rho_{1, T}\\
	\sigma_1\bs\delta_T^{(L)}
    \end{pmatrix}\otimes I_{n-1}.	
\end{equation}
\end{proposition}
\begin{proof}
	The proposition follows directly from Lemma~\ref{lem:Stein}.
\end{proof}

Proposition~\ref{thm:G_beta_weak} entails that the squared singular values of the Jacobian $G_{\beta, T}^*$ are $ \rho_{1, T}^2 + \sigma_1^2 \|\bs\delta_T^{(L)} \|^2 $. Hence, when the original proxy is weak s.t., the first term vanishes in the limit, the remaining identifying strength is measured by $\|\bs\delta_T^{(L)} \|$. For example, $\delta_{2,T} = \E_T[X_{t,T}^{3}]$ and $\delta_{3,T} = \E_T[X_{t,T}^{4}] - 3$, etc., so that the identifying information via the quadratic and cubic ISPs is provided by skewness and excess kurtosis of the standardized target shock.

Note that $\E[X_{t,T}] = 0$ and $\E[X_{t,T}^2] = 1$. A well-known result from \citet[][Lemma 1]{Stein1981} is that if $X \sim N(0,1)$ is standard Gaussian, then 
\[
\delta_{h}(X) = \E[Xh(X)] - \E[h'(X)] = 0
\]
 for every continuously differentiable function $h$ for which the expectations are well-defined. For the polynomial basis $h_i(x) = x^i$ in particular, we obtain the Gaussian moment recursion formula 
 $\E[X^{i+1}] = i\E[X^{i-1}]$ for $i \geq 2$, i.e., $\delta_{i}(X) = 0$.\footnote{The result follows by integration by parts. In fact, the converse also holds if $\delta_{h} = 0$ holds for every continuously differentiable function for which both $h$ and $h'$ are bounded.} 
 By a contrapositive argument, $\delta_{h} \neq 0$ implies non-Gaussianity. 
 Hence, $\delta_{h}(X)$ measures the deviation from the Stein identity and thereby the deviation from Gaussianity once a broad class of $h$ is considered.

Therefore, a weak proxy does not pose a problem for standard asymptotics described above, as long as the test can draw identification information from non-Gaussianity of the shock. This path has recently been exploited in \cite{Junare2026}.
 However, the weak identification problem does arise once the target shock is ``weakly Gaussian'' in the sense that, as $T$ becomes larger, the distribution of $X_{t,T}$ becomes eventually indistinguishable from a Gaussian distribution with respect to the selected collection of $\bs \delta_T^{(L)}$. 
 Formally, a local weak identification setting is that there exists some $\bs\gamma^*:=(\gamma_1,\bs\kappa^{(L)\prime})' \in \R^{L}$ such that
\begin{equation}\label{eq:W3_local}
\sqrt T
\begin{pmatrix}
\rho_{1,T}\\
\sigma_1\bs\delta_T^{(L)}
\end{pmatrix}
\longrightarrow
\begin{pmatrix}
\gamma_1\\
\bs \kappa^{(L)}
\end{pmatrix} 
= \bs \gamma^*, 
\qquad
\text{as } \ T \to \infty,
\end{equation}
i.e., the objective function is asymptotically flat near the true value $\beta_0$. 
Despite being insightful, the local condition in \eqref{eq:W3_local} does not control the population ISP moment function away from $\beta_0$. 
Let 
$\bs \psi_T^*(\beta) := \E_T\!\left[ g_t^{(Z^*)}(\beta,\pi_0,\omega_0) \right]$.
We start by making a global weak identification assumption and will show that it implies local weak identification. 
\begin{assumptionp}{W2}\label{ass:weak3}
For every $T \in \N$, $\bs\psi_T^*$ is continuously differentiable and there exists a continuous matrix-valued function
$\bs G^*:\mc B\to\R^{L(n-1)\times (n-1)}$, s.t.,
\begin{equation}\label{eq:W3_derivative_uniform}
\sup_{\beta\in\mc B}
\left\|
\sqrt T\,
\frac{\partial \bs \psi_T^*(\beta)}{\partial\beta'}
-\bs G^*(\beta)
\right\|
\longrightarrow0.
\end{equation}
\end{assumptionp}

\begin{lemma}\label{lem:weak3_glb}
Under the conditions of
Proposition~\ref{thm:G_beta_weak}, \ref{ass:strong_exogeneity}, \ref{ass:weak3} and
\ref{Tass:ISPweak}(i),(iii), \eqref{eq:W3_local} holds for some
\(\bs\gamma^*\in\R^L\) and 
there exists a function $\bs\Psi^*:\mc B\to\R^{L(n-1)}$ that is continuously differentiable on \(\operatorname{int}(\mc B)\),
s.t., $\bs\Psi^*(\beta_0)=\bs 0$ and 
\begin{equation}\label{eq:ISP_global_drift}
\sup_{\beta\in\mc B}
\left\|
\sqrt T\,\bs\psi_T^*(\beta)-\bs\Psi^*(\beta)
\right\|
\longrightarrow0.
\end{equation}
Moreover, 
$\frac{\partial\bs \Psi^*(\beta)}{\partial\beta'}
=\bs G^*(\beta)$, $\forall \beta\in\operatorname{int}(\mc B)$ and 
$\frac{\partial\bs \Psi^*(\beta_0)}{\partial\beta'}
=\bs G^*(\beta_0)
=-\bs \gamma^*\otimes I_{n-1}.$
\end{lemma}
\begin{proof}
	The proof is provided in Appendix~\ref{app:proof}.
\end{proof}
In the local identification condition, $\bs \gamma^*$ describes the local slope of the population function at the true value
and $\bs\Psi^*(\beta)$ in Lemma~\ref{lem:weak3_glb} describes the population drift at any $\beta$.\footnote{
Note that for the SP-based test discussed in Appendix~\ref{app:weakSP}, 
the population drift term is $(\bs\gamma\otimes I_{n-1})(\beta_0-\beta)$ under weak identification, which is linear and determined by the slope at $\beta_0$, i.e., $\bs\gamma$. This is not the case for the ISP-based tests as $\bs\Psi^*(\beta)$ is generally  nonlinear and is not determined by $\bs \gamma^*$ alone.
} 
We define the stochastic process
\begin{align*}
\mc G_T^*(\beta)
:=
\frac1{\sqrt T}\sum_{t=1}^T
\Big[
g_t^{(Z^*)}(\beta,\pi_0,\omega_0)
-\bs\psi_T^*(\beta)
-G_{\pi}^*(\beta)M^{-1}m_t(\pi_0)
+G_{\omega}^*(\beta)s_t(\pi_0,\omega_0)
\Big],
\end{align*}
i.e., a random $\R^{L(n-1)}$-valued function continuous on $\mc B$. 
Suppose $\mc G_T^*(\beta)\Rightarrow\mc G^*(\beta)$, where \(\mc G^*(\beta)\) is a zero-mean Gaussian stochastic process with covariance function 
$
\Omega^*(\beta_1,\beta_2) := \E\!\left[ \mc G^*(\beta_1)\mc G^*(\beta_2)' \right],$ 
and $\Rightarrow$ denotes
weak convergence of random functions on $\mc B$. We use $\Omega^*(\beta)$ as a shorthand notation for $ \Omega^*(\beta, \beta)$.
Define the (rescaled) sample objective function for a continuous updating estimator
\begin{equation}\label{eq:Q_T_ISP_weak}
Q_T^*(\beta)
:=
T\,
\bar g_T^{(Z^*)}(\beta,\hat\pi_T,\hat\omega_T)'
\hat\Omega_T^*(\beta)^{-1}
\bar g_T^{(Z^*)}(\beta,\hat\pi_T,\hat\omega_T).
\end{equation}
The following proposition provides the weak limit of the objective function. 
\begin{proposition}\label{thm:weak_Q_ISP}
Under the conditions of
Proposition~\ref{thm:G_beta_weak}, \ref{ass:strong_exogeneity}, \ref{ass:weak3},
\ref{Tass:base}, and \ref{Tass:ISPweak}, \begin{equation}\label{eq:Q_ISP_process_limit}
Q_T^*(\beta)
\Rightarrow
\left[
\mc G^*(\beta)+\bs\Psi^*(\beta)
\right]'
\Omega^*(\beta)^{-1}
\left[
\mc G^*(\beta)+\bs\Psi^*(\beta)
\right]
=:
Q_\infty^*(\beta)
\end{equation}
on $\mc B$ with respect to the sup norm, and hence, at the true value, $Q_T^*(\beta_0)\Rightarrow\chi^2_{L(n-1)}.$ Let $\hat\beta_{CUE,T}^*
:=\argmin_{\beta\in\mc B}Q_T^*(\beta)$ and denote $J_{CUE,T}^*$ the corresponding $J$-statistic, then 
$\hat\beta_{CUE,T}^*\Rightarrow
\argmin_{\beta\in\mc B}Q_\infty^*(\beta),$ and
$J_{CUE,T}^*\Rightarrow
\inf_{\beta\in\mc B}Q_\infty^*(\beta).$
\end{proposition}
\begin{proof}
	The proof is provided in Appendix~\ref{app:proof}.
\end{proof}
Hence, the GMM estimator need not be consistent and the limiting distribution of  $J_{CUE,T}^*$ is nonstandard. Despite this, $Q_T^*(\beta_0)\Rightarrow\chi^2_{L(n-1)}$ still allows us to construct a test statistic $\varphi_{T}^*(\alpha)$ based on the corresponding confidence set.
Specifically, for a nominal level $\alpha$, define the set $\mc C_T^*(1-\alpha)$ and test statistic $\varphi_{T}^*(\alpha)$:
\begin{align}\label{eq:CT_varphi_star}
\mc C_T^*(1-\alpha)
    :=
    \left\{
    \beta\in\mc B:
Q_T^*(\beta) \leq \chi^2_{L(n-1),1-\alpha}
    \right\},
    \qquad
\varphi_{T}^*(\alpha)
    :=
\mathbf 1 \left\{ \mc C_T^*(1-\alpha)=\varnothing \right\},
\end{align}
where $\mathbf 1\{.\}$ is the indicator function.
The null hypothesis of strong exogeneity is rejected if no point in
\(\mc B\) satisfies the moment restrictions at the $\alpha$-level critical value so that 
$\varphi_{T}^*(\alpha) = 1$. Notably, the set $\mc C_T^*(1-\alpha)$ can be used to construct the so-called `$S$-sets' in \cite{StockWright2000}, which are asymptotically the nonlinear GMM analogue of the \cite{AndersonRubin1949} confidence sets and have been recently suggested by \cite{MontielOleaStockWatson2021} for constructing weak proxy robust pointwise confidence intervals for impulse responses identified in proxy SVARs.
The robust strong exogeneity test is consistent against fixed alternatives and has size that is asymptotically bounded by $\alpha$.
\begin{proposition}\label{thm:weak_test_ISP}
Under the assumptions of Proposition~\ref{thm:weak_Q_ISP}, under $H_0: \ \E [ g_t^{(Z^*)}(\beta_0,\pi_0,\omega_0)] =\bs0,$
\[
\lim_{T\to\infty} \Pr \left(\beta_0\in \mc C_T^*(1-\alpha) \right)  = 1-\alpha,
\qquad
\limsup_{T\to\infty} \Pr \left(\varphi_{T}^*(\alpha)=1 \right)
        \le \alpha.
\]
Under the fixed $H_1: \liminf_{T \to \infty} \inf_{\beta}\,\big\|\E[g_t^{(Z^*)}(\beta,\pi_0,\omega_0)] \big\| > 0$, 
$\Pr \left(\varphi_{T}^*(\alpha)=1 \right)\to 1.$
\end{proposition}
The proposition follows by the same arguments as in the proof of Proposition~\ref{thm:weak_test} in Appendix~\ref{app:weakSP}. The test is asymptotically conservative.
In particular, when $\bs\Psi^*(\beta)$ is small over a substantial part of $\mc B$, the set
\(\mc C_T^*(1-\alpha)\) can be large. 
In practice, the test is implemented by solving the continuous minimization
problem $\inf_{\beta\in \mc B}Q_T^*(\beta)$.
The null is rejected if the computed
minimum exceeds $\chi^2_{L(n-1),1-\alpha}$. One can resort to a global
optimization over the compact set \(\mc B\) using multiple starting values.  Rejection requires numerical evidence that the global minimum is above the critical value, so if a value $\tilde \beta$ with
$Q_T^*(\tilde\beta)\le\chi^2_{L(n-1),1-\alpha}$ is found, the set is nonempty and the
test does not reject.

\section{Monte Carlo evidence}\label{sec:MC}

This section compares the finite-sample properties of the ISP-based strong exogeneity test to those of the SP-based strong proxy exogeneity test of \cite{BK2024}. We first consider a correctly specified SVAR and study how size and power vary
with the degree of proxy contamination, proxy relevance, non-Gaussianity, and the sample size. We then
examine the identification-robust tests developed in Section~\ref{sec:SP_weak}. We show that ISP is generally more powerful than SP when proxies have low to moderate relevance.
Finally, we depart
from the correct model specification through an omitted nonlinear transmission and an incorrect lag order
to gauge the test performance under potential VAR misspecification.

\subsection{Baseline design}

We consider a three-variable SVAR(4)
\begin{equation}
    y_t=\sum_{i=1}^{4}A_i y_{t-i}+B_0\varepsilon_t,
    \qquad
    B_0=
    \begin{pmatrix}
        1&0&0\\
        1&1&0\\
        -0.5&0&1
    \end{pmatrix},
    \qquad
    \varepsilon_t=(\varepsilon_{1t},\varepsilon_{2t},\varepsilon_{3t})',
\end{equation}
where the coefficient matrices $\{A_i\}_{i=1}^4$ generate a stable process with
moderate persistence and declining lag coefficients. 
The structural shocks are mutually independent with zero mean, unit variance, skewness $-1$ and excess kurtosis $2$ in the baseline setting. 
The first structural shock is the target shock.

Let $v_t\sim N(0,1)$ be independent of $\varepsilon_t$ and define the proxy
\[
z_t = \rho\varepsilon_{1t} + \kappa\rho\varepsilon_{2t} + \sqrt{1-\rho^2(1+\kappa^2)}\,v_t.
\]
Note that $\var(z_t) = 1$ such that the correlations between the proxy and the shocks are 
$\operatorname{Corr}(z_t,\varepsilon_{1t})=\rho$ and $\operatorname{Corr}(z_t,\varepsilon_{2t})=\kappa\rho$. Hence, $\rho$ directly measures the correlation of the proxy with the target
shock, whereas $\kappa$ is the ratio of its non-target and target correlations. The value
$\kappa=0$ corresponds to an exogenous proxy. Negative values of $\kappa$ introduce
contamination by the second structural shock, with larger values of $|\kappa|$ representing
more severe violations of proxy exogeneity. 
We consider sample sizes $T\in\{250,500,1000\}$, proxy correlations
$\rho\in\{0.10,0.20,0.35,0.50,0.60\}$ and $\kappa\in\{0,-0.1,\ldots,-1\}$.\footnote{Since $|\kappa|\leq 1$ throughout, we assume that the researcher is able to obtain or construct a somewhat reasonable proxy that is at least as strongly correlated with the target shock as with the contaminating shock.}
The design includes proxy correlations between $0.20$ and $0.35$, which cover
the range frequently encountered in HFI applications, as well as a particularly weak calibration with $\rho=0.10$.
To relate the correlation calibration to the relevance diagnostics commonly reported in empirical applications, for sample size $T=500$, the five values of population correlation $\rho$ correspond to first-stage $F$ statistics 
$F_T(\rho)\in\{5.0,20.8,69.5,166.0,280.1\}$. 
In the baseline experiments, the researcher estimates a correctly specified reduced-form
VAR(4), retrieves the reduced-form residuals and constructs the SPs and ISPs up to
$L=3$. The standard inference uses a two-step feasible optimal-weight GMM, where a preliminary
estimate is used to estimate the covariance matrix of the nuisance-adjusted influence function $g_t^{(Z), \perp}$ and $g_t^{(Z^*), \perp}$, and its inverse is used to obtain the final GMM estimate and the associated $J$-statistic. The test decisions are formed based on the standard $\chi^2_{(L-1)(n-1)}$ limit at the nominal level $\alpha=0.1$. Each experiment is performed with $1{,}000$ replications.

\begin{figure}[ht]
    \centering
    \caption{Finite-sample size and power of the standard SP and ISP exogeneity tests}
    \label{fig:mc-main-standard-rho-kappa}
    \includegraphics[width=\textwidth]{\detokenize{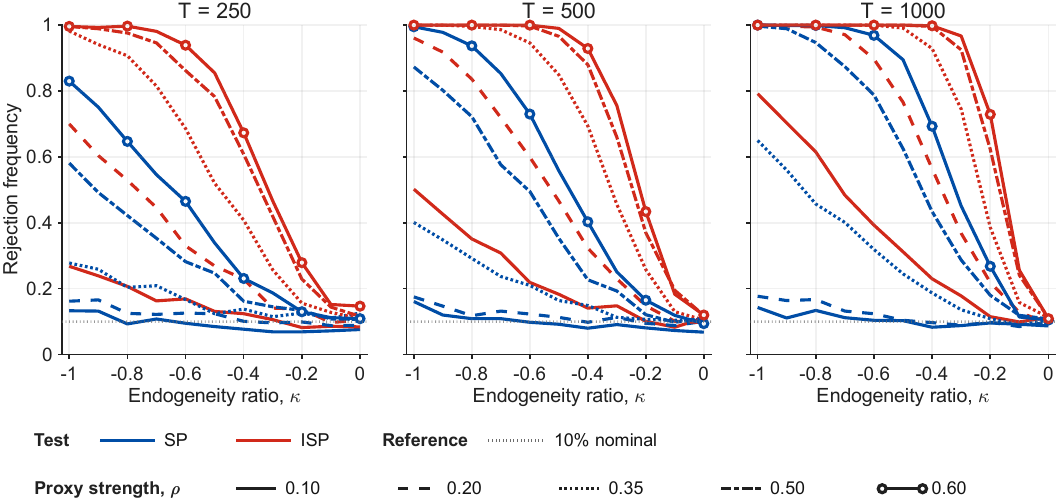}}
    \vspace{0.35em}
    \begin{minipage}{0.98\textwidth}
        \footnotesize
        \textit{Notes:} The DGP is the baseline SVAR(4) with the baseline proxy.
   \end{minipage}
\end{figure}

Figure~\ref{fig:mc-main-standard-rho-kappa} reports the rejection frequencies of the standard SP- and
ISP-based tests. Two findings stand out. First, under exogeneity $\kappa=0$, both tests
have rejection frequencies close to the nominal level across the different sample sizes and
proxy correlations (incl.\ designs with very low $\rho\in\{0.1,0.2\}$). This confirms that feasible estimation of the nuisance-adjusted optimal weighting accounts adequately for the first-step estimation uncertainty so that the usual $\chi^2$ limit of the $J$-statistic is restored. Second, under proxy endogeneity, rejection frequencies generally increase with $|\kappa|$ (stronger evidence against the null in the data), $\rho$ and $T$ (higher power of the tests), while the proposed ISP-based test has more power than the SP-based test to detect the null violations in each configuration.
The power difference is especially pronounced for low to moderately relevant proxies, e.g., $\rho\in\{0.2,0.35,0.5\}$. At the weakest calibration $\rho=0.10$, both procedures eventually lose power, although the proposed test still has nearly three times (more than five times) the chance to correctly reject the null with $T=500$ ($T=1,000$) compared to the SP test at $\kappa = -1$.

\subsection{Non-Gaussianity and testing power}

\begin{figure}[t]
    \centering
    \caption{Finite-sample power and non-Gaussianity}
    \label{fig:mc-nongaussianity-heatmaps}
    \includegraphics[width=\textwidth]{\detokenize{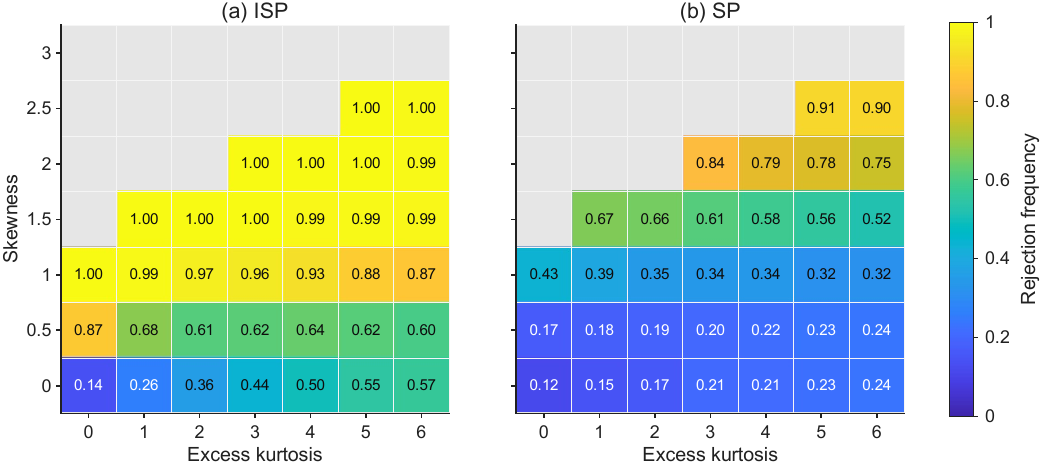}}
    \vspace{0.35em}
    \begin{minipage}{0.98\textwidth}
        \footnotesize
        \textit{Notes:} The panels report rejection frequencies for $T=500$, $\rho=0.50$, and $\kappa=-0.50$. The DGP is the baseline SVAR(4). In a given cell, all three structural shocks are mutually independent and have mean zero and variance one, with the skewness displayed on the vertical axis, and the excess kurtosis displayed on the horizontal axis. 
    \end{minipage}
\end{figure}

As described in the previous section, the power gains come from exploiting the non-Gaussianity of the shocks.
Figure~\ref{fig:mc-nongaussianity-heatmaps} presents systematically how the distribution of
the structural shocks affects the power properties of the tests. For this, we set $T=500$, $\rho=0.50$, and $\kappa=-0.50$ 
(i.e., $z_t = 0.50\varepsilon_{1t} - 0.25\varepsilon_{2t} + \sqrt{0.6875}v_t$),
and vary the skewness and excess kurtosis of the shocks. Moment combinations that do not correspond to a valid distribution are omitted from the figure.

The ISP-based test is more powerful than the SP-based test in every configuration.
For symmetrically distributed shocks, increasing excess kurtosis raises ISP rejection from approximately $0.14$ to $0.57$,
whereas SP rejection increases only from approximately $0.12$ to $0.24$. In this case,
the quadratic ISP has no local identifying content because $\delta_2=\E[\varepsilon_{1t}^3]=0$, but the cubic ISP exploits
$\delta_3=\E[\varepsilon_{1t}^4]-3$, which becomes more informative as excess kurtosis increases.
Once $\E[\varepsilon_{1t}^3]\neq 0$, the quadratic ISP supplies a direct source of identifying information,
and its rejection rises (to unity) rapidly. The improvement is much smaller for the SP test as the quadratic transformation of the external proxy also transforms its measurement noise. 
Nevertheless, once skewness is already substantial, increasing excess kurtosis may eventually reduce rejection frequencies in finite samples. Although excess kurtosis supplies additional identifying information through $\delta_3$, the finite-sample power of the test also depends on whether the test statistic can be reliably estimated. As mentioned in Section~\ref{sec:theory},
heavier tails increase the sampling variation of higher-order moments and can create outliers, making the optimal GMM estimator and the test statistic imprecise.

\begin{figure}[ht]
    \centering
    \caption{Standard and identification-robust exogeneity tests in the baseline design}
    \label{fig:mc-main-standard-weak-rho-kappa}
    \includegraphics[width=\textwidth]{\detokenize{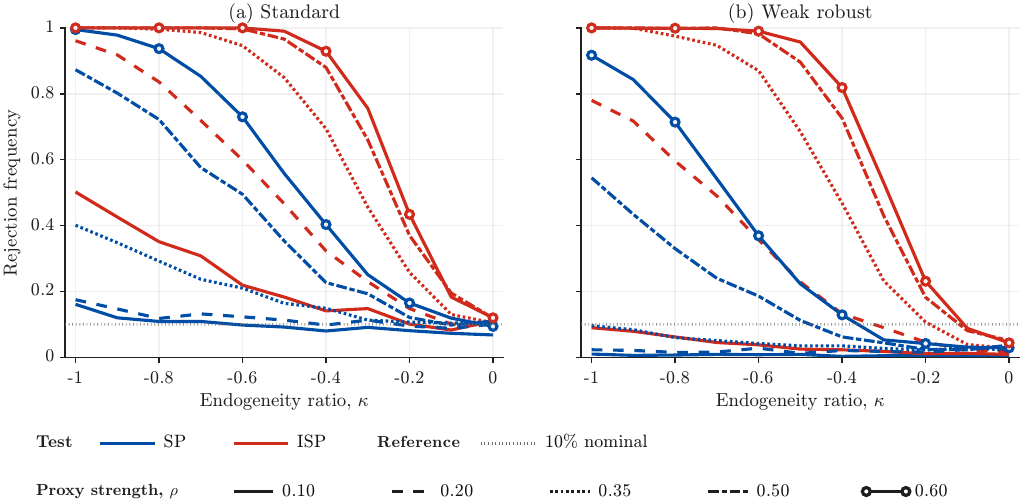}}
    \vspace{0.35em}
    \begin{minipage}{0.98\textwidth}
        \footnotesize The DGP is the baseline SVAR(4) with the baseline proxy.  Panel (a) reports the standard
        SP and ISP $J$-tests; panel (b) reports their identification-robust candidate-set
        counterparts. 
    \end{minipage}
\end{figure}

\subsection{Identification-robust tests and weak proxies}

We next compare the standard procedures with the identification-robust tests from
Section~\ref{sec:SP_weak} and Appendix~\ref{app:weakSP}. The identification-robust candidate sets are obtained by evaluating the objective function over a grid in the parameter space $\mathcal{B}$ and the tests reject at level $\alpha$ in case that 
no candidate value of $\beta$ is compatible with the moment restrictions at the specified level, s.t., 
$\mc C_T(1-\alpha) = \varnothing$ and $\mc C_T^*(1-\alpha)= \varnothing$, respectively.   

Figure~\ref{fig:mc-main-standard-weak-rho-kappa} compares the standard and robust procedures for
$T=500$ in the baseline design. Under $\kappa=0$, the robust tests reject less often
than the nominal level, as expected from the asymptotic conservative property established in
Proposition~\ref{thm:weak_test_ISP}. Under the alternative, when the proxy is strongly relevant, the candidate sets contract and the robust tests retain considerable power, with
the robust ISP-based test continuing to dominate its SP counterpart. For weakly and moderately
relevant proxies, however, the identification-robust tests reject much less often than their standard counterparts. 
Hence, although the identification-robust tests have bounded type-I error rate, this comes at a potentially substantial power cost when the data contain insufficient information to rule out every candidate value of $\beta$.

\begin{figure}[ht]
    \centering
    \caption{Standard and identification-robust exogeneity tests under weak-proxy and Gaussianity}
    \label{fig:mc-weak-proxy-comparison}
    \includegraphics[width=\textwidth]{\detokenize{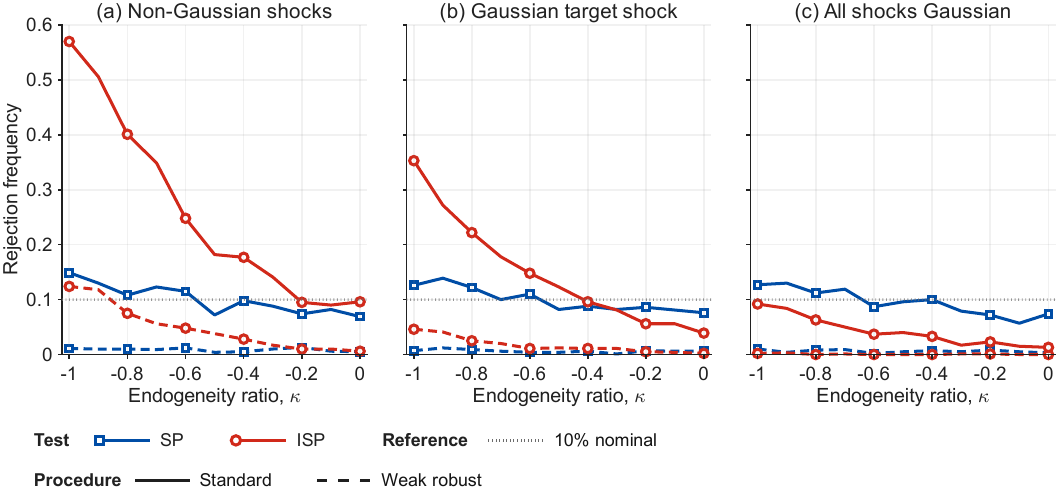}}
    \vspace{0.35em}
    \begin{minipage}{0.98\textwidth}
        \footnotesize
        \textit{Notes:} The panels report rejection frequencies for $T=500$ and target first-stage strength $F=6$ with $\rho\approx0.109$. In panel~(a), all structural shocks follow the baseline non-Gaussian specification, in panel~(b) the target shock is Gaussian, and in panel~(c) all shocks are Gaussian. Solid and dashed lines report the standard and identification-robust tests, respectively.
    \end{minipage}
\end{figure}

To examine the test properties for a weak proxy, we calibrate the correlation between the proxy and the target shock as $\rho=\sqrt{6/504}\approx0.109$, which corresponds to $F_T(\rho)=6$ with $T=500$.
Panel~(a) of Figure~\ref{fig:mc-weak-proxy-comparison} illustrates that a weak proxy does not imply weak identification for the proposed ISP-based test. The non-Gaussianity of the target shock generates
nonzero Stein deviations, so that the quadratic and cubic ISPs continue to supply identifying information. Accordingly, the standard ISP test has substantial power against
proxy contamination despite the low first-stage $F$ statistic. 
However, when the target shock is Gaussian, its skewness and excess kurtosis vanish, and the ISP moments no longer provide the additional local identifying information, so that the power falls accordingly. Finally, when all shocks are Gaussian in panel~(c), rejection frequencies of the standard procedures remain close to or below the nominal level and display little systematic response to $\kappa$, while the robust tests almost never reject. This is the worst case scenario, in which neither the external proxy nor the distributional feature of any shock inside the VAR model can supply the test with identifying content.

\subsection{Model misspecification}
As discussed in Section~\ref{sec:intuition}, the validity of ISPs depend on a correctly specified SVAR. Once the chosen SVAR model is not appropriate to represent the DGP, the asymptotic null distribution of the strong exogeneity test is no longer valid. To illustrate this, we first consider a system in which the true transmission mechanism is nonlinear, 
\[
y_t - \sum_{i=1}^{4}A_i y_{t-i} \equiv u_t =
 B_0\varepsilon_t + \gamma\iota_2 H_2(\varepsilon_{1t}),
 \qquad
 H_2(x)=x^2-1,
 \quad \text{with }\
 \gamma\in\{0.1,0.2\},
\]
which is omitted in a linear VAR representation. 
The econometrician nevertheless estimates a linear SVAR(4) with the correct lag specification. 
Figure~\ref{fig:mc-misspec-nonlinearity-rho-kappa} reports the resulting rejection frequencies.
At $\kappa=0$, the proxy is exogenous with respect to the non-target shocks, but the
structural shocks implied by the misspecified linear representation do not satisfy the
mean-independence restrictions.
To see this, evaluating the linear decomposition at $B_0$ yields
$e_t:=B_0^{-1}u_t
=\varepsilon_t+\gamma B_0^{-1}\iota_2H_2(\varepsilon_{1t})
=\bigl(\varepsilon_{1t},
\varepsilon_{2t}+\gamma H_2(\varepsilon_{1t}),
\varepsilon_{3t}\bigr)'$, such that the pseudo non-target shocks can be predicted by the target shock in mean $\E[e_{2t}\mid e_{1t}]=\gamma H_2(e_{1t})\neq0$.
As a result,  both tests reject substantially more often than the nominal level,
particularly for higher $\rho$ when the proxy is relevant. 
For example, with $\rho=0.60$, rejection frequencies of the SP-based tests are about $0.17$ and $0.32$ for $\gamma=0.1$ and $\gamma=0.2$, respectively.
The ISPs are more sensitive to the higher-order dependence generated by the omitted nonlinear 
relationship between the target and non-target shocks, so that the tests have higher rejection rates
of nearly $0.6$ and $0.95$, for $\gamma=0.1$ and $\gamma=0.2$, respectively.

\begin{figure}[ht]
    \centering
    \caption{Rejection frequencies under omitted nonlinear transmission}
    \label{fig:mc-misspec-nonlinearity-rho-kappa}
    \includegraphics[width=\textwidth]{\detokenize{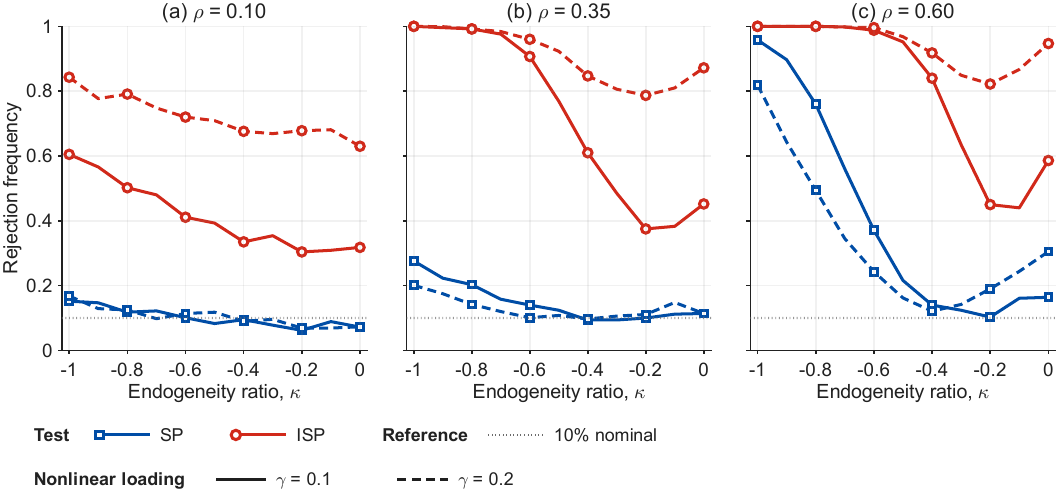}}
    \vspace{0.35em}
    \begin{minipage}{0.98\textwidth}
        \footnotesize
        \textit{Notes:} The panels report  rejection frequencies
        for $T=500$. The DGP is the baseline
        SVAR(4) with an additional nonlinear term such that
        $u_{2t}=\varepsilon_{1t}+\varepsilon_{2t}
        +\gamma(\varepsilon_{1t}^2-1)$.
    \end{minipage}
\end{figure}

Note that in this case, the structural impact multipliers of the shocks have to be interpreted with caution as the true effect may depend on the nuisance parameter describing the nonlinear relationship between the shocks. For example, the impact multiplier of the target shock on $y_{2t}$ becomes
$\E[y_{2t}\mid\varepsilon_{1t}=1,y_{s,s<t}] - \E[y_{2t}\mid\varepsilon_{1t}=0,y_{s,s<t}] = 1+\gamma$.
Furthermore, the standard proxy SVAR approach does not recover the correct structural parameter even asymptotically.\footnote{Note that $\E[z_tu_{1t}]=\rho$, while $\E[z_tu_{2t}]
=\rho(1+\kappa)
+\rho\gamma\E[\varepsilon_{1t}(\varepsilon_{1t}^2-1)]
=\rho(1+\kappa-\gamma)$ as $\E[\varepsilon_{1t}^3]=-1$ in the baseline design. The ratio of two population moments gives $1-\gamma+\kappa$.}
In fact, conditional on the misspecified model, the proxy is endogenous as it correlated with a non-target pseudo shock $e_{2t}$. Hence, the rejection of the proposed strong exogeneity test may yield an important signal for the researcher and prompt her to reflect on the model specification.

\begin{figure}[ht]
    \centering
    \caption{Rejection frequencies under lag-misspecification}
    \label{fig:mc-misspec-lag-order-rho-kappa-var1}
    \includegraphics[width=\textwidth]{\detokenize{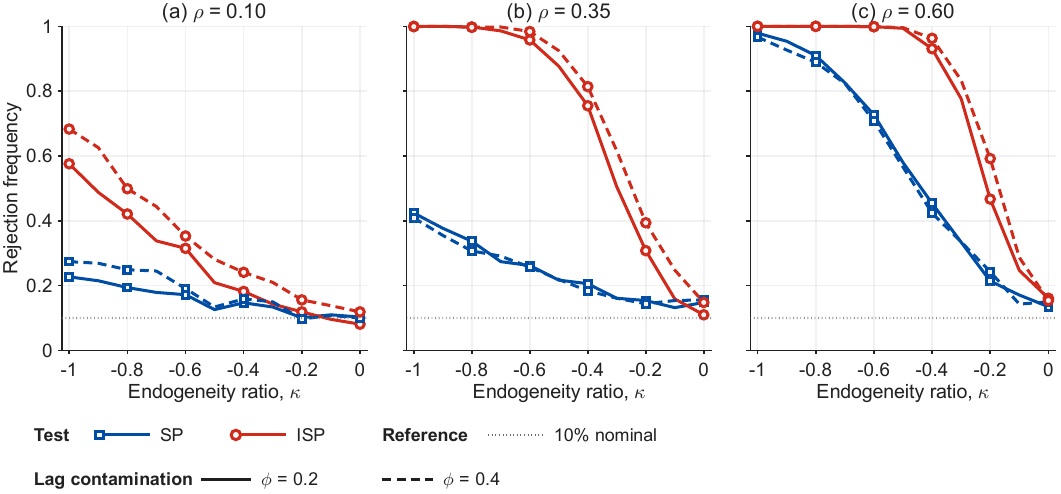}}
    \vspace{0.35em}
    \begin{minipage}{0.98\textwidth}
        \footnotesize
        \textit{Notes:} The panels report rejection frequencies
        for $T=500$. Data are generated by the baseline SVAR(4), but the econometrician estimates a VAR(1). The proxy is
        $z_t=\rho\varepsilon_{1t}+\kappa\rho\varepsilon_{2t}
        +\phi\varepsilon_{1,t-1}
        +\sqrt{1-\rho^2(1+\kappa^2)-\phi^2}\,v_t$.

    \end{minipage}
\end{figure}

In the second experiment, we consider the case in which the data continue to be generated by the baseline SVAR(4), but the researcher estimates a VAR(1). We also allow the proxy to depend  on the lagged target shock
\[
z_t
 = \rho\varepsilon_{1t} + \kappa\rho\varepsilon_{2t} +
 \phi\varepsilon_{1,t-1} + \sqrt{1-\rho^2(1+\kappa^2)-\phi^2}\, v_t
 \qquad
 \text{with} \
 \phi\in\{0.2,0.4\}.
\]
Figure~\ref{fig:mc-misspec-lag-order-rho-kappa-var1} displays rejection frequencies under lag-misspecification. The proposed ISP-based test has substantially greater power than the SP test to correctly detect null violations.  However, at $\kappa=0$, the proxy 
$z_t$  remains strongly exogenous with respect to the current non-target shocks, while 
both tests exhibit type-I error rates slightly above the $10\%$ level for the more relevant proxies. 
Compared to model misspecification due to omitted nonlinearity, however, the oversize issue is rather mild.
In the higher-relevance calibrations ($\rho = 0.35, 0.6$), the type-I error rates of both tests stay approximately within the range 
$11\%$--$16\%$. Note that the residual from the misspecified VAR(1) model still contains higher-order lagged terms of $y_t$, which in turn contain lagged structural shocks.
Since the proxy also loads on $\varepsilon_{1,t-1}$, it is correlated with components of the misspecified residual. The tests treat this correlation a a violation of the null although it originates from omitted dynamics.

\section{Application}\label{sec:app}

Proxy variables constructed from high-frequency financial market data are widely used to identify structural shocks in macroeconomic SVARs. The high-frequency identification approach builds on the insight that asset price movements in narrow event windows around policy announcements are driven predominantly by a single structural shock, making the high-frequency variation a valid external instrument for that shock in the low-frequency SVAR.
Prominent applications include monetary policy shocks identified via interest rate movements around FOMC announcements \citep{gertler2015monetary,miranda2021transmission,nakamura2018high,jarocinski2020deconstructing} and oil supply news shocks identified via oil futures price changes around OPEC announcements \citep{kanzig2021macroeconomic}. This section examines the benchmark oil supply news and monetary policy applications in detail, while Appendix~A.4 applies the tests to additional proxies from a broader range of the existing proxy SVAR literature.

Recent work emphasizes that policy announcements may convey information about multiple structural shocks. For example, FOMC announcements may convey multiple monetary-policy shocks, for instance, surprises to the target rate, forward-guidance shocks, and information shocks \citep{gurkaynak2005actions,cieslak2019non,jarocinski2020deconstructing,swanson2021measuring,bauer2023reassessment,jarocinski2024estimating,lewis2025announcement}.
Similarly, OPEC announcements may contain both supply-news and demand-news components \citep{degasperi2025}. 
Conceptually, a proper separation of high-frequency surprises into multiple latent components requires a system with multiple high-frequency proxies (at least as many as the number of components one aims to recover) together with additional identifying restrictions. In many applications, however, researchers have access to only one high-frequency proxy series that is deemed credible for a given shock.
For example, \cite{gertler2015monetary} work with the short-term interest rate surprise in narrow FOMC windows \citep{gertler2015monetary}, \cite{corsetti2021} use a proxy based on  overnight interest rate swaps for the Euro Area, and \cite{kanzig2021macroeconomic} constructs an OPEC-based oil supply news proxy. In the empirically relevant single-proxy scenario, the researcher faces a stark trade-off: either (i) use the full proxy series as an external instrument for one target shock, accepting the risk that it may embed additional announcement components, or (ii) attempt to split the proxy into multiple components using feasible restrictions based on observable co-movement patterns. A prominent example is the so-called ``poor man's sign restriction''\footnote{The expression ``poor man's sign restrictions'' is used in \citet{jarocinski2020deconstructing} to describe the binary classification that assigns events to one of two shocks based on the sign of equity-price responses.}, which has been applied in the monetary policy context \citep{jarocinski2020deconstructing} and in oil-market applications \citep{degasperi2025, Gazzani_etal_2024}.\footnote{Of course, one could properly separate different announcement components if (a) the researcher collects high-frequency surprises of multiple assets during the announcement window (at least of the same number as the announcement components she intends to identify), and (b) if the researcher imposes credible identifying restrictions on these instruments \citep[see, e.g.,][]{jarocinski2020deconstructing, swanson2021measuring, kroencke2021fomc, jarocinski2024estimating}.
}

Our proposed exogeneity tests provide statistical evidence that speaks directly to whether a single-proxy strategy is defensible. Rejection of the test indicates that the proxy contains information about the conditional mean of at least one non-target structural shock, suggesting that a separation step is warranted. 
In both applications below, we find no evidence against exogeneity of the benchmark high-frequency proxies. This supports the original use of single-proxy approach in \cite{gertler2015monetary} and \cite{kanzig2021macroeconomic}. Importantly, non-rejection does not imply that additional high-frequency shocks are absent, but rather that they do not materially affect the chosen proxy or the low-frequency SVAR responses. This interpretation is consistent with recent evidence in \cite{jarocinski2024estimating}, who finds that short-term interest rates are driven predominantly by the standard monetary policy shock, while other shocks primarily load on proxies based on long-term rates and stock prices. The weighting matrices used in the tests are implemented as detailed in Appendix~\ref{app:weighting_est}.

\subsection{Oil Supply News Shocks}\label{sec:app-oil}

We replicate the baseline SVAR specification in \cite{kanzig2021macroeconomic} using monthly data from 1974 to 2017. The VAR includes six variables: the real oil price, world oil production, world oil inventories, world industrial production, U.S.\ industrial production, and U.S.\ CPI. We estimate a reduced-form VAR with 12 lags and collect residuals $\hat{u}_t$. The proxy is the first principal component of oil futures price changes (spanning the first year of the term structure) around OPEC announcements. 

\cite{kanzig2021macroeconomic} uses variation in oil futures prices around OPEC announcements as a proxy to identify oil supply news shocks. The identifying narrative is that OPEC announcements reveal information primarily about future oil supply, so that futures-price movements within the event window are driven by supply news and are orthogonal to other structural shocks in the monthly SVAR. \cite{degasperi2025} challenges this assumption, arguing that in imperfect information environments, market participants also revise their beliefs about demand conditions during OPEC announcements. In this case, the proxy in \cite{kanzig2021macroeconomic} would be correlated with a non-target shock and fail to be exogenous.

If the proxy captures both supply and demand shocks, the estimates in \cite{kanzig2021macroeconomic} would be biased.  \cite{degasperi2025} proposes separating the proxy into supply and demand components using sign restrictions, finding larger supply shock effects than \cite{kanzig2021macroeconomic}. Conversely, if the proxy is exogenous, this separation would introduce selection bias and the original single-proxy estimates are more reliable.

\begin{table}[htbp]
\centering
\caption{Proxy exogeneity tests: Oil supply news shocks}
\label{tab:kaenzig-exog}
\begin{tabular}{llcc}
\toprule
Test & Version & Statistic & $p$-value \\
\midrule
ISP & Standard & 8.19 & 0.611 \\
SP & Standard & 9.56 & 0.480 \\
ISP & Weak robust & 8.30 & 0.911 \\
SP & Weak robust & 8.68 & 0.893 \\
\bottomrule
\end{tabular}

\end{table}
We apply the proposed SP and ISP tests (both the standard and identification-robust versions) to the \citet{kanzig2021macroeconomic} proxy, using moments up to $L=3$. As shown in Table~\ref{tab:kaenzig-exog}, neither test rejects proxy exogeneity at conventional significance levels. The proxy exhibits moderate non-Gaussianity (skewness $= -0.19$, excess kurtosis $= 2.49$), indicating that the test has power to detect exogeneity violations. This result supports the identification scheme in \cite{kanzig2021macroeconomic}, i.e., any demand-news component in the announcement-window futures variation -- if present -- does not generate a detectable violation of the conditional-mean exogeneity restrictions with respect to the remaining monthly SVAR shocks and has at most a negligible effect on oil futures price variation around OPEC announcements.

\subsection{Monetary Policy Shocks}\label{sec:app-mp}

We replicate the baseline SVAR specification in \cite{jarocinski2020deconstructing} using monthly U.S.\ data from 1984 to 2016. The VAR includes five variables: the one-year Treasury yield, log S\&P 500 index, log real GDP, log GDP deflator, and the excess bond premium. We estimate a reduced-form VAR with 12 lags and collect residuals $\hat{u}_t$. The proxy is the high-frequency change in the three-month-ahead Fed Funds futures rate around FOMC announcements. Data are from the \cite{jarocinski2020deconstructing} replication package.

High-frequency short-rate surprises around central bank announcements are a standard proxy for monetary policy shocks (\citealt{gertler2015monetary}, see also \citealt{corsetti2021} for evidence in the Euro Area). \citet{jarocinski2020deconstructing} argue that announcements may also convey central bank information shocks, i.e., signals about the Fed's assessment of economic conditions, which can induce systematic co-movements between short-rate surprises and equity prices. They propose separating policy and information components using sign and co-movement restrictions in the high-frequency system.

As in the oil application, the empirical issue is whether the benchmark short-rate proxy can be used to identify a monetary policy shock in the SVAR. If information shocks materially load on the short-rate proxy and have different implications for macroeconomic responses, then a single-proxy strategy can confound the two components. Conversely, if short-rate surprises are predominantly driven by the conventional policy shock, separation schemes that condition on realized co-movements, especially in binary ``poor man's'' implementations, can discard signal and induce selection effects.

\begin{table}[htbp]
\centering
\caption{Proxy exogeneity tests: Monetary policy shocks}
\label{tab:jk-exog}
\begin{tabular}{llcc}
\toprule
Test & Version & Statistic & $p$-value \\
\midrule
ISP & Standard & 5.22 & 0.734 \\
SP & Standard & 9.51 & 0.301 \\
ISP & Weak robust & 5.34 & 0.946 \\
SP & Weak robust & 9.25 & 0.682 \\
\bottomrule
\end{tabular}

\end{table}

We apply the proposed tests to the high-frequency short-rate proxy, using moments up to $L=3$. 
As documented in Table~\ref{tab:jk-exog}, the tests fail to reject the null. The proxy exhibits pronounced non-Gaussianity (skewness $= -1.73$, excess kurtosis $= 5.02$), suggesting that the test has power to detect exogeneity violations.  The empirical evidence therefore supports the view that central bank information shocks, while potentially present in broader high-frequency systems, do not induce a detectable violation of conditional-mean exogeneity for this benchmark short-rate proxy with respect to the remaining monthly SVAR shocks. Accordingly, the single-proxy identification of monetary policy shocks appears empirically defensible in this setting. The information shocks, if present, have at most a small effect on short-term interest rate surprises around FOMC announcements. The recent study by \cite{jarocinski2024estimating} provides empirical evidence supporting this finding. 

\section{Conclusion}\label{sec:fazit}

This paper develops a framework for testing strong exogeneity in proxy SVARs when researchers have access to only one credible proxy for a target shock.
Building on the notion of strong exogeneity, this paper introduces the approach of internal synthetic proxies (ISPs). The approach synthesizes over-identifying restrictions from a noise-free proxy constructed internally from reduced-form information and obtains test power by exploiting non-Gaussianity of the target shock.
We provide weak identification theory for the test and show that in case of weak proxy, the non-Gaussianity of the target shock can sustain ISP identification. The identifying content can be summarized by a novel measure of identification strength based on deviations from the Stein identity, which gives a direct and interpretable moment-based quantification of non-Gaussianity. We study the test's behavior when the non-Gaussian information vanishes in the limit. 
Our MC results show that the proposed ISP-based strong exogeneity tests deliver substantial power gains compared to the SP-based strong proxy exogeneity tests of \cite{BK2024}, especially when the proxy has a low signal-to-noise ratio. In empirical applications to high-frequency proxy SVARs, the proposed tests provide direct evidence on whether a single-proxy strategy is defensible. In the benchmark monetary policy and oil market applications, we find no evidence against exogeneity of the standard proxies.

\section{Declaration of generative AI-assisted technologies.}
The authors used OpenAI Codex to assist with programming  and creation of figures. The authors reviewed and edited the code as needed and take full responsibility for the content.

 \bibliography{literatur}

\clearpage
\appendix  
\section{Appendix}
\subsection{Technical assumptions}\label{app:tech_ass}

\begin{assumptionp}{T1}\label{Tass:base}
\begin{enumerate}[label=(\roman*)]
\item The VAR($p$) in \eqref{eq:SVAR} is stable s.t., the lag polynomial $A(z) = I_n - \sum^p_{i = 1} A_i z^i$ satisfies $\det A(z) \neq 0$ for all $|z|\leq 1$.

\item Matrix $B_0$ is non-singular and $\{(\varepsilon_t, z_t)\}_{t\in\N}$ is an i.i.d.\ sequence with $\E[\varepsilon_t]=0$ and $\E[\varepsilon_t\varepsilon_t']=\Sigma_\varepsilon$ diagonal and positive definite. Moreover, there exists $\delta>0$ such that $\E\|\varepsilon_t\|^{4+\delta}<\infty$.

\end{enumerate}
\end{assumptionp}

\begin{assumptionp}{T2}\label{Tass:ISP}
\begin{enumerate}[label=(\roman*)]
\item The true values $(\pi_0,\beta_0,\omega_0)$ lie in the (relative) interior of the corresponding parameter spaces, which are compact subspaces in $\R^{d_\pi}$, $\R^{n-1}$ and $\R^{n^2}$, respectively.

\item The map $\omega\mapsto \Sigma_u(\omega)$ is continuously differentiable and $\Sigma_u(\omega)$ is finite and symmetric positive definite for all $\omega$. 

\item There exists a symmetric positive definite matrix $W\in\R^{L(n-1)\times L(n-1)}$ such that $\|\hat W_T-W\|=o_p(1)$.

\item The map $(\beta,\pi,\omega)\mapsto g_t^{(Z^*)}(\beta,\pi,\omega)$ is almost surely continuously differentiable in a neighbourhood $\mathcal{N}^*$ of $(\beta_0,\pi_0,\omega_0)$ and
\[
\E\sup_{(\beta,\pi,\omega)\in\mathcal{N}^*}\left\|\frac{\partial g_t^{(Z^*)}(\beta,\pi,\omega)}{\partial\pi'}\right\|<\infty,\quad
\E\sup_{(\beta,\pi,\omega)\in\mathcal{N}^*}\left\|\frac{\partial g_t^{(Z^*)}(\beta,\pi,\omega)}{\partial\beta'}\right\|<\infty,
\]
\[
\E\sup_{(\beta,\pi,\omega)\in\mathcal{N}^*}\left\|\frac{\partial g_t^{(Z^*)}(\beta,\pi,\omega)}{\partial\omega'}\right\|<\infty.
\]
Moreover, $\E\|g_t^{(Z^*)}(\beta_0,\pi_0,\omega_0)\|^{2+\delta}<\infty$ for some $\delta>0$.


\item $\rank(G_\beta^*)=n-1$.
\end{enumerate}
\end{assumptionp}

Note that since $\Sigma_u$ is symmetric, $\omega_0$ does not live in an unrestricted $n^2$-dimensional interior.
For notational convenience, derivatives with respect to $\omega$ are
represented by the ambient derivatives of the corresponding smooth matrix
functions evaluated at $\Sigma_{u}(\omega_0) \equiv \Sigma_u$.\footnote{An equivalent formulation can be $\operatorname{vech}(\Sigma_u)$, which yields the same results, but would require introducing duplication and elimination matrices throughout the derivative calculations and the
proofs.}  
Assumption~\ref{Tass:ISP}(ii) together with compactness in \ref{Tass:ISP}(i) implies that
$\sup_{\omega}\|\Sigma_u(\omega)\|<\infty$ and $\sup_{\omega}\|\Sigma_u(\omega)^{-1}\|<\infty$ so that the internal proxy is well-defined and smooth uniformly. Finally, for polynomial ISPs up to degree $L$, the moment condition
$\E\|g_t^{(Z^*)}(\beta_0,\pi_0,\omega_0)\|^{2+\delta}<\infty$ in \ref{Tass:ISP}(iv) is satisfied, for example,
when $\E|z_t|^{(2+\delta)(L+1)/L}<\infty$ and $\E\|\varepsilon_t\|^{(2+\delta)(L+1)}<\infty$.\footnote{Note that for the first block, we can bound $\E[|z_t|^{2+\delta}\|\bs\varepsilon_{2t}\|^{2+\delta}]$ by $\left(\E|z_t|^{(2+\delta)(L+1)/L} \right)^{L/(L+1)} \left( \E\|\bs\varepsilon_{2t}\|^{(2+\delta)(L+1)} \right)^{1/(L+1)}$ by the H\"older's inequality.}
Again, a sufficient condition for Assumption~\ref{Tass:ISP}(v) is the relevance condition of the external proxy, i.e., $\rho_1 \neq 0$.

\begin{assumptionp}{T2-W}\label{Tass:ISPweak}
\begin{enumerate}[label=(\roman*)]
\item The true values \((\pi_0,\beta_0,\omega_0)\) lie in the (relative) interiors of
their respective parameter spaces. In particular,
\(\beta_0\in\operatorname{int}(\mc B)\), where \(\mc B\subset\R^{n-1}\)
is compact and convex. The map $\omega\mapsto \Sigma_u(\omega)$ is continuously differentiable and $\Sigma_u(\omega)$ is finite and symmetric positive definite for all $\omega$. 
\item $\mc G_T^*(\beta)\Rightarrow\mc G^*(\beta)$ as a random function on \(\mc B\) with respect to the sup norm, where
\(\mc G^*(\beta)\) is a zero-mean Gaussian process with a.s.\ continuous sample paths and covariance function
$
\Omega^*(\beta_1,\beta_2) := \E\!\left[ \mc G^*(\beta_1)\mc G^*(\beta_2)' \right].$ 
The map
\(\beta\mapsto\Omega^*(\beta):=\Omega^*(\beta,\beta)\) is continuous and
uniformly positive definite on \(\mc B\).

\item The map
\((\beta,\pi,\omega)\mapsto
g_t^{(Z^*)}(\beta,\pi,\omega)\) is almost surely continuously
differentiable on
\(\mc B\times\mc N_\pi\times\mc N_\omega\), where
\(\mc N_\pi\) and \(\mc N_\omega\) are neighbourhoods of
\(\pi_0\) and \(\omega_0\), and
$\sup_T\E_T$
$\sup_{\substack{\beta\in\mc B,\,
\pi\in\mc N_\pi,\,
\omega\in\mc N_\omega}}
\left(
\left\|
\frac{\partial g_t^{(Z^*)}(\beta,\pi,\omega)}
{\partial\pi'}
\right\|
+
\left\|
\frac{\partial g_t^{(Z^*)}(\beta,\pi,\omega)}
{\partial\omega'}
\right\|
+
\left\|
\frac{\partial g_t^{(Z^*)}(\beta,\pi,\omega)}
{\partial\beta'}
\right\|
\right)
<\infty.$

\item There exists a feasible estimator
\(\hat\Omega_T^*(\beta)\), s.t., $
\sup_{\beta\in\mc B}
\left\|
\hat\Omega_T^*(\beta)-\Omega^*(\beta)
\right\|
=o_p(1).$
\end{enumerate}
\end{assumptionp}
Assumption~\ref{Tass:ISPweak}(ii) is the analogue of the functional central limit condition in Assumption~\ref{Tass:SPweak}(ii). Assumption~\ref{Tass:ISPweak}(iii) allows for differentiation under the expectation so that $G_{\beta, T}^*  = \frac{\partial \bs \psi_T^*(\beta_0)}{\partial\beta'}$. It also implies  
\begin{align*}
&\sup_{\beta\in\mc B}
\Bigg\|
\sqrt T
\left[
\bar g_T^{(Z^*)}(\beta,\hat\pi_T,\hat\omega_T)
-\bar g_T^{(Z^*)}(\beta,\pi_0,\omega_0)
\right]
-G_{\pi,T}^*(\beta)\sqrt T(\hat\pi_T-\pi_0)
-G_{\omega,T}^*(\beta)\sqrt T(\hat\omega_T-\omega_0)
\Bigg\|
=o_p(1).
\end{align*}

\subsection{Proofs of the main results and lemmata}\label{app:proof}

\begin{proof}[Proof of lemma~\ref{lem:weak3_glb}]
For $\beta\in\mc B$, define the line segment 
$\beta(s):=\beta_0+s(\beta-\beta_0)$ for $0\leq s\leq1$.
Since $\mc B$ is convex, $\beta(s) \in \mc B$, on which $\bs\psi_T^*$ is continuously differentiable,
the fundamental theorem of calculus (FTC) and 
the chain rule applied to $s \mapsto \bs\psi_T^*(\beta(s))$ 
entail that 
\begin{align*}
\sqrt T \bs\psi_T^*(\beta)=	\sqrt T \left[\bs\psi_T^*(\beta(1))-\bs\psi_T^*(\beta(0)) \right] &=
\int_0^1
\sqrt T\frac{\partial\bs\psi_T^*(\beta_0+s(\beta-\beta_0))}
{\partial\beta'} (\beta-\beta_0)\,\dd s 
\end{align*}
since $\bs\psi_T^*(\beta(0)) = \bs\psi_T^*(\beta_0) = 0$ under 
Assumption~\ref{ass:strong_exogeneity}.
Define 
\[
\bs \Psi^*(\beta) := \int_0^1
\bs G^*(\beta_0 + s (\beta - \beta_0)) (\beta-\beta_0)\,\dd s,
\]
and note that $\bs \Psi^*(\beta_0) = \bs 0$ by construction.
By Assumption~\ref{ass:weak3} and compactness of $\mc B$,
\begin{align*}
	\sup_{\beta\in\mc B}
\left\|
\sqrt T\,\bs\psi_T^*(\beta)-\bs\Psi^*(\beta)
\right\|
&= \sup_{\beta\in\mc B}
\left\|
\int_0^1 
\left( \sqrt T\frac{\partial\bs\psi_T^*(\beta(s))}
{\partial\beta'}
 - \bs G^*(\beta(s))  \right)(\beta-\beta_0) \,\dd s
\right\| \\
& \leq 
\sup_{\beta\in\mc B}
\left\| \beta-\beta_0 \right\| 
\sup_{s \in [0,1]}
\left\|  \sqrt T\frac{\partial\bs\psi_T^*(\beta(s))}
{\partial\beta'}  -  \bs G^*(\beta(s)) \right\|  \int_0^1 \,\dd s \\
& \leq 
\sup_{\beta\in\mc B}
\left\| \beta-\beta_0 \right\| 
\sup_{\tilde \beta\in\mc B}
\left\|  \sqrt T\frac{\partial\bs\psi_T^*(\tilde \beta)}
{\partial\beta'}  -  \bs G^*(\tilde \beta) \right\| \longrightarrow 0.
\end{align*}
Note that the continuity of $\bs G^*$ extends to $\bs\Psi^*$.
To show its derivative, we start by showing that $\bs\Psi^*$ is differentiable. 
For this, define for $\beta, \beta + \Delta \beta \in \mc B$ the line segment 
$\beta + s(\beta + \Delta \beta - \beta ) = \beta + s \Delta \beta$ with $0\leq s\leq1$. Then, FTC entails that
\[
\sqrt T \left[\bs\psi_T^*(\beta + \Delta  \beta)-\bs\psi_T^*(\beta) \right] =
\int_0^1
\sqrt T\frac{\partial\bs\psi_T^*(\beta + s \Delta \beta)}
{\partial\beta'} \Delta \beta\,\dd s.
\]

By \eqref{eq:ISP_global_drift} on the left-hand side and by \eqref{eq:W3_derivative_uniform} for the right-hand side, 
\begin{align*}
	\bs\Psi^*(\beta + \Delta \beta)- \bs\Psi^*(\beta) 
	&= \int_0^1 \bs G^*(\beta + s \Delta \beta)  \Delta \beta\,\dd s. \\
\bs\Psi^*(\beta + \Delta \beta)- \bs\Psi^*(\beta) - \bs G^*(\beta) \Delta \beta
&= 	\int_0^1 \left( \bs G^*(\beta + s \Delta \beta) - \bs G^*(\beta) \right)   \Delta \beta\,\dd s.
\end{align*}
Hence,
\begin{align*}
\left\|
\bs\Psi^*(\beta + \Delta \beta)- \bs\Psi^*(\beta) - \bs G^*(\beta) \Delta \beta
\right \|	
& \leq 
\sup_{s \in [0,1]}
\left\|
 \bs G^*(\beta + s \Delta \beta) - \bs G^*(\beta) 
\right \|	\left\|\Delta \beta \right \|
 = o(\left\|\Delta \beta \right \|),
\end{align*}
where the last equality follows by continuity of $\bs G^*$ under 
Assumption~\ref{ass:weak3}. 
Hence, $\frac{\partial\bs \Psi^*(\beta)}{\partial\beta'}
=\bs G^*(\beta)$, $\forall \beta\in\operatorname{int}(\mc B)$.
Evaluating $\beta = \beta_0$, by \ref{Tass:ISPweak}(iii) Proposition~\ref{thm:G_beta_weak} entails 
\[
G_{\beta, T}^* 
= \frac{\partial \bs \psi_T^*(\beta_0)}{\partial\beta'}
= -\begin{pmatrix}
    \rho_{1, T}\\
	\sigma_1\bs\delta_T^{(L)}
    \end{pmatrix}\otimes I_{n-1}.
\]
Hence, 
$\sqrt{T}\frac{\partial \bs \psi_T^*(\beta_0)}{\partial\beta'}
= -\sqrt{T}\begin{pmatrix}
    \rho_{1, T}\\
	\sigma_1\bs\delta_T^{(L)}
    \end{pmatrix}\otimes I_{n-1}
\longrightarrow \bs G^*(\beta_0) = \frac{\partial\bs \Psi^*(\beta_0)}{\partial\beta'}.
$ 
Since $\left\{x\otimes I_{n-1}:x\in\R^L\right\}$ is a closed linear subspace of $\R^{L(n-1)\times(n-1)}$ and the map $x\mapsto x\otimes I_{n-1}$ is one-to-one, there exists a unique $\bs\gamma^*\in\R^L$, s.t., $\bs G^*(\beta_0)=-\bs\gamma^*\otimes I_{n-1}$. This verifies the validity of condition~\eqref{eq:W3_local} and completes the proof.
\end{proof}

\begin{proof}[Proof of Proposition~\ref{thm:beta_ISP}]
We first verify $\E[g_t^{(Z^*)}(\beta_0,\pi_0,\omega_0)]=0$.
By Lemma~\ref{lem:ISP_ass}, $\E[\boldsymbol\varepsilon_{2t}\mid z_t]=\mathbf 0$ and $\E[\boldsymbol\varepsilon_{2t}\mid \varepsilon_{1t}]=\mathbf 0$.
At $(\beta_0,\pi_0)$, we have $r_t(\beta_0,\pi_0)=(\mc J-\beta_0\iota_1')u_t(\pi_0)$ and
$(\mc J-\beta_0\iota_1')b_0=0$, hence $r_t(\beta_0,\pi_0)$ depends only on $\boldsymbol\varepsilon_{2t}$.
Therefore,
\[
\E[g_{1,t}^*(\beta_0,\pi_0)]
=
\E[z_t\,r_t(\beta_0,\pi_0)]
=
\E\!\left[\E[z_t\,r_t(\beta_0,\pi_0)\mid z_t]\right]
=
\E\!\left[z_t\,\E[r_t(\beta_0,\pi_0)\mid z_t]\right]
=
0,
\]
while the internal proxy $z_t^*(\beta_0,\pi_0,\omega_0)$ is $\bs\sigma(u_t(\pi_0))$-measurable and, under Assumption~\ref{ass:strong_exogeneity}, its population counterpart is a function of $\varepsilon_{1t}$ only. Hence, for $i\ge2$,
$h_i(z_t^*(\beta_0,\pi_0,\omega_0))$ is measurable with respect to $\bs \sigma(\varepsilon_{1t})$ and
\[
\E[g_{i,t}^*(\beta_0,\pi_0,\omega_0)]
=
\E\!\left[h_i(z_t^*(\beta_0,\pi_0,\omega_0))\,\E[r_t(\beta_0,\pi_0)\mid \varepsilon_{1t}]\right]
=
0,
\]
where the last equality follows from $\E[\boldsymbol\varepsilon_{2t}\mid \varepsilon_{1t}]=\mathbf 0$. Stacking $i=1,\ldots,L$ yields $\E[g_t^{(Z^*)}(\beta_0,\pi_0,\omega_0)]=0$.

Next, it has already been established in Proposition~\ref{thm:beta_SP}  that  $\sqrt{T}(\hat \pi_T - \pi_0) = -M^{-1} \frac{1}{\sqrt{T}}\sum_{t=1}^T m_t(\pi_0) + o_p(1).$ 
For $\hat\omega_T$, denote $S_\pi \coloneqq \E\left[\frac{\partial s_t(\pi_0, \omega_0)}{\partial \pi'} \right]$ and $S_\omega \coloneqq  \E\left[\frac{\partial s_t(\pi_0, \omega_0)}{\partial \omega'} \right].$ Recall $s_t(\pi,\omega)=\vect{u_t(\pi)u_t(\pi)'-\Sigma_u(\omega)}$, by our parametrization, we have $S_\omega = -I_{n^2}$, while $\E[u_t(\pi_0)Y_{t-1}']=0$ and therefore $S_\pi = - \E\left[((u_t(\pi_0) \otimes I_n) + ( I_n \otimes u_t(\pi_0) ) )(Y_{t-1}' \otimes I_n)\right] = \bs 0$. Expanding the moment function around the true value and inserting the influence function for $\hat \pi_T$ yields
 \begin{align*}
        \sqrt{T}(\hat \omega_T - \omega_0) &= \frac{1}{\sqrt{T}}\sum_{t=1}^T -S_\omega^{-1} \left( s_t(\pi_0, \omega_0) - S_\pi M^{-1} m_t(\pi_0) \right) + o_p(1)\\
        &= \frac{1}{\sqrt{T}}\sum_{t=1}^T s_t(\pi_0, \omega_0) + o_p(1).
\end{align*}
For the GMM estimator $\hat\beta_T^*(W)$, expanding the first order condition yields   
 \[
 \sqrt{T}(\hat\beta_T^*(W)-\beta_0)
=
-(G_\beta^{*'}WG_\beta^*)^{-1}G_\beta^{*'}W\,
\frac1{\sqrt{T}}\sum_{t=1}^T g_t^{(Z^*)}(\beta_0,\hat\pi_T,\hat\omega_T)
+o_p(1),
 \]
where 
 \begin{align*}
      \frac{1}{\sqrt{T}} \sum_{t=1}^T g_t^{(Z^*)}(\beta_0,\hat \pi_T, \hat \omega_T)  &=  \frac{1}{\sqrt{T}} \sum_{t=1}^T g_t^{(Z^*)}(\beta_0, \pi_0,  \omega_0) + G_\pi^*  \sqrt{T}(\hat \pi_T - \pi_0)  + G_\omega^*  \sqrt{T}(\hat \omega_T - \omega_0)  + o_p(1) \\
      &=  \frac1{\sqrt{T}}\sum_{t=1}^T\left[g_t^{(Z^*)}(\beta_0,\pi_0,\omega_0) - G_\pi^*M^{-1}m_t(\pi_0) + G_\omega^*s_t(\pi_0,\omega_0) \right]+o_p(1). 
\end{align*}
Define $g_t^{(Z^*),\perp}$ as in the proposition. Then
\[
\sqrt{T}(\hat\beta_T^*(W)-\beta_0)
=
- (G_\beta^{*'}WG_\beta^* )^{-1}G_\beta^{*'}W\,
\frac1{\sqrt{T}}\sum_{t=1}^T g_t^{(Z^*),\perp}
+o_p(1),
\]
where
\[
\frac1{\sqrt{T}}\sum_{t=1}^T g_t^{(Z^*),\perp}\ \Rightarrow\ N(0,S_\perp^*),
\qquad
S_\perp^*:=\E \left[g_t^{(Z^*),\perp}g_t^{(Z^*),\perp\,\prime} \right].
\]
This yields the expression for $V_\beta^*$. The efficiency claim follows by the same argument as in Proposition~\ref{thm:beta_SP}, where the optimal weighting is $S_\perp^{*-1}$, with the asymptotic variance given by $(G_\beta^{*'}S_\perp^{*-1}G_\beta^*)^{-1}$.

Assume $L>1$. Then under $H_0:\E[g_t^{(Z^*)}(\beta_0,\pi_0,\omega_0)]=0$, the same expansion yields
\[
\sqrt{T}\,\bar g_T^{(Z^*)}(\hat\beta_T^*(W),\hat\pi_T,\hat\omega_T)
=
\Big(I_{L(n-1)}-G_\beta^*(G_\beta^{*'}WG_\beta^*)^{-1}G_\beta^{*'}W\Big)
\frac1{\sqrt{T}}\sum_{t=1}^T g_t^{(Z^*),\perp}
+o_p(1).
\]
Let $S_\perp^{*1/2}$ be the symmetric square root of $S_\perp^*$ and set
$\bs\zeta_T:=T^{-1/2}\sum_{t=1}^T S_\perp^{*-1/2}g_t^{(Z^*),\perp}$, so that $\bs\zeta_T\Rightarrow N(0,I_{L(n-1)})$. Note that $J_T^*(S_\perp^{*-1})$ is asymptotically a quadratic form in $\bs\zeta_T$ with a symmetric idempotent matrix of rank $(L-1)(n-1)$. The null distribution and the consistency under fixed alternatives follows by the same argument as in Proposition~\ref{thm:beta_SP}.
\end{proof}

\begin{proof}[Proof of Proposition~\ref{thm:weak_Q_ISP}]
The proof follows by the same argument in the proof of Proposition~\ref{thm:weak_Q}. 
Note that
under Assumption~\ref{Tass:ISPweak}(iii),  the influence function 
can be written as, uniformly on
$\mc B$,
\begin{align*}
\sqrt T
\bar g_T^{(Z^*)}(\beta,\hat\pi_T,\hat\omega_T)
&=
\sqrt T
\bar g_T^{(Z^*)}(\beta,\pi_0,\omega_0)
-G_{\pi,T}^*(\beta)M^{-1}
\frac1{\sqrt T}\sum_{t=1}^T m_t(\pi_0) \\
&+G_{\omega,T}^*(\beta)
\frac1{\sqrt T}\sum_{t=1}^T s_t(\pi_0,\omega_0)
+o_p(1) =
\mc G_T^*(\beta)
+\sqrt T\,\bs\psi_T^*(\beta)
+o_p(1).
\end{align*}
Assumption~\ref{Tass:ISPweak}(ii) and Lemma~\ref{lem:weak3_glb} imply
$
\sqrt T\,
\bar g_T^{(Z^*)}(\beta,\hat\pi_T,\hat\omega_T)
\Rightarrow
\mc G^*(\beta)+\bs\Psi^*(\beta)
$
uniformly on \(\mc B\) and \eqref{eq:Q_ISP_process_limit} follows from Assumption~\ref{Tass:ISPweak}(iv), continuity and
uniform positive definiteness of \(\Omega^*(\beta)\).
At \(\beta=\beta_0\), 
$\bs\Psi^*(\beta_0)=\bs0$ by Lemma~\ref{lem:weak3_glb} and hence
\[
\sqrt T\,
\bar g_T^{(Z^*)}(\beta_0,\hat\pi_T,\hat\omega_T)
\Rightarrow
\mc G^*(\beta_0)
\sim
N\!\left(\bs0,\Omega^*(\beta_0)\right),
\]
which proves $Q_T^*(\beta_0)\Rightarrow\chi^2_{L(n-1)}$ and the limits of the continuous updating estimator as well as the associated $J$-statistic.
\end{proof}

\begin{lemma}\label{lem:ISP_ass}
Under Assumption~\ref{ass:strong_exogeneity}, $\E[\boldsymbol\varepsilon_{2t}\mid z_t]=\bs 0$,
$\E[\boldsymbol\varepsilon_{2t}\mid \varepsilon_{1t}]=\bs 0$ a.s., and $\E[v_t\varepsilon_t]=\bs 0$.
\end{lemma}
\begin{proof}[Proof of Lemma~\ref{lem:ISP_ass}]
Since $z_t=f(\varepsilon_{1t},v_t)$ is measurable with respect to the $\bs\sigma(\varepsilon_{1t},v_t)$, the law of iterated expectations and Assumption~\ref{ass:strong_exogeneity} give
$
\E[\boldsymbol\varepsilon_{2t}\mid z_t] = \E\left[\E[\boldsymbol\varepsilon_{2t}\mid \varepsilon_{1t},v_t]\mid z_t\right] =\E[\bs 0\mid z_t] = \bs 0.$
Similarly,
$ \E[\boldsymbol\varepsilon_{2t}\mid \varepsilon_{1t}] = \E\left[\E[\boldsymbol\varepsilon_{2t}\mid \varepsilon_{1t},v_t]\mid \varepsilon_{1t}\right] =\bs 0.$
Finally, $\E[v_t\varepsilon_{1t}]=0$ holds by assumption, while $ \E[v_t\boldsymbol\varepsilon_{2t}]
= \E\left[v_t\,\E[\boldsymbol\varepsilon_{2t}\mid v_t]\right]= \E\left[v_t\,\E\!\left[\E[\boldsymbol\varepsilon_{2t}\mid \varepsilon_{1t},v_t]\mid v_t\right]\right] = \bs 0.$
\end{proof}

\begin{lemma}\label{lem:G_omega_star}
Under the conditions in Proposition~\ref{thm:beta_ISP}, denote $G_{\omega, i}^*:=\E[\partial g_{i,t}^*(\beta_0,\pi_0,\omega_0)/\partial\omega']$, if  $\E[\bs\varepsilon_{2t}\bs\varepsilon_{2t}'\mid \varepsilon_{1t}]=\Sigma_{22}$ a.s., then 
\[
G_{\omega, 1}^* = 0, \qquad 
G_{\omega, i}^*
=
C_i\,\E[\varepsilon_{1t}^{i-1}],
\;
\text{for } \ i =2, \ldots, L,
\]
for some finite constant matrix $C_i$. Consequently, if $\E[\varepsilon_{1t}^{i-1}]=0$ for all $2 \leq i \leq L$, then $G_\omega^*=0$ and the estimation of $\Sigma_u$ does not influence the asymptotic variance of the GMM estimator for $\beta$.
\end{lemma}
\begin{proof}[Proof of Lemma~\ref{lem:G_omega_star}]
For $i=1$, since the raw instrument does not depend on $\omega$, $G_{\omega,1}^* = 0$. For $i = 2, \ldots, L$, taking differential of $z_t^*$ yields
    \[
d z_t^*(\beta,\pi,\omega) = -\frac{b(\beta)'(\Sigma_u^{-1} (d \Sigma_u) \Sigma_u^{-1})u_t(\pi)}{\dot\sigma_1(\beta,\omega)} +
\frac{z_t^*(\beta,\pi,\omega)}{2\dot\sigma_1^2(\beta,\omega)} b(\beta)'(\Sigma_u^{-1} (d \Sigma_u) \Sigma_u^{-1}) b(\beta)
\]
and hence, at the true value
\[
\frac{\partial z_t^*(\beta_0,\pi_0,\omega_0)}{\partial \omega'} =
-\frac{\varepsilon_{1t}}{2\dot\sigma_1} (b_0' \Sigma_{u,0}^{-1} \otimes b_0' \Sigma_{u,0}^{-1})  -
\frac{1}{\dot\sigma_1} (\bs\varepsilon_{2t}'  \mc J B_0'  \Sigma_{u,0}^{-1} \otimes  b_0' \Sigma_{u,0}^{-1}),
\]
where $\dot\sigma_1 \coloneqq \dot\sigma_1(\beta_0,\omega_0) = \sigma_1^{-1}$. Note that $r_t(\beta_0,\pi_0) = (B_{22} - \beta_0 b_2') \bs\varepsilon_{2t}$, the component proportional
to $\varepsilon_{1t}$ has expectation zero by
$\E[\bs\varepsilon_{2t}\mid\varepsilon_{1t}]=\bs0$. Conditional
homoskedasticity therefore gives
\begin{align*}
G_{\omega,i}^*
&=
-\frac{i}{\dot\sigma_1}
\E\left[
\left(\frac{\varepsilon_{1t}}{\sigma_1}\right)^{i-1}
\right]
\left(
(B_{22}-\beta_0b_2')\Sigma_{22}B_2'\Sigma_{u,0}^{-1}
\otimes
b_0'\Sigma_{u,0}^{-1}
\right) =
C_i\,\E[\varepsilon_{1t}^{i-1}].
\end{align*}
This proves the Lemma with $C_i := -\frac{i}{\dot\sigma_1\sigma_1^{i-1}}
\left( (B_{22}-\beta_0b_2')\Sigma_{22}B_2'\Sigma_{u,0}^{-1} \otimes b_0'\Sigma_{u,0}^{-1} \right)$.
\end{proof}

\begin{lemma}\label{lem:G_pi_star}
    Under the conditions in Proposition~\ref{thm:beta_ISP}, denote $G_{\pi, i}^*:=\E[\partial g_{i,t}^*(\beta_0,\pi_0,\omega_0)/\partial\pi']$, then 
    $G_{\pi,1}^*=-\E[z_t Y_{t-1}']\otimes(\mc J-\beta_0\iota_1'),$
    and for $2 \leq i\leq L$,
    \[
    G_{\pi,i}^*
    =
    -\E\left[\left(\frac{\varepsilon_{1t}}{\sigma_1}\right)^i Y_{t-1}'\right]\otimes(\mc J-\beta_0\iota_1')
-\frac{i}{\dot\sigma_1} (B_{22} - \beta_0 b_2') 
\E\left[\left(\frac{\varepsilon_{1t}}{\sigma_1}\right)^{i-1} \bs\varepsilon_{2t}(Y_{t-1}'\otimes b_0'\Sigma_{u,0}^{-1})\right].
    \]
    Consequently, if $z_t$ is lag exogenous (s.t., $\E[z_t Y_{t-1}']=0$),
    $\varepsilon_t$ is serially independent and $\E[Y_{t-1}] = 0$, then $G_\pi^*=0$ and estimation of $\pi$ does not influence the asymptotic variance of the GMM estimator for $\beta$.
\end{lemma}
\begin{proof}[Proof of Lemma~\ref{lem:G_pi_star}]
    For $i=1$, we have $ \partial g_{1t}^*(\beta,\pi)/\partial\pi'
=
z_t(\mc J-\beta\iota_1')\frac{\partial u_t(\pi)}{\partial\pi'}
=
-(z_t Y_{t-1}')\otimes(\mc J-\beta\iota_1').$ For $i \geq 2$, note that 
\[
\frac{\partial z_t^*(\beta,\pi,\omega)}{\partial\pi'}
=
-\frac{1}{\dot\sigma_1(\beta,\omega)}
\left(
Y_{t-1}'\otimes
b(\beta)'\Sigma_u(\omega)^{-1}
\right),
\]
and hence
\begin{align*}
    \frac{\partial g_{it}^*(\beta,\pi,\omega)}{\partial\pi'}
&=
-( (z_t^*(\beta,\pi,\omega))^i Y_{t-1}')\otimes(\mc J-\beta\iota_1')\\
\hspace{3ex} &-\frac{i}{\dot\sigma_1(\beta,\omega)}(z_t^*(\beta,\pi,\omega))^{i-1}r_t(\beta,\pi)(Y_{t-1}'\otimes b(\beta)'\Sigma_u(\omega)^{-1}).
\end{align*}
Since $z_t^*(\beta_0,\pi_0,\omega_0) = \varepsilon_{1t}/\sigma_1$ and $r_t(\beta_0,\pi_0) = (B_{22} - \beta_0 b_2') \bs\varepsilon_{2t}$, 
taking expectation and evaluating the expression at the true value yield the result. Hence, $G_\pi^* = 0$ holds if $\E[z_t Y_{t-1}']=0$, and for all $2 \leq i \leq L$, $\E[\varepsilon_{1t}^i Y_{t-1}']=0$ and $\E\left[\varepsilon_{1t}^{i-1}\,\bs\varepsilon_{2t}\,Y_{t-1}'\right]=0$. To see that the latter condition is satisfied under our assumption, note that $\E[\varepsilon_{1t}^i Y_{t-1}']=\E[\varepsilon_{1t}^i]\E[Y_{t-1}'] = 0$ and
\[
\E[\varepsilon_{1t}^{i-1}\bs\varepsilon_{2t}Y_{t-1}']
=
\E[\varepsilon_{1t}^{i-1}\bs\varepsilon_{2t}]\, \E[Y_{t-1}']
=
\E\big[\varepsilon_{1t}^{i-1}\E[\bs\varepsilon_{2t}\mid \varepsilon_{1t}]\big]\,\E[Y_{t-1}']
=0,
\]
by serial independence of $\varepsilon_{t}$ (and hence independent of $Y_{t-1}$), and $\E[Y_{t-1}]=0$. This establishes $G_\pi^*=0$.
\end{proof}

The following lemma clarifies what Assumption~\ref{ass:weak1} implies for higher-order SP relevance in the classical linear measurement-error model.
\begin{lemma}
\label{lem:weak_lin_proxy}
Let $z_{t,T}=\phi_T\varepsilon_{1t}+v_t,$ with $\phi_T=\frac{\phi^*}{\sqrt T}$, $\phi^*\neq0,$
where the noise $v_t \sim (0, \sigma_v^2)$ with $\sigma_v > 0$, is independent of $\varepsilon_{1t}$ and all expectations appearing below exist. Then, Assumption~\ref{ass:weak1} holds with $\gamma_1=\phi^*\sigma_1^2$, and
\[
\forall i =1, \ldots, L, \quad
\rho_{i,T}=O(T^{-k_i/2}),
\quad \text{with }\ k_i \geq 1
\]
In particular, 
\[
\rho_{2,T} = 
\begin{cases}
    O(T^{-1}) & \text{if } \ \E[\varepsilon_{1t}^3]  \neq 0 \\
    0 & \text{else. } 
\end{cases}
\quad
\text{and}
\quad
\rho_{3,T} = O(T^{-1/2}) + o(T^{-1/2}).
\]
\end{lemma}
\begin{proof}[Proof of Lemma~\ref{lem:weak_lin_proxy}]
Direct calculation of $\E[z_{t,T} \varepsilon_{1t}]$ yields $\gamma_1=\phi^*\sigma_1^2$. 
For any $i = 1, \ldots, L$
\begin{align*}
    \rho_{i,T} 
    &= \E[z_{t,T}^i\varepsilon_{1t}] = \E\!\left[\sum_{k=0}^i {i\choose k}\phi_T^k\varepsilon_{1t}^{k+1}v_t^{i-k}\right] =
\sum_{k=0}^i {i\choose k}\phi_T^k
\E[\varepsilon_{1t}^{k+1}v_t^{i-k}]\\
&=
\sum_{k=1}^i
{i\choose k}
(\phi^*)^kT^{-k/2}\E[\varepsilon_{1t}^{k+1}] \E[v_t^{i-k}],
\end{align*}
where the second line follows from independence between $\varepsilon_{1t}$ and $v_t$ and that $\E[\varepsilon_{1t}]=0$.
Hence, for each $i = 2, \ldots, L$, 
the rate of $\rho_{i,T}$ is governed by $k_i := \min_{k \in \{1, \ldots, i\}}\{k: \E[\varepsilon_{1t}^{k+1}] \E[v_t^{i-k}] \neq 0\}$. Note that $k_i \geq 1$ by construction and the conclusion follows.
Direct calculation yields $\rho_{2,T} = (\phi^*)^2\E[\varepsilon_{1t}^{3}] T^{-1}$ and $\rho_{3,T}
=
3\phi^*\sigma_1^2\sigma_v^2T^{-1/2}
+
(\phi^*)^3\E[\varepsilon_{1t}^{4}]T^{-3/2}.$ This completes the proof.
\end{proof}

Lemma~\ref{lem:weak_lin_proxy} shows that, for a local linear measurement-error proxy, Assumption~\ref{ass:weak1} implies that all $\rho_{i,T}$ decreases at a rate of $T^{-1/2}$ or faster, and therefore, all SPs are weak, i.e., $\bs\rho_{T} = O(T^{-1/2})$. However, the speed, at which $\rho_{i,T}$ vanishes, does not increase monotonically with $i$. As shown in the Lemma, both $\rho_{1,T}$ and $\rho_{3,T}$ decrease at the same rate.
With this being said, if the proxy $z_t= f(\varepsilon_{1,t}, v_t)$ is nonlinear in the target shock $\varepsilon_{1,t}$, higher-order SPs may remain strongly relevant even though the original proxy is weak. 
As an example, let $\varepsilon_{1t}=X_t-1$ with $X_t\sim\mathrm{Exp}(1)$, then
$\E[\varepsilon_{1t}] = 0$, $\sigma_1^2 = \E[\varepsilon_{1t}^2] = 1$, 
$\E[\varepsilon_{1t}^3] = 2$, $\E[\varepsilon_{1t}^4] = 9$ and $\E[\varepsilon_{1t}^5] = 44$.\footnote{Recall that $\E[X_t^i] = i !$. }
To better illustrate this, let $\tilde f(\varepsilon_{1,t}) = \varepsilon_{1,t}^2 - 2 \varepsilon_{1,t} -1$ and define a proxy
\[
z_{t,T} = \tilde f(\varepsilon_{1,t}) + \frac{\gamma_1 \varepsilon_{1,t}}{\sqrt{T}}.
\]
Note that
$\E[\tilde f(\varepsilon_{1,t})\varepsilon_{1,t}] = 0$,
$\E[\tilde f(\varepsilon_{1,t})^2\varepsilon_{1,t}] = 16$,
and $\E[\tilde f(\varepsilon_{1,t})\varepsilon_{1,t}^2] = 4$. 
Hence, $\E[z_{t,T} \varepsilon_{1,T}] = \gamma_1/\sqrt{T}$ and Assumption~\ref{ass:weak1} is respected, whereas
\begin{align*}
    \E[z_{t,T}^2\varepsilon_{1t}]
= \E[\tilde f(\varepsilon_{1,t})^2\varepsilon_{1,t}] +  \frac{2 \gamma_1 }{\sqrt{T}}\E[\tilde f(\varepsilon_{1,t})\varepsilon_{1,t}^2] + \frac{ \gamma_1^2 }{T}\E[\varepsilon_{1,t}^3] =
16+8 \gamma_1 T^{-1/2}+2\gamma_1^2T^{-1}.
\end{align*}
Hence, $\rho_{2,T}\to16\neq0$ as $T \to \infty$. 
Consequently, the second-order synthetic remains relevant despite the weakness of the original proxy, so that
$\lim_{T \to \infty} \bs \rho_T = \bs \rho_\infty$ with $\rho_\infty \neq 0$ and the standard asymptotics apply.

\begin{lemma}\label{lem:Stein}
Under Assumption~\ref{ass:strong_exogeneity} and provided that the expectations exist, for continuously differentiable $h:\R\to\R$, the implied internal moment function $g_{h,t}^*(\beta) := g_{h,t}^*(\beta, \pi_0, \omega_0)  =h(z_t^*(\beta,\pi_0,\omega_0))  r_t(\beta,\pi_0)$ satisfies 
\begin{align*}
	\E_T\left[
\frac{\partial g_{h,t}^*(\beta_0)}{\partial\beta'}  \right]
  =\sigma_{1}\left( \E_T \left[h'(X_{t,T})C_0 \bs\varepsilon_{2t}\bs\varepsilon_{2t}' \Sigma_{22}^{-1} C_0^{-1} \right]-\E_T[X_{t,T}h(X_{t,T})] I_{n-1} \right),
\end{align*}
where $C_0 = C(\beta_0)$ with $C(\beta) := (\mc J-\beta\iota_1') B_2$ and $X_{t,T} := \varepsilon_{1t}/\sigma_1$ is the standardized target shock. If 
$\E[\bs\varepsilon_{2t}\bs\varepsilon_{2t}'\mid \varepsilon_{1t}]=\Sigma_{22}$ a.s., then
\begin{align}\label{eq:Stein_general}
	\E_T\left[ \frac{\partial g_{h,t}^*(\beta_0)}{\partial\beta'}  \right]
  =\sigma_{1}\left( \E_T[h'(X_{t,T})]-\E_T[X_{t,T}h(X_{t,T})] \right)I_{n-1},
\end{align}
Therefore, for $h_i(x)=x^i$, we have
\begin{align}\label{eq:Stein_poly}
  	\E_T\left[ \frac{\partial g_{i,t}^*(\beta_0)}{\partial\beta'}  \right]
  =\sigma_{1}\left( i\E_T[X_{t,T}^{i-1}]-\E_T[X_{t,T}^{i+1}]\right)I_{n-1}, 
  \qquad \text{for } \
  i = 1, 2, \ldots.
\end{align}
In particular, for $i =2$ and $i = 3$, we have 
\[
\E_T\left[ \frac{\partial g_{2,t}^*(\beta_0)}{\partial\beta'}  \right] = - \sigma_{1} \E_T[X_{t,T}^{3}] I_{n-1},
  \quad \text{and} \quad 
\E_T\left[ \frac{\partial g_{3,t}^*(\beta_0)}{\partial\beta'}  \right] = \sigma_{1} \left( 3 - \E_T[X_{t,T}^{4}]\right) I_{n-1}.
\]
\end{lemma}
\begin{proof}[Proof of Lemma~\ref{lem:Stein}]
Taking differential of $r_t(\beta)=:r_t(\beta, \pi_0)$ and $z_t^*(\beta)$, yields $\dd r_t(\beta, \pi_0) = -(\varepsilon_{1t}+ b_2' \bs\varepsilon_{2t}) \dd \beta = -u_{1t} \dd\beta$ 
and
	\begin{align*} 
		\dd z_t^*(\beta) &= \frac{(\dd b(\beta))' \Sigma_u^{-1} u_t}{\dot\sigma_1(\beta, \omega_0)}
		- \frac{1}{2} \frac{ b(\beta)' \Sigma_u^{-1} u_t}{\dot\sigma_1(\beta, \omega_0)} \frac{ \dd \dot\sigma_1^2(\beta, \omega_0)}{\dot\sigma_1^2(\beta, \omega_0)}\\
		&= \frac{(\dd b(\beta))' \Sigma_u^{-1} u_t}{\dot\sigma_1(\beta, \omega_0)}
		-  z_t^*(\beta) \frac{ (\dd b(\beta))' \Sigma_u^{-1}  b(\beta)}{\dot\sigma_1^2(\beta, \omega_0)}
	\end{align*}
	At $\beta = \beta_0$, $z_t^*(\beta_0) = \varepsilon_{1t}/\sigma_1 = X_{t,T}$, $\dot\sigma_1(\beta_0, \omega_0) = 1/\sigma_1$ and
	\begin{align*}
		r_t (\beta_0) &= (\mc J-\beta_0\iota_1') (b_0 \varepsilon_{1t} + B_2 \bs\varepsilon_{2t}) 
		= \beta_0 \varepsilon_{1t} - \beta_0 \iota_1' b_0  \varepsilon_{1t} + (\mc J-\beta_0\iota_1') B_2 \bs\varepsilon_{2t}\\
		&= \beta_0 \varepsilon_{1t} - \beta_0 \varepsilon_{1t}+
		 (\mc J-\beta_0\iota_1') B_2 \bs\varepsilon_{2t} 
		 =  C_0 \bs\varepsilon_{2t}.
	\end{align*}
	Therefore,
	\begin{align*}
		\dd z_t^*(\beta_0) &= \sigma_1 (\dd \beta)' \mc J \Sigma_u^{-1} u_t - X_{t,T} \sigma_1^2 (\dd \beta)' \mc J \Sigma_u^{-1} b_0 \\
		&= \sigma_1^2 (\dd \beta)' \mc J \Sigma_u^{-1} b_0 X_{t,T}
		+ \sigma_1 (\dd \beta)' \mc J \Sigma_u^{-1} B_2  \bs\varepsilon_{2t}  - X_{t,T} \sigma_1^2(\dd \beta)' \mc J \Sigma_u^{-1} b_0 \\
		&= \sigma_1 (\dd \beta)' \mc J \Sigma_u^{-1} B_2  \bs\varepsilon_{2t}, 
	\end{align*}
	thus 
	$\frac{\partial z_t^*(\beta_0)}{\partial \beta'} = \sigma_1 \bs\varepsilon_{2t}' B_2' \Sigma_u^{-1}  \mc J '$. Next, we show that $\frac{\partial z_t^*(\beta_0)}{\partial \beta'} = \sigma_1 \bs\varepsilon_{2t}' \Sigma_{22}^{-1} C_0^{-1}$.
	To prove the claim, note that
	\[
B_{0} \begin{pmatrix}-b_{2}'C_0^{-1}\\ C_0^{-1}\end{pmatrix}
  =\begin{pmatrix}
  1 & b_{2}'\\
  \beta_{0} & B_{22}
  \end{pmatrix}
  \begin{pmatrix}-b_{2}'C_0^{-1}\\ C_0^{-1}\end{pmatrix}
  =\begin{pmatrix}
  \b 0_{1\times(n-1)}\\
  (B_{22}-\beta_{0}b_{2}')C_0^{-1}
  \end{pmatrix}
  =\mc J',
	\]
	where the last equality holds by right-multiplying the
	identity $C_0 = B_{22}-\beta_{0}b_{2}'$ with $C_0^{-1}$. 
	Note that $C_0^{-1}$ is well-defined by non-singularity of $B_0$.
	Next, right-multiplying the identity $B_2'=\mc J B_0'$ by $(B_0')^{-1}$ yields $B_2' (B_0')^{-1} = \mc J$. 
	Hence,
\begin{align*}
	\frac{\partial z_t^*(\beta_0)}{\partial \beta'} &= 
\sigma_1 \bs\varepsilon_{2t}' B_2' \Sigma_u^{-1}  \mc J '
 = \sigma_1 \bs\varepsilon_{2t}' B_2' (B_0')^{-1} \begin{pmatrix}
   1/\sigma_1^2 & 0 \\
	0 & \Sigma_{22}^{-1}
	\end{pmatrix}(B_0)^{-1}  \mc J ' \\
&= \sigma_1 \bs\varepsilon_{2t}' \mc J \begin{pmatrix}
   1/\sigma_1^2 & 0 \\
	0 & \Sigma_{22}^{-1}
	\end{pmatrix}
\begin{pmatrix}-b_{2}'C_0^{-1}\\ C_0^{-1}\end{pmatrix}
 = \sigma_1 \bs\varepsilon_{2t}' \Sigma_{22}^{-1} C_0^{-1}.
\end{align*}
Therefore, 
\begin{align*}
	\frac{\partial g_{i,t}^*(\beta_0)}{\partial \beta'}
	&= h_i'(z_t^*(\beta_0)) r_t(\beta_0) \frac{\partial z_t^*(\beta_0)}{\partial \beta'}
	+ h_i(z_t^*(\beta_0))   \frac{\partial r_t(\beta_0)}{\partial \beta'}\\
	&= \sigma_1 h_i'(X_{t,T})C_0 \bs\varepsilon_{2t}\bs\varepsilon_{2t}' \Sigma_{22}^{-1} C_0^{-1} - h_i(X_{t,T}) u_{1t}I_{n-1} \\
	&=  \sigma_1 h_i'(X_{t,T})C_0 \bs\varepsilon_{2t}\bs\varepsilon_{2t}' \Sigma_{22}^{-1} C_0^{-1} 
	- \left[ \sigma_1 h_i(X_{t,T}) X_{t,T}  + h_i(X_{t,T}) b_2' \bs\varepsilon_{2t} \right] I_{n-1}, 
\end{align*} 
where the last equality holds due to $u_{1t} = \varepsilon_{1t} + b_2' \bs\varepsilon_{2t}  = \sigma_1 X_{t,T} + b_2' \bs\varepsilon_{2t}$.
Taking expectation and the last term vanishes since $\E_T[\bs\varepsilon_{2t} \mid X_{t,T}] = 0$ as implied by Assumption~\ref{ass:strong_exogeneity} the law of iterated expectations, which establishes the general statement. Under conditional homoskedasticity, $\E[\bs\varepsilon_{2t}\bs\varepsilon_{2t}'\mid \varepsilon_{1t}]=\Sigma_{22}$ cancels with $\Sigma_{22}^{-1}$ and $C_0$ with $C_0^{-1}$, which yields \eqref{eq:Stein_general}. The last part of the Lemma is obtained by plugging in the power functions. 
\end{proof}

\begin{lemma}\label{lem:conditional_expectation}
Let $X, Y$ be random vectors on $\R^{d_1}$ and $\R^{d_2}$, s.t., $\E\|Y\| < \infty$. Then, the following statements are equivalent:
\begin{enumerate}[label=(\roman*)]
    \item $\E[Y \mid X] = \bs 0$, a.s.
    \item For every Borel-measurable $f: \R^{d_1} \mapsto \R$, s.t., $\E[|f(X)|\|Y\|] < \infty$, 
    $\E[f(X)Y] = \b 0$.
    \item For every Borel set $B \in \B(\R^{d_1})$, $\E[Y \mathbf{1}_{X\in B}] = \b 0$.
\end{enumerate}
\end{lemma}
\begin{proof}[Proof of Lemma~\ref{lem:conditional_expectation}]
     To show $(i) \Rightarrow (ii)$, suppose $\E[Y \mid X] = \bs 0$, a.s., and let $f: \R^{d_1} \mapsto \R$ be Borel-measurable. Then, $f(X)$ is $\bs\sigma(X)$-measurable.
     For every $n \in \N$, define truncation 
     $f_n(x) := \max\{-n, \min\{f(x), n\}\}$. Then, $f_n$ is bounded and also $\bs\sigma(X)$-measurable, and thus, $\forall n \in \N$, $\E[f_n(X)Y] = \E[f_n(X)\E[Y \mid X]] = \b 0$.
     Note $f_n(X) \to f(X)$ a.s. and $\|Y(f_n(X) - f(X))\| \leq |f(X)|\|Y\|$. Since $\E[|f(X)|\|Y\|] < \infty$, the dominated convergence theorem entails that $\E[f_n(X)Y] \to \E[f(X)Y]$, where the left-hand side is zero for all $n$. Hence, $\E[f(X)Y] = \b 0$.
     The step $(ii) \Rightarrow (iii)$ follows trivially, since the indicator function $f(X) := \mathbf{1}_{X \in B}$ is a Borel-measurable function and satisfies $\E[|f(X)|\|Y\|] \leq \sup_X|f(X)| \E\|Y\| = \E\|Y\| < \infty$.
     To show $(iii) \Rightarrow (i)$, suppose $\E[Y \mathbf{1}_{X\in B}] =\b 0$, for all $\B(\R^{d_1})$-sets. By definition, $\bs\sigma(X) = \{X^{-1}(B): B \in \B(\R^{d_1}) \}$, s.t., for every $A \in \bs\sigma(X)$, there exists some $B \in \B(\R^{d_1})$, s.t., $A = \{X \in B\}$. Therefore, for every $A \in \bs\sigma(X)$, it holds that $\int_AY d\mathbb{P} = \E[Y \mathbf{1}_{X\in B}] =\b 0$. Note that $\bs 0$ is $\bs\sigma(X)$-measurable, and by the Radon–Nikodym theorem, it is a version of $\E[Y \mid X] $, which is a.s.\ unique. 
     This completes the proof
\end{proof}

\subsection{Asymptotic properties of the strong proxy exogeneity test based on SPs}
\label{app:SP_asymptotics}
Consider a sequence of power functions $\{h_i\}_{i = 1}^L$ with $h_i(z) = z^i$, define 
\begin{equation}\label{eq:g_i}
g_{i,t}(\beta, \pi) \coloneq h_i(z_t) r_t(\beta, \pi),
\quad 
\text{with }
\;
r_t(\beta, \pi) \coloneq \bs{u_{2t}}(\pi) - \beta u_{1t}(\pi) = (\mc J-\beta\iota_1') u_t(\pi),
\end{equation}
for $i =1, \ldots, L$. For $L =1$, relevance of the original proxy, $\rho_1 \neq 0$, yields $n-1$ moment conditions for the $n-1$ components in the structural parameter vector $\beta$, so that $\beta$ is just-identified. When $L > 1$, the additional blocks $ i =2, \ldots, L$ generate over-identifying restrictions, which allow us to test whether any moment condition is violated.
Stacking moments yields $g_t^{(Z)}(\beta,\pi):=(g_{1,t}(\beta,\pi)',\ldots,g_{L,t}(\beta,\pi)')'\in\R^{L(n-1)}$.
For a weighting matrix $\hat W_T$ with $\|\hat W_T-W\|=o_p(1)$ for some symmetric positive definite matrix $W$,
\begin{align}
	\label{eq:beta_GMM}
	\hat\beta_T(W) \coloneq \argmin \limits_{\beta} \bar g_T^{(Z)}(\beta,\hat\pi_T)'\,
\hat W_T\,
    \bar g_T^{(Z)}(\beta,\hat\pi_T), \;
	 \text{with } \
	 \bar g_T^{(Z)}(\beta, \pi)= \frac{1}{T}\sum_{t=1}^{T} g_t^{(Z)}(\beta,\pi),
\end{align}
while the associated $J$-test statistic is
\begin{equation}\label{eq:J_test_stat}
	J_T(W)\coloneq T\,\bar g_T^{(Z)}(\hat\beta_T(W),\hat\pi_T)'\,
\hat W_T\,
\bar g_T^{(Z)}(\hat\beta_T(W),\hat\pi_T).
\end{equation}

\cite{BK2024} analyze the special case for the static SEM/VAR(0) setting, where the reduced-form innovations are directly observed (i.e., there is no first-step estimation of reduced-form parameters). In a general SVAR, by contrast, the reduced-form shocks must be constructed from estimated $\hat\pi_T$, which induces an additional first-step error component in the sample moments that is asymptotically non-negligible. Consequently, the weighting matrix adopted in \cite{BK2024} when $u_t$ is observed, namely $W=S_{gg}^{-1}$ with $S_{gg}:=\E[g_t^{(Z)}(\beta_0,\pi_0)g_t^{(Z)}(\beta_0,\pi_0)']$ is not efficient, and using it can also distort the null distribution of the $J$-statistic in dynamic settings. 

\subsubsection{Standard test under identification}
We start by characterizing large-sample properties of $\hat\beta_T(W)$ and $J_T(W)$ under the condition that the proxy remains relevant asymptotically so that $G_{\beta}$ has full rank.
\begin{assumptionp}{T3}\label{Tass:SP}
\begin{enumerate}[label=(\roman*)]
\item The true values $(\pi_0, \beta_0)$ lie in the interior of the corresponding parameter spaces, which are compact subspaces in $\R^{d_\pi}$ and $\R^{n-1}$, respectively.

\item There exists a symmetric positive definite matrix $W\in\R^{L(n-1)\times L(n-1)}$ such that $\|\hat W_T-W\|=o_p(1)$.

\item The map $(\beta,\pi)\mapsto g_t^{(Z)}(\beta,\pi)$ is almost surely continuously differentiable in a neighbourhood $\mathcal{N}$ of $(\beta_0,\pi_0)$ and
\[
\E\sup_{(\beta,\pi)\in\mathcal{N}}\left\|\frac{\partial g_t^{(Z)}(\beta,\pi)}{\partial\pi'}\right\|<\infty,
\qquad
\E\sup_{(\beta,\pi)\in\mathcal{N}}\left\|\frac{\partial g_t^{(Z)}(\beta,\pi)}{\partial\beta'}\right\|<\infty.
\]
Moreover, $\E\|g_t^{(Z)}(\beta_0,\pi_0)\|^{2+\delta}<\infty$ for some $\delta>0$.

\item  There exists $i\in\{1,\ldots,L\}$ with $\rho_i:=\E[h_i(z_t)\varepsilon_{1t}]\neq 0$ such that $\rank(G_\beta) = n-1$.
\end{enumerate}
\end{assumptionp}
For simplicity, we present the proof under the i.i.d.\ assumption for $(\varepsilon_t,z_t)$ in Assumption~\ref{Tass:base}(ii), which allows us to write down the covariances transparently. This condition can be relaxed to martingale difference sequences (MDSs), in which case the martingale versions of LLN and CLT yield the same expansions.
Assumption~\ref{Tass:base} implies that $\Gamma:=\E[Y_{t-1}Y_{t-1}']$ is finite and non-singular, see, e.g., Lemma 3.1 in \cite{Lutke_MTSA_book}. Existence of fourth-order moments and nonsingularity of $B_0$ in Assumption~\ref{Tass:base}(ii) also imply that $\E\|s_t(\pi_0,\omega_0)\|^{2}<\infty$. Assumption~\ref{Tass:SP}(iii) is satisfied for the polynomial $h_i$ by construction, and $\E\|g_t^{(Z)}(\beta_0,\pi_0)\|^{2+\delta}<\infty$ holds in particular when $\E|z_t|^{2L(2+\delta)}<\infty$ and $\E\|\varepsilon_t\|^{2(2+\delta)}<\infty$. 
Finally, Assumption~\ref{Tass:SP}(iv) is an identification condition, and under Assumption~\ref{ass:strong_proxy_exogeneity} we have $\rho_1=\E[z_t\varepsilon_{1t}]\neq 0$, so $\rank(G_\beta)=n-1$ and $G_\beta'WG_\beta$ is non-singular.

\begin{proposition}\label{thm:beta_SP}
Let $\Gamma:= \E[Y_{t-1} Y_{t-1}']$, 
$\bs{h}(z_t) = (h_1(z_t), \ldots, h_L(z_t))'$
and define $\rho_i \coloneqq \E[h_i(z_t) \varepsilon_{1t}]$,
$\bs\rho = (\rho_1, \ldots, \rho_L)'$,
$
M \coloneqq  - (\Gamma \otimes I_n).$
Let
\[
G_\pi \coloneqq\E \!\left[\frac{\partial g_t^{(Z)}(\beta_0, \pi_0)}{\partial \pi'} \right] = -\E[\bs{h}(z_t)Y_{t-1}'] \otimes  (\mc J - \beta_0 \iota_1'),
\quad
G_\beta \coloneqq \E\! \left[\frac{\partial g_t^{(Z)}(\beta_0, \pi_0)}{\partial \beta'} \right] =  - \bs\rho \otimes I_{n-1},
\]
and define moment functions $g_t^{(Z), \perp}(\beta_0, \pi_0) \coloneqq g_t^{(Z)}(\beta_0, \pi_0) - G_\pi M^{-1} m_t(\pi_0)$. Then, under Assumptions~\ref{ass:strong_proxy_exogeneity},~\ref{Tass:base} and~\ref{Tass:SP}
$$
\sqrt{T} \left(\hat \beta_T (W) - \beta_0 \right) \ \Rightarrow \ N(\bs 0, V_{\beta}),
\qquad
V_{\beta} = (G_\beta' W G_\beta)^{-1} G_\beta' W  S_{\perp} W  G_\beta (G_\beta' W G_\beta)^{-1},
$$
with
$S_{\perp} \coloneqq \E \left[ g_t^{(Z), \perp}(\beta_0, \pi_0)\,  g_t^{(Z), \perp}(\beta_0, \pi_0 )' \right]$.
Moreover, if $L > 1$, the $J$-test satisfies
\[
J_T(W)\ \Rightarrow\ \sum_{i=1}^{L(n-1)} \lambda_i(W) \zeta_i^2,
\qquad
\text{under} \quad H_0:\ \E[g_t^{(Z)}(\beta_0, \pi_0)] = \bs 0,
\]
where $\zeta_i \stackrel{iid}{\sim} N(0,1)$ are standard normal variables and 
$\{\lambda_i(W)\}_{i=1}^{L(n-1)}$ are the eigenvalues of matrix
$
S_\perp^{1/2}W \left(I_{L(n-1)} - G_\beta(G_\beta' W G_\beta)^{-1} G_\beta' W  \right)S_\perp^{1/2}.
$ Under the fixed alternative $H_1: \inf_{\beta}\,\big\|\E[g_t^{(Z)}(\beta, \pi_0)]\big\| > 0$, $J_T(W) \to \infty$ as $T \to \infty$.
\end{proposition}
\begin{proof}[Proof of Proposition~\ref{thm:beta_SP}]
Under Assumption~\ref{Tass:base}(i)--(ii), the stable VAR admits the strictly stationary causal representation, and an ergodic theorem yields, $\left\|\frac{1}{T}\sum_{t=1}^T Y_{t-1}Y_{t-1}'-\Gamma\right\|=o_p(1),$ with $\Gamma:=\E[Y_{t-1}Y_{t-1}']$ finite and non-singular. Moreover,
Since $u_t=B_0\varepsilon_t$ is independent of $\mathcal{F}_{t-1}:=\sigma(\varepsilon_s,z_s:s\le t-1)$ and $\E[u_t]=0$, the score
$m_t(\pi_0)=(Y_{t-1}\otimes I_n)u_t$ satisfies $\E[m_t(\pi_0)\mid\mathcal{F}_{t-1}]=0$, so $\{m_t(\pi_0)\}$ is a square-integrable martingale difference sequence. Moreover, Assumption~\ref{Tass:base} implies
\[
\frac{1}{\sqrt{T}}\sum_{t=1}^T m_t(\pi_0)\ \Rightarrow\ N(0,S_{mm}),
\qquad
S_{mm}:=\E[m_t(\pi_0)m_t(\pi_0)']=\Gamma\otimes\Sigma_u,
\]
where $\Sigma_u:=\E[u_tu_t']$. Since $\partial m_t(\pi)/\partial\pi'=-(Y_{t-1}Y_{t-1}'\otimes I_n)$ does not depend on $\pi$, it follows that $\left\|\frac{1}{T}\sum_{t=1}^T\frac{\partial m_t(\pi_0)}{\partial\pi'}-M\right\|=o_p(1),$ $M:=-(\Gamma\otimes I_n),$ 
and $M$ is non-singular by the non-singularity of $\Gamma$. Therefore,
\[
\sqrt{T}(\hat\pi_T-\pi_0)=-M^{-1}\frac{1}{\sqrt{T}}\sum_{t=1}^T m_t(\pi_0)+o_p(1).
\]

Next, by standard GMM consistency (compactness in Assumption~\ref{Tass:SP}(i), continuity and identification in Assumptions~\ref{ass:strong_proxy_exogeneity} and~\ref{Tass:SP}(iv)), $\hat\beta_T(W)\overset{p}{\to}\beta_0$. Hence, with probability approaching one, $(\hat\pi_T,\hat\beta_T(W))$ lies in the neighbourhood $\mathcal{N}$ from Assumption~\ref{Tass:SP}(iii), we have
\begin{align*}
	\sqrt{T}(\hat\beta_T(W) - \beta_0) &=
	- \left[\left(\frac{1}{T}\sum_{t=1}^T \frac{\partial g_t^{(Z)}(\hat \pi_T, \bar\beta) }{\partial \beta'} \right)'
	\hat W_T
	\left(\frac{1}{T}\sum_{t=1}^T \frac{\partial g_t^{(Z)}(\hat \pi_T, \bar\beta) }{\partial \beta'} \right)
	\right]^{-1} \times \\
	& \hspace{.7cm} \left(\frac{1}{T}\sum_{t=1}^T \frac{\partial g_t^{(Z)}(\hat \pi_T, \bar\beta) }{\partial \beta'} \right)'
	\hat W_T \frac{1}{\sqrt{T}}\sum_{t=1}^T g_t^{(Z)}(\hat \pi_T, \beta_0) 
	\end{align*}
	where $\bar\beta$ is a mean value between $\hat \beta_T(W) $ and $\beta_0$ and
	\begin{align*}
		\frac{1}{\sqrt{T}}\sum_{t=1}^T g_t^{(Z)}(\hat \pi_T, \beta_0)  &= \frac{1}{\sqrt{T}}\sum_{t=1}^T g_t^{(Z)}(\pi_0, \beta_0) + \left(\frac{1}{T}\sum_{t=1}^T \frac{\partial g_t^{(Z)}(\bar\pi, \beta_0) }{\partial \pi'} \right) \sqrt{T}(\hat \pi_T - \pi_0).
	\end{align*}
	By Assumption~\ref{Tass:SP}(iii), a uniform ergodic theorem yields
	\[
	\left\|\frac{1}{T}\sum_{t=1}^T \frac{\partial g_t^{(Z)}(\bar\pi, \beta_0) }{\partial \pi'} - G_\pi \right\| = o_p(1)
	\quad 
	\left\|\frac{1}{T}\sum_{t=1}^T \frac{\partial g_t^{(Z)}(\hat \pi_T, \bar\beta) }{\partial \beta'} - G_\beta \right\| = o_p(1),
	\]
	where $G_\pi  = -\E[\bs{h}(z_t)Y_{t-1}'] \otimes  (\mc J - \beta_0 \iota_1')$ and $G_\beta=  - \bs\rho \otimes I_{n-1},$ 
    since for $i = 1, \ldots, L$,
    \begin{align*}
     \E\left[\frac{\partial g_{i,t}(\pi, \beta)}{\partial \pi'}\right] &=  \E\left[ h_i(z_t)(\mc J - \beta \iota_1') \frac{\partial }{\partial \pi'} u_t(\pi)\right] =  \E\left[- h_i(z_t)(\mc J - \beta \iota_1') (Y_{t-1}' \otimes I_n)\right]\\
     \E\left[\frac{\partial g_{i,t}(\pi_0, \beta_0)}{\partial \beta'} \right] &= - \E[h_i(z_t) \varepsilon_{1t}] I_{n-1} = -\rho_i I_{n-1}
    \end{align*}
    Therefore, and by \ref{Tass:SP}(iii)
	\begin{align*}
	\sqrt{T}(\hat \beta_T(W) - \beta_0) &= -(G_\beta' W G_\beta)^{-1} G_\beta' W \left(\frac{1}{\sqrt{T}}\sum_{t=1}^T g_t^{(Z)}( \pi_0, \beta_0) + G_\pi \sqrt{T}(\hat \pi_T - \pi_0) \right) + o_p(1)\\
	&=  -(G_\beta' W G_\beta)^{-1} G_\beta' W  \frac{1}{\sqrt{T}}\sum_{t=1}^T \left( g_t^{(Z)}( \pi_0, \beta_0) - G_\pi  M^{-1} m_t(\pi_0)\right) + o_p(1).
	\end{align*}
	Hence, the influence function of the second-step estimator is
	\begin{align*}
	\psi_\beta(y_t, z_t) &= -(G_\beta' W G_\beta)^{-1} G_\beta' W \left(g_t^{(Z)}(  \pi_0, \beta_0) - G_\pi M^{-1} m_t(\pi_0)\right)\\ 
    &=: -(G_\beta' W G_\beta)^{-1} G_\beta' W g_t^{(Z), \perp}(\beta_0, \pi_0).
\end{align*}
	By Assumptions~\ref{Tass:base} and~\ref{Tass:SP}, $\{g_t^{(Z), \perp}(\beta_0, \pi_0)\}_t$ is a square-integrable MDS, and by a CLT, $\frac{1}{\sqrt{T}}\sum_{t=1}^T g_t^{(Z), \perp}(\beta_0, \pi_0) \Rightarrow N(0, S_\perp)$ and hence $	\sqrt{T}(\hat \beta_T(W) - \beta_0) \Rightarrow N(0, V_\beta)$,
where $V_\beta = (G_\beta' W G_\beta)^{-1} G_\beta' W  S_{\perp} W  G_\beta (G_\beta' W G_\beta)^{-1}$.\\
Assume $L>1$. Under $H_0:\E[g_t^{(Z)}(\beta_0,\pi_0)]=0$, we have $\E[g_t^{(Z),\perp}(\beta_0,\pi_0)]=0$.
Expanding the sample moments at $(\hat\beta_T(W),\hat\pi_T)$ yields
\[
\sqrt{T}\,\bar g_T(\hat\beta_T(W),\hat\pi_T)
=
\Big(I_{L(n-1)}-G_\beta(G_\beta'WG_\beta)^{-1}G_\beta'W\Big)
\frac1{\sqrt{T}}\sum_{t=1}^T g_t^{(Z),\perp}(\beta_0,\pi_0)
+o_p(1).
\]
Let $S_{\perp}^{1/2}$ be the symmetric square root of $S_\perp$, define
	$\bs\zeta_T \coloneqq   \frac{1}{\sqrt{T}}\sum_{t=1}^T S_{\perp}^{-1/2} g_t^{(Z), \perp}(\beta_0, \pi_0)$ so that $\bs\zeta_T \Rightarrow  N(0, I_{L(n-1)})$ and $R(W) \coloneqq I_{L(n-1)} - G_\beta (G_\beta' W G_\beta)^{-1} G_\beta' W $. Then, the $J$-statistic can be written as
	$
	J_T(W) = \bs\zeta_T' S_\perp^{1/2} R(W)' W R(W) S_\perp^{1/2} \bs\zeta_T + o_p(1).
	$
    Note that $W$ is symmetric, $R(W)' W  = W R(W)$, and $R(W)$ is idempotent. Since $\bs\zeta_T \Rightarrow  N(0, I_{L(n-1)})$, 
	\[
	\bs\zeta_T' S_\perp^{1/2} W R(W) S_\perp^{1/2} \bs\zeta_T \ \Rightarrow \ \sum_{i=1}^{L(n-1)} \lambda_i(W) \zeta_i^2,
	\] where $\{\zeta_i\}_{i=1}^{L(n-1)}$ are independent standard normal random variable and $\{\lambda_i(W)\}_{i=1}^{L(n-1)}$ are the eigenvalues of $S_\perp^{1/2} W R(W) S_\perp^{1/2}$.
    Under the fixed alternative $H_1:\inf_{\beta}\|\E[g_t^{(Z)}(\beta,\pi_0)]\|>0$,
by the consistency theorem,  $\|\bar g_T(\hat\beta_T(W),\hat\pi_T)\| \overset{p}{\to}\inf_\beta\|\E[g_t^{(Z)}(\beta,\pi_0)]\|>0$. Since $W$ is positive definite, $J_T(W) 
=T\cdot \|\bar g_T(\hat\beta_T(W),\hat\pi_T)\|_W^2\overset{p}{\to}\infty$ as $T \to \infty$.
\end{proof}

The weighting $S_{gg}^{-1}$ chosen in previous studies would be optimal if the true value of $\pi$ were known and the VAR residuals $u_t = B_0\varepsilon_{t}$ were perfectly observed by the researchers. In general, since $u_t(\hat \pi_T)$ is retrieved conditional on the estimate $\hat \pi_T$, the naive weighting $S_{gg}^{-1}$ is suboptimal in two senses. First, there exists an optimal weighting that generally differs from $S_{gg}^{-1}$, which yields an efficient estimator (with the smallest asymptotic variance among all GMM estimators based on the same moment conditions). Second and more importantly, the statistic $J_T(S_{gg}^{-1})$ does not converge to the usual $\chi^2$ distribution under the null, such that a test using $\chi^2$ critical values does not have the correct size, not even asymptotically. Instead, the optimal weighting is obtained by taking the inverse of the covariance of the moment function $g_t^{(Z), \perp}$, which is the part of the original moment function $g_t^{(Z)}$ orthogonal to the first-step moment function. Using the optimal weighting will also restore the usual $\chi^2$ null distribution of the $J$-statistic, so that the exogeneity test will have the correct size.
The proposition also suggests that for a fixed $W$, the critical value might be simulated based on estimated $\lambda_i(W)$.\footnote{In additional MC experiments, this procedure with identity weight yields rejection frequencies generally close to the nominal level for the SP test, whereas the analogous ISP test introduced below over-rejects substantially at the sample sizes considered. The distortion becomes negligible for large \(T\), we therefore recommend the feasible optimal weighting for testing.}

\begin{corollary}
The asymptotically efficient GMM choice is $ S_\perp^{-1}$ such that
\[
\sqrt{T} \left(\hat \beta_T (S_\perp^{-1}) - \beta_0 \right) \ \Rightarrow \ N \left(\bs 0, (G_\beta' S_{\perp}^{-1} G_\beta)^{-1}\right),
\]
Moreover, if $L>1$, under $H_0$,
$
J_T\left(S_\perp^{-1}\right)\ \Rightarrow\ \chi^2_{(L-1)(n-1)},
$ while under the fixed alternative $H_1$,  $\Pr\big(J_T\left(S_\perp^{-1}\right) > c^*_{1-\alpha}\big)\to 1$,
where $c^*_{1-\alpha}$ is the $(1-\alpha)$-quantile of $\chi^2_{(L-1)(n-1)}$.
\end{corollary}
The corollary follows by substituting the optimal weighting $S_\perp^{-1}$ into
Proposition~\ref{thm:beta_SP} and noting that $S_\perp^{-1/2}G_\beta(G_\beta'S_\perp^{-1}G_\beta)^{-1}G_\beta' S_\perp^{-1/2}$
is symmetric idempotent with rank $n-1$, so that
$I_{L(n-1)}-S_\perp^{-1/2}G_\beta(G_\beta'S_\perp^{-1}G_\beta)^{-1}G_\beta' S_\perp^{-1/2}$ is symmetric idempotent with rank $L(n-1)-(n-1)$.

\subsubsection{Identification-robust test under weak proxies}
\label{app:weakSP}
As mentioned in Section~\ref{sec:SP_weak}, Assumption~\ref{ass:weak1} does not, by itself, restrict the relevance of the higher-order SPs. It may occur that the higher-order SPs can be strictly relevant although $\rho_{1,T}$ vanishes in the limit. In Appendix~\ref{app:proof}, we give an example in which $\rho_{1,T} = O(T^{-1/2})$ while $\rho_{2,T} = O(1)$. 
This is an important point because identification of $\beta_0$ is governed by $G_{\beta, T}=-(\bs \rho_T\otimes I_{n-1})$, and (strong) identification is provided if $\rank(G_{\beta, T}) = n-1$ (see Assumption~\ref{Tass:SP}(iv)). If there exists at least one $i = 1, \ldots, L$, such that $\rho_{i,T} \to \rho_{i, \infty}$ and $\rho_{i, \infty} \neq 0$, 
we have $\bs \rho_T\to \bs\rho_\infty$ with $\rank(\bs \rho_\infty\otimes I_{n-1})=n-1$, then the rank condition in Assumption~\ref{Tass:SP}(iv) is respected.  Consequently, the standard inference and  $J$-tests discussed in the previous sections remain applicable. 

This suggests a form of robustness of the SP approach under the usual
weak proxy setting, as $\beta$ can remain strongly identified 
even when the correlation between the proxy $z_t$ and the target shock vanishes in the limit.
The genuinely weak identification problem arises only in the stronger sense that the entire $\bs \rho$ vector is jointly local to zero:
\begin{assumptionp}{W3}\label{ass:weak2}
There exists $\bs \gamma \in \R^L$, such that $\sqrt{T} \bs \rho_{T} \to \bs \gamma$ as $T\to \infty$.
\end{assumptionp}
The vector $\bs \gamma$ is allowed to be zero, so that the boundary case of a flat population criterion is also nested within our analysis.
The relevant case is demonstrated in Lemma~\ref{lem:weak_lin_proxy} in Appendix, where we show that for a classical linear measurement-error proxy $z_{t,T} = \phi_T \varepsilon_{1t} + v_t$ with $\phi_T = \phi^*/\sqrt{T}$, the relevance of all SPs is at most of order $O(T^{-1/2})$ under \ref{ass:weak1}. 
However, it is worth highlighting that the speed, at which $\rho_{i,T}$ vanishes, does not increase monotonically with $i$. For example, the lemma shows $\rho_{3,T} = O(T^{-1/2})$ so that both $\rho_{1,T}$ and $\rho_{3,T}$ decrease at the same rate.

Under Assumption~\ref{ass:weak2}, $G_{\beta, T} = O(T^{-1/2})$, the population criterion is nearly flat in the $\beta$ direction and the standard asymptotics no longer apply. Let 
\[
\mc G_T(\beta)
:= \frac1{\sqrt T}\sum_{t=1}^T
\left[
g_t^{(Z)}(\beta,\pi_0)
-\E [g_t^{(Z)}(\beta,\pi_0)]
-G_{\pi}(\beta)M^{-1}m_t(\pi_0)
\right].
\]
be a random $\R^{L(n-1)}$-valued function continuous on $\mc B$. We make the following assumptions:
\begin{assumptionp}{T3-W}
\label{Tass:SPweak}
\begin{enumerate}[label=(\roman*)]
\item The true values $(\pi_0, \beta_0)$ lie in the interior of compact parameter spaces, and in particular, $\beta_0 \in \text{int}(\mc B) \subset \R^{n-1}$ with $\mc B$ compact.

\item 
$\mc G_T(\beta) \ \Rightarrow \ \mc G(\beta),$ 
where $\mc G(\beta)$ is a zero-mean Gaussian stochastic process with covariance function
$\Omega(\beta_1, \beta_2) := \E[\mc G(\beta_1)\mc G(\beta_2)']$, and $\Rightarrow$ denotes
weak convergence of random functions on $\mc B$ with respect to the sup norm, i.e., $d_\infty(f_1,f_2):=\sup_{\beta\in\mathcal B}\|f_1(\beta)-f_2(\beta)\|$. The map $\beta\mapsto\Omega(\beta)$ is continuous and uniformly positive definite on $\mathcal B$.
\item 
The map $(\beta,\pi)\mapsto g_t^{(Z)}(\beta,\pi)$ is almost surely continuously differentiable on $\mc B \times \mc N_{\pi}$ with $\mc N_{\pi}$ being a neighborhood of $\pi_0$ and
$\sup_T\E_T\sup_{\beta \in \mc B, \pi \in \mc N_\pi}\left\|\frac{\partial g_t^{(Z)}(\beta,\pi)}{\partial\pi'}\right\|<\infty.$

\item There exists a feasible estimator $\hat\Omega_T(\beta)$ s.t., $\sup_{\beta\in\mathcal B}
        \|\hat\Omega_T(\beta)-\Omega(\beta)\|
        =o_p(1).$ 

\end{enumerate}
\end{assumptionp}
\ref{Tass:SPweak}(ii) is the analogue of the functional central limit condition in
\cite{StockWright2000}, while the influence function for the structural parameter of interest in our setting also includes the uncertainty coming from the estimation of the reduced-form parameter $\pi$. 
The condition would coincide with that in \cite{StockWright2000} if the adaptivity condition in 
Proposition~\ref{thm:adapt_SP} is satisfied or one uses the efficient synthetic proxies as described in 
Appendix~\ref{app:adaptivity}.
The weak convergence can be established by finite-dimensional convergence and stochastic equicontinuity,
which can be obtained under weaker primitive conditions since both $g_t^{(Z)}(\beta,\pi_0)$ and $G_{\pi,T}(\beta)$ are affine in $\beta$ (see Assumption B and B' in \citealt{StockWright2000}).
\ref{Tass:SPweak}(iii) strengthens \ref{Tass:SP} from a neighborhood of $\beta_0$ to 
the compact set $\mc B$, which is required here since under weak proxy,
the GMM minimizer $\hat\beta_T$ need not lie in the neighborhood of $\beta_0$. 

Define the (rescaled) sample objective of an efficient GMM based on SPs 
\[
Q_T(\beta) :=
        T\bar g_T^{(Z)}(\beta,\hat\pi_T)'
        \hat\Omega_T(\beta)^{-1}
        \bar g_T^{(Z)}(\beta,\hat\pi_T),
\]
with $\hat\Omega_T(\beta) \to \Omega(\beta)$ uniformly on $\mc B$ in probability. 
The following Proposition provides the weak limit of the objective function. 

\begin{proposition}
\label{thm:weak_Q}
Under Assumptions~\ref{ass:strong_proxy_exogeneity},~\ref{ass:weak2},~\ref{Tass:base}, and
\ref{Tass:SPweak},
\begin{align*}
    Q_T(\beta) \ &\Rightarrow \ \left[
        \mc G(\beta)+(\bs\gamma\otimes I_{n-1})(\beta_0-\beta)
        \right]'
        \Omega(\beta)^{-1}
        \left[
        \mc G(\beta)+(\bs\gamma\otimes I_{n-1})(\beta_0-\beta)
        \right]\\ 
        &:= Q_\infty(\beta),
\end{align*}
on $\mc B$ with respect to the sup norm, and hence, at the true value, $Q_T(\beta_0)  \Rightarrow  \chi^2_{L(n-1)}.$
\end{proposition}
\begin{proof}[Proof of Proposition~\ref{thm:weak_Q}]
By~\ref{Tass:SPweak}(iii), a uniform LLN entails that
\[\sup_{\beta\in\mathcal B}
\left\|
\sqrt T \left( \bar g_T^{(Z)}(\beta,\hat\pi_T)
-\bar g_T^{(Z)}(\beta,\pi_0) \right)
+
G_{\pi}(\beta)M^{-1}
\frac1{\sqrt T}\sum_{t=1}^T m_t(\pi_0)
\right\|
=o_p(1),\]
hence it holds uniformly on $\mc B$ that
\begin{align*}
    \sqrt T\bar g_T^{(Z)}(\beta,\hat\pi_T)
&=
\sqrt T\bar g_T^{(Z)}(\beta,\pi_0)
-
G_{\pi}(\beta)M^{-1}
\frac1{\sqrt T}\sum_{t=1}^T m_t(\pi_0)
+
o_p(1)\\
&= \frac1{\sqrt T}\sum_{t=1}^T
\left[
g_t^{(Z)}(\beta,\pi_0)
-\E [g_t^{(Z)}(\beta,\pi_0)] 
-G_{\pi}(\beta)M^{-1}m_t(\pi_0) \right] \\
& \hspace{3ex} + \sqrt{T} \E [g_t^{(Z)}(\beta,\pi_0)] +o_p(1) \\
&=: \mc G_T(\beta) + \sqrt{T} \E [g_t^{(Z)}(\beta,\pi_0)] +o_p(1).
\end{align*}
Assumption~\ref{ass:strong_proxy_exogeneity} implies that $\forall 1\leq i \leq L$,
$\E[h_i(z_t)r_t(\beta,\pi_0)] = (\beta_0-\beta)\E[h_i(z_t)\varepsilon_{1t}]$,
and hence Assumption~\ref{ass:weak2} entails that
\[
        \sqrt T\,\E[g_t^{(Z)}(\beta,\pi_0)]
        =
        \sqrt T (\bs\rho_T\otimes I_{n-1})(\beta_0-\beta) \to
        (\bs\gamma\otimes I_{n-1})(\beta_0-\beta)
\]
uniformly on $\mc B$. 
Assumptions \ref{Tass:SPweak}(ii) and (iv), uniform positive definiteness of $\Omega(\beta)$ and the continuous mapping theorem 
yield the stated convergence of the objective process. Finally, if $\beta=\beta_0$, the local drift in the limit vanishes, s.t.,
$\sqrt T\bar g_T^{(Z)}(\beta_0,\hat\pi_T) \Rightarrow \mc G(\beta_0)$, 
and consequently,
\[
        Q_T(\beta_0)
        \ \Rightarrow \
        \mc G(\beta_0)'\Omega(\beta_0)^{-1} \mc G(\beta_0)
        \sim
        \chi^2_{L(n-1)}.
\]
\end{proof}

In general, following Proposition~\ref{thm:weak_Q}, the GMM estimator need not be consistent and the $J$-test statistics do not have the usual $\chi^2_{(L-1)(n-1)}$ distribution under the null, even when efficient weighting is adopted.
The next Corollary summarizes the limiting behavior of GMM estimators, where the efficient weighting is
either estimated in two steps or via continuous updating, as well as their associated $J$-test statistics.
Denote the initial estimator $\hat \beta_{1,T}$ obtained as $\argmin_{\beta\in\mathcal B} T\bar g_T^{(Z)}(\beta,\hat\pi_T)'W_1 \bar g_T^{(Z)}(\beta,\hat\pi_T)$ for some fixed positive definite $W_1$ (e.g., $I_{L(n-1)}$) and the two-step efficient weighting as  $\hat\Omega_T(\hat\beta_{1,T})^{-1}$. 
The two-step estimator, denoted $\hat \beta_{2,T}$ is obtained  as $\argmin_{\beta\in\mathcal B} T\bar g_T^{(Z)}(\beta,\hat\pi_T)' \hat\Omega_T(\hat\beta_{1,T})^{-1} \bar g_T^{(Z)}(\beta,\hat\pi_T)$, while the continuous updating estimator, denoted $\hat \beta_{CUE,T}$, is obtained as $\argmin_{\beta\in\mathcal B} Q_T(\beta)$. 
Let $J_{2,T}$ and $J_{CUE, T}$ be the $J$-test statistics based on the two-step and the continuous updating estimator, respectively.
\begin{corollary}\label{cor:weak_failure}
Under assumptions of Proposition~\ref{thm:weak_Q}, 
\begin{align*}
    \hat\beta_{1,T} \ \Rightarrow \ \beta_1^\dagger, 
    \qquad
\hat\beta_{2,T} \ \Rightarrow \ \argmin_{\beta\in\mathcal B} Q_\infty^\dagger(\beta),
    \qquad
\hat\beta_{CUE,T} \ \Rightarrow \ \argmin_{\beta\in\mathcal B} Q_\infty(\beta),
\end{align*}
where $\beta_1^\dagger := \argmin_{\beta\in\mathcal B}
        \left[\mc G(\beta)+(\bs\gamma\otimes I_{n-1})(\beta_0-\beta)\right]'
        W_1
        \left[\mc G(\beta)+(\bs\gamma\otimes I_{n-1})(\beta_0-\beta)\right]$
and\\ 
$Q_\infty^\dagger(\beta) := \left[\mc G(\beta)+(\bs\gamma\otimes I_{n-1})(\beta_0-\beta)\right]'
        \Omega(\beta_1^\dagger)^{-1}
        \left[\mc G(\beta)+(\bs\gamma\otimes I_{n-1})(\beta_0-\beta)\right]$. 
Moreover, 
\[
J_{2,T} \ \Rightarrow \ \inf_{\beta \in\mathcal B} Q_\infty^\dagger(\beta)
\qquad
J_{CUE,T} \ \Rightarrow \ \inf_{\beta \in\mathcal B} Q_\infty(\beta).
\]
\end{corollary}
The corollary follows from Proposition~\ref{thm:weak_Q} by the same argument used in Corollary 4 in \cite{StockWright2000}. In particular, the limiting distributions of $J_{2,T}$ and $J_{CUE,T}$ 
are nonstandard and depend on the nuisance parameter such as $\bs \gamma$, 
which is not known to the econometrician in practice.
As a consequence, standard inference that involves rejecting the null when the $J$-statistic exceeds
the $\chi^2_{(L-1)(n-1)}$ critical value is not valid.
With this being said, Proposition~\ref{thm:weak_Q} shows that the sample objective function converges to a standard $\chi^2$ distribution at the true value $\beta_0$, despite being weakly identified, 
provided that an efficient weighting $\Omega(\beta)^{-1}$ is employed. This gives rise to a weak proxy robust test based on inverting the objective function. 
Specifically, for a nominal level $\alpha$, define the set $\mc C_T(1-\alpha)$ and test statistic $\varphi_{T}(\alpha)$:
\begin{align}\label{eq:CT_varphi}
\mc C_T(1-\alpha)
    :=
    \left\{
    \beta\in\mathcal B:
    Q_T(\beta)\le \chi^2_{L(n-1),1-\alpha}
    \right\},
    \qquad
\varphi_{T}(\alpha)
    :=
\mathbf 1 \left\{ \mc C_T(1-\alpha)=\varnothing \right\},
\end{align}
where $\mathbf 1\{.\}$ is the indicator function.
The null hypothesis of strong proxy exogeneity is rejected if $\varphi_{T}(\alpha) = 1$, i.e., 
if the moment restrictions implied by SPs are incompatible with the data such that
the set $\mc C_T(1-\alpha)$ is empty.

Next, we show that the test size based on $ \varphi_{T}(\alpha)$ is asymptotically controlled -- albeit conservative -- under weak identification and consistent against fixed alternatives.
\begin{proposition}
\label{thm:weak_test}
Under assumptions of Proposition~\ref{thm:weak_Q}, under $H_0:\ \E[g_t^{(Z)}(\beta_0, \pi_0)] = \bs 0,$
\[
\lim_{T\to\infty} \Pr \left(\beta_0\in \mc C_T(1-\alpha) \right)  = 1-\alpha,
\qquad
\limsup_{T\to\infty} \Pr \left(\varphi_{T}(\alpha)=1 \right)
        \le \alpha.
\]
Under the fixed $H_1: \liminf_{T \to \infty} \inf_{\beta}\,\big\|\E[g_t^{(Z)}(\beta, \pi_0)]\big\| > 0$, we have
$\Pr \left(\varphi_{T}(\alpha)=1 \right)\to 1.$
\end{proposition}
\begin{proof}[Proof of Proposition~\ref{thm:weak_test}]
By Proposition~\ref{thm:weak_Q}, $Q_T(\beta_0)  \Rightarrow  \chi^2_{L(n-1)}$ and therefore,
\[
\Pr \left(\beta_0\in \mc C_T(1-\alpha) \right) = 
\Pr \left( Q_T(\beta_0)\le \chi^2_{L(n-1),1-\alpha} \right) \to 1-\alpha,
\]
under the null. Next, since $\beta_0\in\mc B$, the event $\inf_{\beta\in\mathcal B}Q_T(\beta)> \chi^2_{L(n-1),1-\alpha} $ implies $Q_T(\beta_0)> \chi^2_{L(n-1),1-\alpha}$, i.e.,
\[
\left\{ \inf_{\beta\in\mathcal B}Q_T(\beta)> \chi^2_{L(n-1),1-\alpha} \right\}
\subseteq
\left\{ Q_T(\beta_0)> \chi^2_{L(n-1),1-\alpha} \right\},
\]
and the monotonicity of the probability measure entails that
\[
\Pr \left(\varphi_{T}(\alpha)=1 \right)
\le
\Pr \left( Q_T(\beta_0)>  \chi^2_{L(n-1),1-\alpha} \right) \ \to \ \alpha,
\]
which shows the asymptotic size control. Note that the inequality can be strict, so the test is asymptotically  conservative. 
For the consistency statement, under the fixed alternative, there exists $c_1 > 0$, s.t., $ \inf_{\beta}\,\big\|\E[g_t^{(Z)}(\beta, \pi_0)]\big\| \geq  c_1$ for all sufficiently large $T$. 
By the Assumption~\ref{Tass:SPweak}(iii), (iv), a uniform LLN entails that
$\sup_{\beta\in\mc B}\left\| \bar g_T^{(Z)}(\beta,\hat\pi_T)- \E[g_t^{(Z)}(\beta, \pi_0)] \right\| =o_p(1)$, which implies that $\inf_{\beta\in\mathcal B} \left\|\bar g_T^{(Z)}(\beta,\hat\pi_T) \right\| \ge c_1/2$ with probability approaching one (w.p.a.o). 
By Assumption~\ref{Tass:SPweak}(iv) and uniform positive definiteness of $\Omega(\beta)$ on the compact set $\mc B$, the eigenvalues of the estimated efficient weighting are strictly positive, so there exists a constant $c_2 > 0$, s.t.,
$\lambda_{\min}\left (\hat\Omega_T(\beta)^{-1}\right) \ge  c_2 $  w.p.a.o.\ uniformly on $\mc B$.
Hence,
\[
\inf_{\beta\in\mathcal B}Q_T(\beta) \ge
T c_2 \inf_{\beta\in\mathcal B} \left\| \bar g_T^{(Z)}(\beta,\hat\pi_T) \right\|^2
\ge T c_2 c_1^2/4 \to\infty
\]
in probability, and therefore $\Pr \left(\varphi_{T}(\alpha)=1 \right)\to 1$, which proves consistency of the test.
\end{proof}

For given $\alpha$, if the set $\mc C_T(1-\alpha)$ is empty, the test rejects the strong proxy exogeneity. Note that $\Pr \left(\varphi_{T}(\alpha)=1 \right)$ can also be smaller than $\alpha$ under the null in large samples. Hence, the test is asymptotically conservative. 
If the set $\mc C_T(1-\alpha)$ is not empty, the test fails to reject the null, which however 
does not necessarily lead to the conclusion that the strong proxy exogeneity holds.
A failure of rejection can reflect either validity of the moment restrictions or 
that the data contain little information about the weakly identified parameter.
If the SPs are jointly weak or not relevant at all ($\bs \gamma = \bs 0$, i.e., lack of identification), the
set can be large or contain the whole parameter space $\mc B$, respectively.

\subsection{Adaptive testing}
\label{app:adaptivity}
Proposition~\ref{thm:beta_SP} also reveals an interesting adaptivity property. We next derive the necessary and sufficient condition for $G_\pi=\bs 0$, under which estimation of the reduced form does not affect the estimation and testing of the structural parameter.
\begin{proposition}\label{thm:adapt_SP}
   Under the conditions in Proposition~\ref{thm:beta_SP}, $G_\pi=\bs 0$ if and only if $\E[\bs h(z_t)Y_{t-1}']=\bs 0.$
   In this case, $\hat\beta_T(W)$ is adaptive such that the reduced-form parameter $\pi_0$ can be treated as if its true value were known. Moreover, $S_\perp=S_{gg}$, $\hat\beta_T(S_{gg}^{-1})$ is asymptotically efficient and, for $L>1$, $J_T(S_{gg}^{-1})\Rightarrow\chi^2_{(L-1)(n-1)}$ under $H_0$.
\end{proposition}
\begin{proof}[Proof of Proposition~\ref{thm:adapt_SP}]
By Proposition~\ref{thm:beta_SP}, the influence function of $\hat\beta_T(W)$ does not depend on $\pi$
if and only if $G_\pi=0$. For $n\ge2$, we show that $\mc J-\beta_0\iota_1'\neq0$.
Indeed, for any $j\in\{2,\ldots,n\}$, the $(j-1)$-th row of $\mc J-\beta_0\iota_1'$
equals $\iota_j'-\beta_{0,j-1}\iota_1'$, whose $j$-th entry is unity, i.e.,
$(\iota_j'-\beta_{0,j-1}\iota_1')\iota_j=1$,
so this row is nonzero. Therefore $\mc J-\beta_0\iota_1'\neq0$.
Since
$
G_\pi = -\E[\bs{h}(z_t)Y_{t-1}'] \otimes (\mc J-\beta_0\iota_1'),
$
it follows that $G_\pi=0$ if and only if $\E[\bs h(z_t)Y_{t-1}']=0$.
The remaining claims follow by substituting $G_\pi=0$ into the expressions for
$g_t^{(Z),\perp}$ and $S_\perp$, i.e., 
$g_t^{(Z),\perp}(\beta_0,\pi_0)=g_t^{(Z)}(\beta_0,\pi_0)$ and hence $S_\perp=S_{gg}$.
\end{proof}
When adaptivity applies,  orthogonality to the nuisance tangent space would lead to the same limiting behavior of $\hat \beta_T$, as if the reduced-form parameter $\pi$ were perfectly known. For example, if the proxy is strongly lag-exogenous $\E[\varepsilon_s | z_t]=0$ for $s<t$ and the VAR process (for demeaned data) has no deterministic term s.t.\ $\E[Y_{t-1}] = 0$, then $\E[\bs h(z_t)Y_{t-1}']=\bs 0$ and adaptivity holds.\footnote{Note that for such a model, the regressor vector reads as $Y_{t-1} = (y_{t-1}',\ldots,y_{t-p}')'$.} But Proposition~\ref{thm:adapt_SP} also suggests an easy way to construct \emph{efficient synthetic proxies}, defined in population as
\[
\bs h^\perp(z_t)  :=\ \bs h(z_t)-\E[\bs h(z_t)Y_{t-1}']\,\Gamma^{-1}Y_{t-1}
\]
which will satisfy $\E[\bs h^\perp(z_t) Y_{t-1}'] = \bs 0 $ and lead to adaptivity automatically. In practice, we can obtain efficient SPs by  regressing synthesized instruments on the VAR regressor and taking the residual from the regression as new instruments. Replacing $\bs h(z_t)$ by $\bs h^\perp(z_t) $ leads to an adaptive procedure for estimation and testing, so that the naive weighting is optimal and yields the correct size in over-identification tests.

For testing with ISP, the next corollary shows that, for the case when the researcher employs a single quadratic ISP ($L=2$), a simple set of primitive conditions yields \emph{joint adaptivity} with respect to $(\pi,\omega)$, so that their estimation uncertainty becomes asymptotically negligible.
\begin{corollary}\label{cor:ISP_adapt}
Let $L=2$ with moment stack
$
g_t^{(Z^*)}(\beta,\pi,\omega)
=
\big(g_{1,t}^*(\beta,\pi)',\,g_{2,t}^*(\beta,\pi,\omega)'\big)',
$
where $g_{1,t}^*(\beta,\pi)=z_t r_t(\beta,\pi)$ and
$g_{2,t}^*(\beta,\pi,\omega)=\big(z_t^*(\beta,\pi,\omega)\big)^2 r_t(\beta,\pi)$.
Under the conditions in Proposition~\ref{thm:beta_ISP} and suppose, in addition, that
i) $\E[z_t Y_{t-1}']=\bs 0$,
ii) VAR has no deterministic terms s.t., $\E[Y_{t-1}]=\bs 0$,
and
iii)  $\E[\bs \varepsilon_{2t}\bs \varepsilon_{2t}'\mid \varepsilon_{1t}]=\Sigma_{22}$,
then, the GMM estimator $\hat\beta_T^*$ based on $g_t^{(Z^*)}(\beta,\hat\pi_T,\hat\omega_T)$ is \emph{jointly adaptive} with respect to $(\pi,\omega)$ in the sense that $G_\pi^*=\bs 0$ and $G_\omega^*=\bs 0$,
so that the influence function of $\hat\beta_T^*$ coincides with the influence function obtained when $(\pi_0,\omega_0)$ are treated as known. Consequently, $S_\perp^* = S_{gg}^*:=\E[g_t^{(Z^*)}(\beta_0,\pi_0,\omega_0)g_t^{(Z^*)}(\beta_0,\pi_0,\omega_0)']$ and $S_{gg}^{*-1}$ is the optimal weighting and $J_T^*\left(S_{gg}^{*-1}\right) \ \Rightarrow\ \chi^2_{(L-1)(n-1)}$ under $H_0$.
\end{corollary}
The corollary follows from Lemma~\ref{lem:G_omega_star} and Lemma~\ref{lem:G_pi_star}  in Appendix~\ref{app:proof}.

\subsection{Feasible estimation of the optimal ISP weighting matrix}\label{app:weighting_est}
This appendix details the computation of the feasible counterpart of the efficient weighting matrix
$S_\perp^{*-1}$ for the ISP-GMM estimator. The key object is the plug-in covariance of the
orthogonalized moments
\[
\hat S_{\perp,T}^* \;:=\; \frac1T\sum_{t=1}^T
\hat g_{t,T}^{(Z^*),\perp}\,\hat g_{t,T}^{(Z^*),\perp\,\prime}
\]
where the orthogonalized moment contribution at $t$ is
\[
\hat g_{t,T}^{(Z^*),\perp}
:=
g_t^{(Z^*)}(\tilde\beta_T^*,\hat\pi_T,\hat\omega_T)
-\hat G_{\pi,T}^*\,\hat M_T^{-1}\,m_t(\hat\pi_T)
+\hat G_{\omega,T}^*\,s_t(\hat\pi_T,\hat\omega_T),
\]
and $\tilde\beta_T^*$ is any preliminary consistent estimator of $\beta$ (e.g., one-step GMM
under identity or naive weighting).

\begin{enumerate}
    \item Estimate the reduced-form parameters $(\hat \pi_T, \hat \omega_T)$:  
    \[
    \hat \pi_T = [(YY')^{-1}Y\otimes I_n]\b y 
    \quad
    \text{and}
    \quad
    \hat \omega_T = \vect{\hat \Sigma_u},
    \, \text{ with }
    \hat \Sigma_u = \frac1T\sum^T_{t=1} u_t(\hat \pi_T) u_t(\hat \pi_T)'.
    \]
    \item Compute a preliminary GMM estimator $\tilde\beta_T^*$ with identity or naive weighting.
    \item Compute the moment contributions for each $t = 1, \ldots, T$:
        \[
        g_t^{(Z^*)}(\tilde \beta_T^*, \hat \pi_T, \hat \omega_T):=\big(g_{1,t}^*(\tilde \beta_T^*, \hat \pi_T)',g_{2,t}^*(\tilde \beta_T^*, \hat \pi_T, \hat \omega_T)',\ldots,g_{L,t}^*(\tilde \beta_T^*, \hat \pi_T, \hat \omega_T)'\big)',
        \]
        where $g_{1,t}^*(\tilde \beta_T^*, \hat \pi_T):=h_1(z_t) r_t(\tilde \beta_T^*, \hat \pi_T) = z_t r_t(\tilde \beta_T^*, \hat \pi_T)$
        is the original (external) proxy moment (which does not depend on $\omega$) and for $i=2,\ldots,L$,
        \[
        g_{i,t}^*(\tilde \beta_T^*, \hat \pi_T, \hat \omega_T):=h_i (\hat z_t^*)\,r_t(\tilde \beta_T^*, \hat \pi_T),
        \]
        with ISP given by
        \[
        \hat z_t^*:= \hat z_t^*(\tilde \beta_T^*, \hat \pi_T, \hat \omega_T) = \frac{b(\tilde \beta_T^*)'\Sigma_u(\hat \omega_T)^{-1}u_t(\hat \pi_T)}{\dot\sigma_1(\tilde \beta_T^*, \hat \omega_T)}.
        \]
    Also compute the $\pi$- and $\omega$-moment contributions:
    \[
    m_t(\hat \pi_T) =(Y_{t-1} \otimes I_n)  u_t(\hat \pi_T) 
    \quad
    \text{and}
    \quad
    s_t(\hat\pi_T,\hat\omega_T) = \vect{ u_t(\hat \pi_T) u_t(\hat \pi_T)' - \Sigma_u(\hat \omega_T)}.
    \]
    \item Compute Jacobian estimates
    \begin{enumerate}
        \item Compute the $d_\pi\times d_\pi$ matrix: $\hat M_T = -(\hat \Gamma_T \otimes I_n)$, where $\hat \Gamma_T = \frac1T \sum^T_{t=1} Y_{t-1} Y_{t-1}'$ 
        \item Compute the $L(n-1) \times d_\pi$ matrix:
        \[
        \hat G_{\pi,T}^* = \frac1T \sum^T_{t=1} 
        \begin{pmatrix}
            \frac{\partial g_{1,t}^*(\tilde\beta_T^*,\hat\pi_T)}{\partial\pi'}\\
            \frac{\partial g_{2,t}^*(\tilde\beta_T^*,\hat\pi_T,\hat\omega_T)}{\partial\pi'}\\
            \vdots\\
            \frac{\partial g_{L,t}^*(\tilde\beta_T^*,\hat\pi_T,\hat\omega_T)}{\partial\pi'}
        \end{pmatrix},
        \quad 
        \text{where }
        \frac{\partial g_{1,t}^*(\tilde\beta_T^*,\hat\pi_T)}{\partial\pi'} = -z_t(Y_{t-1}'\otimes(\mc J-\tilde\beta_T^* \iota_1'))
        \]
        and for $i \geq 2$,
        \begin{align*}
            \frac{\partial g_{i,t}^*(\tilde\beta_T^*,\hat\pi_T,\hat\omega_T)}{\partial\pi'} &=
        -(\hat z_t^*)^i (Y_{t-1}'\otimes(\mc J- \tilde \beta_T^* \iota_1'))\\
        &\hspace{3ex} -\frac{i}{\dot\sigma_1(\tilde \beta_T^*, \hat \omega_T)}(\hat z_t^*)^{i-1}\,r_t(\tilde \beta_T^*, \hat \pi_T)\,(Y_{t-1}'\otimes b(\tilde \beta_T^*)'\hat \Sigma_u^{-1}).
        \end{align*}
        \item Compute the $L(n-1) \times n^2$ matrix:
        \[
        \hat G_{\omega,T}^*
        =
        \frac1T\sum_{t=1}^T
        \begin{pmatrix}
        0\\
        \frac{\partial g_{2,t}^*(\tilde\beta_T^*,\hat\pi_T,\hat\omega_T)}{\partial\omega'}\\
        \vdots\\
        \frac{\partial g_{L,t}^*(\tilde\beta_T^*,\hat\pi_T,\hat\omega_T)}{\partial\omega'}
        \end{pmatrix},
        \]
        where for $i \geq 2$,
        \begin{align*}
        \frac{\partial g_{i,t}^*(\tilde \beta_T^*, \hat \pi_T,\hat \omega_T)}{\partial\omega'} &=
        i (\hat z_t^*)^{i-1}\, r_t(\tilde \beta_T^*, \hat \pi_T)
        \Bigg[-\frac{1}{\dot\sigma_1(\tilde \beta_T^*, \hat \omega_T)}(\hat u_t' \hat \Sigma_u^{-1} \otimes b(\tilde \beta_T^*)' \hat\Sigma_u^{-1})  \\
        &\hspace{22ex} + \frac{\hat z_t^*}{2 \dot\sigma_1^2(\tilde \beta_T^*, \hat \omega_T)}
        (b(\tilde \beta_T^*)' \hat\Sigma_u^{-1} \otimes b(\tilde \beta_T^*)' \hat\Sigma_u^{-1})
        \Bigg]
        \end{align*}
    \end{enumerate}
    \item Assemble orthogonalized moments and estimate $\hat S_{\perp, T}^*$. A feasible efficient estimator is then obtained by re-estimating $\beta$ via GMM using weighting $\hat S_{\perp, T}^{*-1}$.
\end{enumerate}

\subsection{Further applications}\label{app:empirical}

In this appendix, we employ the SP and ISP tests to assess the strong exogeneity of various proxies employed in a wide range of existing proxy SVAR applications. Table~\ref{tab:exploratory-applications} documents the results of these tests, along with the corresponding $F$-statistics, estimated skewness and excess kurtosis of the target shock, which help gauge the power properties of the tests.

We apply the tests to a variety of monetary policy applications for each individual proxy. 
First, we consider the UK monetary policy SVAR from \cite{braun2025measuring}, using their two proxies for the target and path shocks.
Second, we consider the US monetary policy SVAR from \cite{jarocinski2020deconstructing} with eight different proxies. FF4 is the high-frequency change in three-month-ahead federal funds futures around FOMC announcements. CBI is the sign-decomposed central-bank-information component from \cite{jarocinski2020deconstructing}, while MP is the sign-decomposed monetary policy component from \cite{jarocinski2020deconstructing}. MPI is the information-adjusted high-frequency monetary policy instrument from \cite{miranda2021transmission}.
CMP is based on the conventional monetary policy shocks from \cite{georgiadis2025global} OFG is based on the Odyssean forward-guidance shocks from \cite{georgiadis2025global}, LSAP is based on the large-scale asset-purchase shocks from \cite{georgiadis2025global}, and DFG is based on the Delphic forward-guidance shocks from \cite{georgiadis2025global}. 
We find that most monetary policy proxies display small $F$-values, indicating only weak correlation between the proxy and the target shock. Moreover, most identified monetary policy shocks are left-skewed and exhibit positive excess kurtosis. We cannot reject the null of strong exogeneity for any of the monetary policy proxies. However, particularly in applications with low proxy relevance, failure to reject the null of proxy exogeneity may reflect low power due to weak proxy strength.

In addition, we analyze two proxies from the oil-market literature. First, we replicate \cite{kanzig2021macroeconomic} using the oil-supply news proxy. Second, we replicate the three-variable oil-market SVAR used in \cite{MontielOleaStockWatson2021}, using the \cite{kilian2008exogenous} oil supply shock measure as an oil supply proxy.
We find that both oil-market proxies have $F$-values above $10$ and are associated with left-skewed shocks with positive excess kurtosis. In this setting, the strong exogeneity tests are expected to have power, yet neither test rejects the null of proxy exogeneity.

Moreover, we analyze three proxies in the fiscal policy SVAR from \cite{mertens2014reconciliation}.
The spending proxy uses the growth rate of military spending as a proxy for government spending shocks, similar to \cite{keweloh2025estimating}. The tax proxy uses the narrative series of exogenous tax changes from \cite{mertens2014reconciliation} as a proxy for tax shocks, while the TFP proxy uses the TFP series from \cite{fernald2014quarterly}, as used by \cite{caldara2017analytics}, as a non-fiscal proxy.
The fiscal proxy tests are similar to those considered in \cite{BK2024}, however, we use $L=3$ instead of only $L=2$. Similar to \cite{BK2024}, we find no evidence against exogeneity for the spending and TFP proxies, while the ISP test rejects the null of strong exogeneity for the tax proxy, which is in line with the results in \cite{keweloh2025estimating}.

Finally, we consider four additional proxy applications. We replicate \cite{kanzig2023unequal} to analyze the exogeneity of the carbon policy surprise proxy, \cite{piffer2018identifying} to analyze the exogeneity of the uncertainty proxy based on the gold price, \cite{miranda2026patents} to test the exogeneity of the patent technology news proxy, and \cite{lagerborg2023sentimental} to evaluate the exogeneity of the fatalities proxy for sentiment shocks.
We find no evidence against strong exogeneity for the carbon, technology, and sentiment proxies, although the small $F$-value of the technology proxy suggests that the tests may have little power in this application.
For the uncertainty proxy based on the gold price, however, the ISP test provides evidence against strong exogeneity. This is in line with the results of \cite{piffer2018identifying}, who presume that their gold proxy is simultaneously affected by uncertainty and news shocks. The ISP test supports this presumption.

\begin{sidewaystable}[p]
\centering
\caption{Overview of strong exogeneity tests in different applications}
\label{tab:exploratory-applications}
\footnotesize
\setlength{\tabcolsep}{3pt}
\renewcommand{\arraystretch}{1.08}

\begin{tabular}{@{}
    >{\raggedright\arraybackslash}p{5cm}
    >{\raggedright\arraybackslash}p{2cm}
    >{\centering\arraybackslash}p{2.5cm}
    >{\centering\arraybackslash}p{2.2cm}
    >{\centering\arraybackslash}p{3.2cm}
    >{\centering\arraybackslash}p{2.2cm}
    >{\centering\arraybackslash}p{2.2cm}
    @{}}
\toprule
Application & Proxy & $F$ ($p$) & Robust $F$ ($p$) & \shortstack{Target-shock skewness /\\ excess kurtosis} & \shortstack{ISP test:\\ $J$ ($p$)} & \shortstack{SP test:\\ $J$ ($p$)} \\
\midrule
\cite{braun2025measuring} & Target & $38.61$ ($0.000$) & $26.32$ ($0.000$) & $-0.10 / 1.21$ & $10.54\;(0.569)$ & $14.85\;(0.249)$ \\
\cite{braun2025measuring} & Path & $5.30$ ($0.022$) & $3.30$ ($0.070$) & $-0.16 / -0.01$ & $8.69\;(0.729)$ & $14.31\;(0.282)$ \\
\cite{jarocinski2020deconstructing} & FF4 & $13.15$ ($0.000$) & $7.64$ ($0.006$) & $-0.06 / 0.81$ & $5.22\;(0.734)$ & $9.51\;(0.301)$ \\
\cite{jarocinski2020deconstructing} & CBI & $1.12$ ($0.290$) & $0.24$ ($0.623$) & $-0.30 / 1.01$ & $7.54\;(0.480)$ & $4.64\;(0.795)$ \\
\cite{jarocinski2020deconstructing} & MP & $5.03$ ($0.025$) & $3.32$ ($0.069$) & $0.31 / 2.96$ & $9.22\;(0.324)$ & $10.51\;(0.231)$ \\
\cite{jarocinski2020deconstructing} & MPI & $10.87$ ($0.001$) & $5.13$ ($0.024$) & $-0.18 / 0.09$ & $1.71\;(0.989)$ & $11.17\;(0.192)$ \\
\cite{jarocinski2020deconstructing}  & CMP  & $6.62$ ($0.011$) & $4.00$ ($0.047$) & $0.14 / 2.07$ & $6.29\;(0.615)$ & $13.26\;(0.103)$ \\
\cite{jarocinski2020deconstructing}  & OFG & $9.60$ ($0.002$) & $5.49$ ($0.020$) & $-0.13 / 1.05$ & $12.05\;(0.149)$ & $4.63\;(0.796)$ \\
\cite{jarocinski2020deconstructing}  & LSAP  & $1.71$ ($0.192$) & $1.73$ ($0.189$) & $-0.16 / 0.11$ & $4.40\;(0.819)$ & $6.24\;(0.620)$ \\
\cite{jarocinski2020deconstructing}  & DFG  & $28.62$ ($0.000$) & $16.37$ ($0.000$) & $-0.51 / 2.74$ & $5.50\;(0.703)$ & $8.95\;(0.347)$ \\
\cite{kanzig2021macroeconomic} & Supply & $22.67$ ($0.000$) & $10.55$ ($0.001$) & $-0.17 / 1.20$ & $8.19\;(0.611)$ & $9.56\;(0.480)$ \\
\cite{MontielOleaStockWatson2021} & Supply & $15.97$ ($0.000$) & $8.96$ ($0.003$) & $-0.65 / 2.75$ & $2.73\;(0.604)$ & $3.68\;(0.451)$ \\
\cite{mertens2014reconciliation} & Spending & $131.49$ ($0.000$) & $69.80$ ($0.000$) & $0.07 / 1.00$ & $4.85\;(0.303)$ & $5.32\;(0.256)$ \\
\cite{mertens2014reconciliation} & Tax  & $4.22$ ($0.041$) & $1.62$ ($0.204$) & $-0.25 / 2.03$ & $\mathbf{10.12\;(0.039)}$ & $7.69\;(0.103)$ \\
\cite{mertens2014reconciliation} & TFP & $57.62$ ($0.000$) & $39.40$ ($0.000$) & $-0.13 / 1.51$ & $6.86\;(0.143)$ & $0.66\;(0.957)$ \\
\cite{kanzig2023unequal} & CPS & $5.95$ ($0.015$) & $16.85$ ($0.000$) & $-0.27 / 0.59$ & $6.70\;(0.946)$ & $6.82\;(0.941)$ \\
\cite{piffer2018identifying} & Gold & $19.38$ ($0.000$) & $2.99$ ($0.085$) & $1.07 / 4.39$ & $\mathbf{42.52\;(0.000)}$ & $11.88\;(0.616)$ \\
\cite{miranda2026patents} & Technology & $0.20$ ($0.659$) & $0.56$ ($0.455$) & $-0.52 / 3.11$ & $2.01\;(0.981)$ & $4.78\;(0.781)$ \\
\cite{lagerborg2023sentimental} & Fatalities  & $11.63$ ($0.001$) & $18.04$ ($0.000$) & $-0.32 / 1.22$ & $8.11\;(0.777)$ & $15.60\;(0.210)$ \\

\bottomrule
\end{tabular}

\vspace{0.5em}
\begin{minipage}{0.95\textheight}
\footnotesize
\textit{Notes:} All reduced form VARs are estimated by OLS, and all exogeneity tests use $L=3$.
\end{minipage}
\end{sidewaystable}

\end{document}